\documentclass[11pt]{article}

\usepackage{algorithm}
\usepackage{algorithmic}
\usepackage{epsfig}
\usepackage{subfigure}
\usepackage{enumerate}
\usepackage{amsmath}

\def\qed{\hbox{\rlap{$\sqcap$}$\sqcup$}}
\newtheorem{theorem}{Theorem}[section]
\newtheorem{lemma}{Lemma}[section]
\newtheorem{cor}{Corollary}[theorem]
\newtheorem{definition}{Definition}[section]
\newtheorem{property}{Property}[section]
\newenvironment{proof}{\par\noindent{\bf Proof:}}{\mbox{}\hfill$\qed$\\}
\newcounter{rem}
\newcommand{\ignore}[1]{ }

\def\cA{{\cal A}}
\def\cW{{\cal W}}
\def\cP{{\cal P}}
\def\cT{{\cal T}}

\begin{document}

\title{A near optimal algorithm for finding Euclidean shortest path in polygonal domain}
\author{Rajasekhar Inkulu \thanks{Department of Computer Science, Indian Institute of Technology, Guwahati, India.  E-mail: rinkulu@iitg.ac.in} \and Sanjiv Kapoor \thanks{Department of Computer Science, Illinois Institute of Technology, Chicago, USA.  E-mail: kapoor@iit.edu} \and S. N. Maheshwari \thanks{Department of Computer Science, Indian Institute of Technology, Delhi, India.  E-mail: snm@cse.iitd.ernet.in}}
\date{}
\maketitle

\begin{abstract}
We present an algorithm to find an {\it Euclidean Shortest Path} from a source vertex $s$ to a sink vertex $t$ in the presence of obstacles in $\Re^2$.
Our algorithm takes $O(T+m(\lg{m})(\lg{n}))$ time and $O(n)$ space.
Here, $O(T)$ is the time to triangulate the polygonal region, $m$ is the number of obstacles, and $n$ is the number of vertices.
This bound is close to the known lower bound of $O(n+m\lg{m})$ time and $O(n)$ space.
Our approach involve progressing shortest path wavefront as in continuous Dijkstra-type method, and confining its expansion to regions of interest.
\end{abstract}

\section{Introduction}
\label{sect:intro}


The shortest path problem in $\Re^d$ is that of finding a shortest route from one point to another among the presence of obstacles.
Even in $\Re^3$ under the Euclidean metric, it is not even known whether the shortest path problem in the presence of polyhedral obstacles is in NP though the problem has been shown to be NP-hard.  
This paper considers the case in $\Re^2$.
The Euclidean shortest path problem in a polygonal region is one of the oldest and best-known in computational geometry due to its various applications.  
Mitchell \cite{Mitchell00} provides an extensive survey of research accomplished in determining shortest paths in polygonal and polyhedral domains.
\hfil\break

We assume that the domain is defined by a simple polygon having $m$ obstacles comprising a total of $n$ vertices.
There are two fundamentally different approaches in solving this problem: the visibility graph method, and the wavefront method.  
These approaches assume a triangulation of the domain, which can be accomplished in $O(n+m(\lg{m})^{1+\epsilon})$ using the algorithm from Bar-Yehuda and Chazelle \cite{Chaz94}. 
\hfil\break

The visibility graph method is based on constructing a graph whose nodes are the vertices of the obstacles and the edges are pairs of mutually visible vertices.  
Welzl \cite{Welzl85} provides an algorithm for constructing the visibility graph with $n$ line segments in $O(n^2)$ time.
Ghosh and Mount \cite{Ghosh91}, and Kapoor and Maheshwari \cite{Kapoor00} found an algorithm to construct the visibility graph of time complexity $O(n \lg{n}+E)$, where $E$ is the number of edges in the graph.  
Applying Dijkstra-type algorithm on this graph, one can determine a shortest path in $O(n \lg{n}+E)$.  
Unfortunately the visibility graph can have $\Omega(n^2)$ edges in the worst case, so any shortest path algorithm that depends on an explicit construction of the visibility graph will have a similar worst-case running-time.  
\hfil\break

Storer and Reif \cite{Reif94} presented $O(T+mn)$ time algorithm which constructs a data structure so that the shortest path from $s$ to any point on the plane can be determined in $O(1)$ time.
Using the concept of corridors, Kapoor and Maheshwari \cite{Kapoor88} presented an algorithm of time complexity $O(m^2 \lg{n}+n \lg{n})$.
\hfil\break

The second approach used by Hershberger and Suri \cite{Hersh93}, Mitchell \cite{Mitchell88}, Mitchell \cite{Mitchell93}, and Kapoor \cite{Kapoor98, Kapoor99} gave algorithms to find a shortest path by expanding a wavefront from source $s$ till it reaches the destination $t$.
This approach seems inherently more geometric than the graph-theoretic method based on visibility graphs.
However, this method when directly applied does not achieve the known lower bound of $\Omega(n + m \log{m})$, and resolving this is an open problem for several years.  
Mitchell \cite{Mitchell93} gave an algorithm for computing a shortest path map, an encoding of shortest paths from $s$ to all points of the plane in $O(n^{3/2+\epsilon})$ time and space.
More recently, Hershberger and Suri \cite{Hersh97} presented $O(n \lg{n})$ time algorithm.
\hfil\break

In this paper, we combine corridors with the wavefront approach to obtain a $O(n+m(\lg{m})(\lg{n}))$ time algorithm which uses $O(n)$ space for computing Euclidean shortest path among obstacles in $\Re^2$, whereas the Problem 21 of The Open Problems Project (TOPP) intends for a solution with $O(n+m\lg{m})$ time using $O(n)$ space.
We assume a model of computation where real arithmetic is allowed, though the results apply even when finite precision arithmetic is used (ignoring the numerical complexity of the schemes). 
Our algorithm proceeds by first triangulating the given polygonal region and then identifying the useful corridors and junctions among those triangles as in Kapoor et al. \cite{Kapoor88}.  
Then we initiate a shortest path wavefront from source and progress it as in continuous Dijkstra-type of method; however, to reduce the number of event points, we confine the wavefront to progress in regions of interest.
\hfil\break

Section \ref{sect:defsandprops} gives basic definitions, properties, and the utility of various constructs that we use in developing the algorithm.
Algorithm outline is mentioned in Section \ref{sect:algooutline}.
Section \ref{sect:datastr} gives details of data structures and the operations on each of them.
More technical details of algorithm are presented in Sections \ref{sect:sdcomp}, \ref{sect:IIntersectProc}, \ref{sect:merging}, \ref{sect:boundarysplit}, and \ref{sect:dirtybridges}.
The algorithm in terms of event point types and their handling is described in Section \ref{sect:evtpts}. 
Both the analysis and proof of correctness are spread all through the paper, whereas the analysis required for the overall time and space complexity analysis is presented in Section \ref{sect:moreanalysis}.
Section \ref{sect:conclu} concludes with possible generalizations.

\section{Definitions and Properties}
\label{sect:defsandprops}

\indent Let $v_{l}\hspace{-0.05in}=\hspace{-0.05in}v_0, v_1, \ldots, v_{l-1}$ be $l$ points, known as vertices, in the plane.
The sequence of $l$ line segments, known as edges, $e_l\hspace{-0.05in} =\hspace{-0.05in} e_0\hspace{-0.05in}=\hspace{-0.05in}v_0v_1, e_1\hspace{-0.05in}=\hspace{-0.05in}v_1v_2, \ldots, e_{l-1}\hspace{-0.05in}=\hspace{-0.05in}v_{l-1}v_0$ together form a {\it closed polygonal chain}, say $C$.
The polygonal chain $C$ is {\it simple} if and only if $\forall_i e_i \cap e_{i+1} = v_{i+1}$ and $\forall_{i, j \ne i+1} e_i \cap e_j = \phi$.
Then $C$ with the region bounded by it together is known as a {\it simple polygon}, say $P'$.
Let ${\cal O} = \{ P_1, P_2, \ldots, P_m \}$ be the set of simple polygons interior to $C$ s.t. $\forall_{i, j \ne i} P_i \cap P_j = \phi$. 
Then $P-\bigcup_i P_i$ is known as the {\it polygonal domain} or {\it polygon with holes}, say ${\cal P}$.
The set $V$ comprising the vertices of $\cP$ is of size $n$, whereas the number of obstacles $|{\cal O}|$ is $m$.
We denote the Euclidean shortest distance between two points $p_1$ and $p_2$ in {\cP} with $d(p_1, p_2)$.
We intend to find an Euclidean shortest path between two points in $\cP$.
Among these two points, one is termed as source $s$ and the other is sink $t$.
We consider both $s$ and $t$ as degenerate single point obstacles.
In other words, $s, t \in {\cal O}$.
We use the continuous Dijkstra's approach in finding a shortest path from $s$ to $t$.

\paragraph{Wavefront Progression with Triangulation}

\begin{definition}
The {\bf shortest path wavefront} $\cW (d)$ at distance $d$ is the locus of points at Euclidean shortest distance $d$ from the source $s$.
\end{definition}

Initially, the wavefront is a circle centered at $s$ with radius $\epsilon$, where $\epsilon$ is a small positive constant. 
The algorithm proceeds by expanding the wavefront in ${\cP}$. 
As the wavefront progresses, it may encounter various vertices and edges of $\cP$.
Let the shortest path wavefront be at distance $d'$ from $s$.
A point $p \in \cP$ is considered as {\bf traversed} if $d(p, s) \leq d'$.
Otherwise, the point $p$ is said to be untraversed.
An edge $e$ is traversed, if there exists a point $p \in e$ such that $p$ is traversed.
A region $R$ is traversed, if for every point $p \in R$, $p$ is traversed. 
An edge $e$ is defined as {\bf struck} if there exists a point $p \in e$ such that $d(p, s) = d'$ and for $p \ne p'$ and $p' \in e$, $d(p', s) \ge d'$. 
For any vertex $v$ of ${\cP}$, when the wavefront strikes $v$, a new arc with center $v$ may be initiated and inserted into the wavefront. 

\begin{property}
A {\bf wavefront segment} $w(v)$ is a circular arc with center $v \in \cP$, such that each point on $w(v)$ has a shortest path to source $s$ via $v$.
\end{property}

\begin{property}
At any stage of the algorithm, the wavefront $\cW(d)$ comprises a contiguous (abutting) sequence of wavefront segments, $w(v_1), \ldots, w(v_k)$ for some $k \le n$.
\end{property}

The algorithm halts when the wavefront strikes $t$.
Suppose $w(v_k)$ is the arc that struck $t$, $w(v_{k-1})$ is the arc that struck $v_k$, \ldots, $w(s)$ is the arc that struck $v_1$. 
Then our algorithm outputs the shortest path from $s$ to $t$ which comprises of line segments $sv_1, v_1v_2, \ldots, v_{k-1}v_k$ and $v_kt$.
Adding the Euclidean distances along these line segments yields the shortest distance from $s$ to $t$.
\hfil\break

\begin{lemma}
\label{lem:spnoncrossing}
Consider shortest paths from source $s$ to two points $p_i$ and $p_j$ in $\cP$.
Suppose the interior of a line segment in shortest path from $s$ to $p_i$ intersects with the interior of a line segment in shortest path from $s$ to $p_j$.
Then there always exists a shortest path from source $s$ to $p_i$ (resp. $s$ to $p_j$) so that the interior of no line segment in shortest path from $s$ to $p_i$ (resp. $s$ to $p_j$) intersects with the interior of line segments in the given shortest path from $s$ to $p_j$ (resp. $s$ to $p_j$).
This property is termed as {\bf non-crossing property of shortest paths}.
\end{lemma}
\begin{proof}

\begin{figure}
\center{
\subfigure[Before applying the property]{\epsfysize=170pt \epsfbox{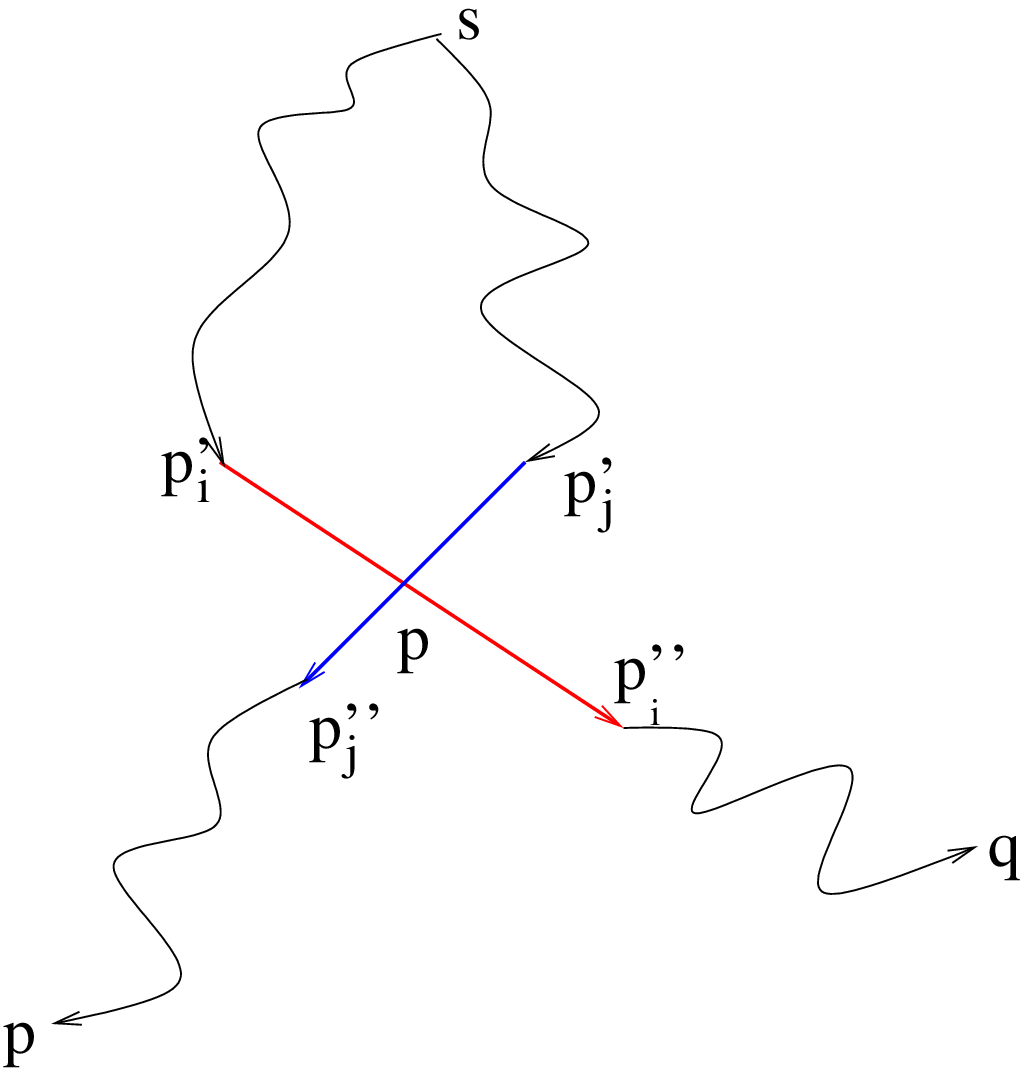}}
\subfigure[After applying the property]{\epsfysize=170pt \epsfbox{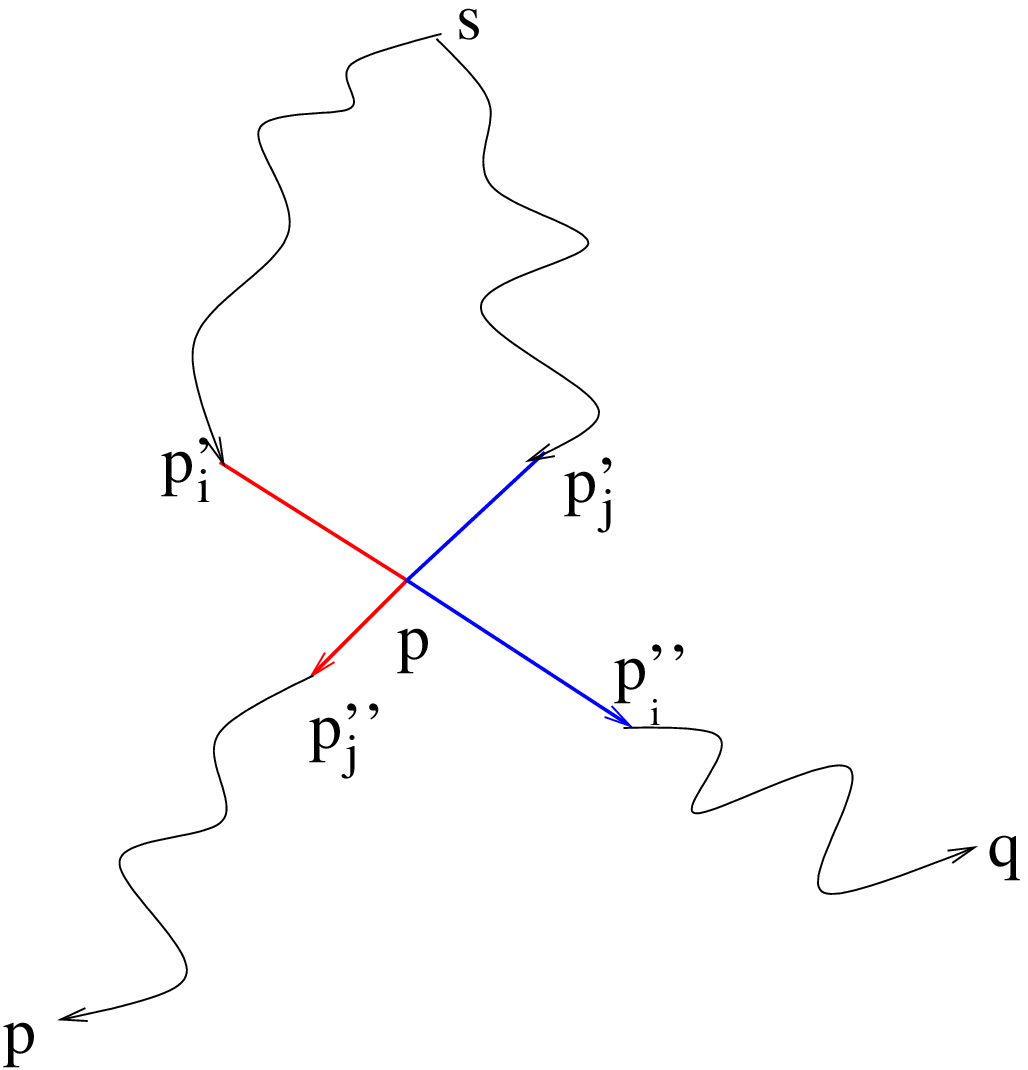}}
}
\caption{\label{fig:noncontigprop} Non-contiguity property of shortest paths}
\end{figure}

Consider the point $p$ at which the interior of line segments $p_i'p_i''$ and $p_j'p_j''$ respectively belonging to shortest paths from $s$ to $p_i$ and $p_j$ intersect.
First, note that there exists at least two shortest paths from $s$ to $p$: one via $p_i'$ and the other via $p_j'$. 
Replacing the given shortest path to $p_i''$ via $p_i'$ with $p_i''p$ and $pp_j'$, and replacing the given shortest path to $p_j''$ via $p_j'$ with $p_j''p$ and $pp_i'$ yields two shortest paths that do not cross with each other at $p$.
Since this local operation reduces the number of interior intersection of line segments in both the shortest paths by one, repeatedly applying this operation at every such intersection point yields the required. 
\end{proof}

To guide the wavefront progression, we triangulate the polygonal domain using the algorithm by Bar-Yehuda and Chazelle \cite{Chaz94}.
First we triangulate $\cP - \{s, t\}$, and obtain a triangulation.
We locate $s$ in a triangle $T=(v_1, v_2, v_3)$ and introduce triangulation edges $sv_1, sv_2, sv_3)$; similarly, after locating $t$, we introduce triangulation edges to yield a resultant triangulation, denoted with $\cT$.
\hfil\break

A sequence of wavefront segments in the wavefront, not necessarily contiguous, together are termed as a {\it section of wavefront}.
We define a line segment $l$ to be {\bf reachable} by a section of wavefront, if there exists a point $p$ on $l$ and a point $w$ on $W$ such that the interior of line segment $pw$ does not intersect any untraversed vertex/edge in $\cP$.
Typically, we are interested in reachable edges of triangulation $\cT$.
The wavefront is progressed based on its interaction with the reachable edges, and the interaction among the wavefront segments.
\hfil\break

The algorithm is event based.
Let $\cW (d)$ be the wavefront.
Primarily, the event points can be categorized into two: finding $d' \ge d$ so that the $\cW (d')$ strikes a reachable triangulation edge; finding $d'' \ge d$ so that the $\cW (d'')$ at which two wavefront segments intersect.
The events occur as the wavefront progresses and are maintained in a min-heap, with the corresponding shortest distance at which the event occurs ($d'$ or $d''$) as the key. 
The former causes the updates to the set of reachable edges, and the new segments may possibly be incorporated into the wavefront.
The latter could cause some wavefront segments to change shape, some to be removed from the wavefront, and may change the set of reachable edges.
\hfil\break

With this approach, there are $O(n)$ vertices from each of which a wavefront segment could be initiated; these $O(n)$ wavefront segments could interact with $O(n)$ triangulation edges, causing quadratic time complexity in terms of number of vertices in $\cP$. 

\paragraph{Corridors and Junctions} \hfil\break

We intend to reduce the number of edges with which the wavefront interacts.
This we accomplish by exploiting the structure in the triangulation and obtain a coarser data structure.
The number of elements in the new data structure are $O(m)$ and the wavefront may strike $O(m)$ entities rather than $O(n)$, hence is an improvement.
\hfil\break

\begin{figure}
\centerline{\epsfxsize=350pt \epsfbox{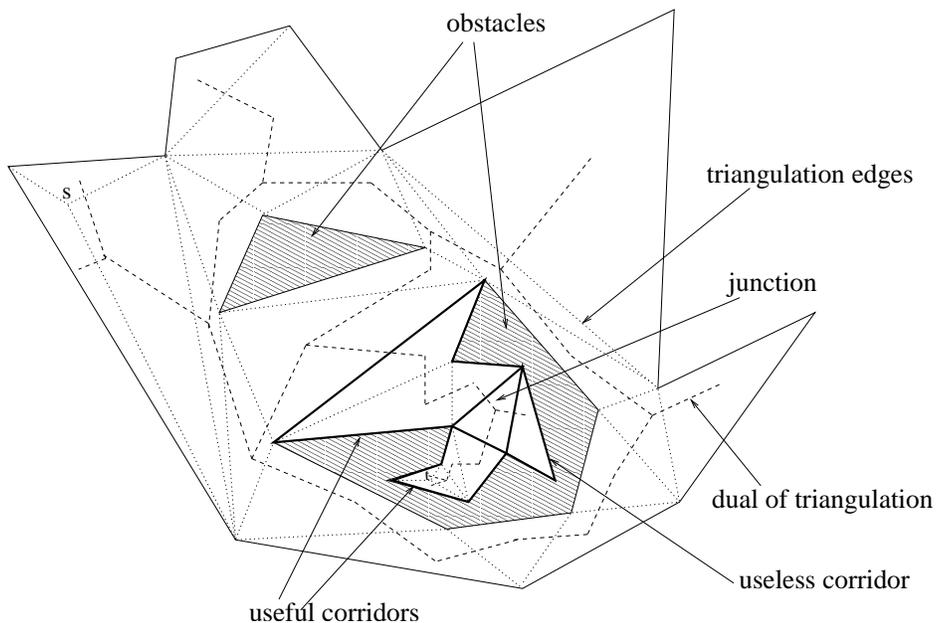}}
\caption{\label{fig:junctcorri}{Corridors and Junctions}}
\end{figure}

The coarser structural elements obtained from triangulation are termed as corridors and junctions.
We outline their descriptions from Kapoor et al. \cite{Kapoor97b}.
Consider the triangulation $\cT$ and its dual graph.
See Fig. \ref{fig:junctcorri}.
The dual graph is divided into paths where a path is composed of maximally connected vertices of degree two.  
Each such path defines a {\it corridor} formed by the sequence of dual regions or triangles corresponding to the vertices of the dual graph.  
A corridor is a region confined by (at most) four geometric entities - two convex chains on opposite sides, termed as {\em corridor convex chains}; and two edges, termed as {\it (wavefront) enter/exit bounding edges or enter/exit boundaries}, that are incident to both the corridor convex chains.
Each of these convex chains' is a section of boundary of an upper or lower hull (see \cite{Prep85}).
A junction is a triangle enclosed by at most three edges, each edge from a distinct corridor.
A side of an edge belonging to a useful corridor convex chain or enter/exit bounding edge is known as a {\it bounding edge}.
An end point of a bounding edge is known as a {\it bounding vertex}.
An edge $e$ in the dual graph is defined to be {\it useful} if there exists a simple path from $s$ to $t$ that is having a point in common with the triangles associated with dual edge $e$. 
Otherwise, it is {\it useless}.
A corridor is defined to be {\it useful} if the triangles constituting the corridor contribute an useful edge in the dual graph.  
While partitioning $\cP$ into corridors and junctions, we define $s$ and $t$ as degenerate corridors.
The algorithm starts with a subdivision of the polygonal region into $O(m)$ useful {\it corridors} and {\it junctions}.
\hfil\break

The corridors can be classified by their structure into two types, {\it open and closed}.
A corridor $C$ is termed as an {\it open corridor} whenever there exists two points such that $p_1$ lies on one enter/exit boundary of $C$ and $p_2$ lies on the other enter/exit boundary edge of $C$, and $p_1$ is visible to $p_2$.
See Fig. \ref{fig:evolveboundcycle}.
Otherwise, a corridor is termed as a {\it closed corridor}.
A closed corridor gives rise to two funnels each with an apex, and each funnel has two convex chains.
See Fig. \ref{fig:closedcorri}.
\hfil\break

The advantages in progressing shortest path wavefront using corridors is two fold. 
First, rather than interacting with the $O(n)$ triangulation edges, it interacts with $O(m)$ corridor convex chains and corridor enter/exit bounding edges.
Further, as explained latter, we exploit the coherence in the wavefront segments that are originated from the successive vertices along a corridor convex chain.
\hfil\break

A corridor convex chain $C$ is {\it reachable}, if there exist a bounding edge $b \in C$ that is reachable.
A contiguous cycle of reachable corridor convex chains and/or enter/exit boundaries is termed as a {\bf boundary cycle}.
\hfil\break

\begin{figure}
\centerline{
\subfigure[Boundary cycle $b_1, b_2, b_3$]{\epsfxsize=200pt \epsfbox{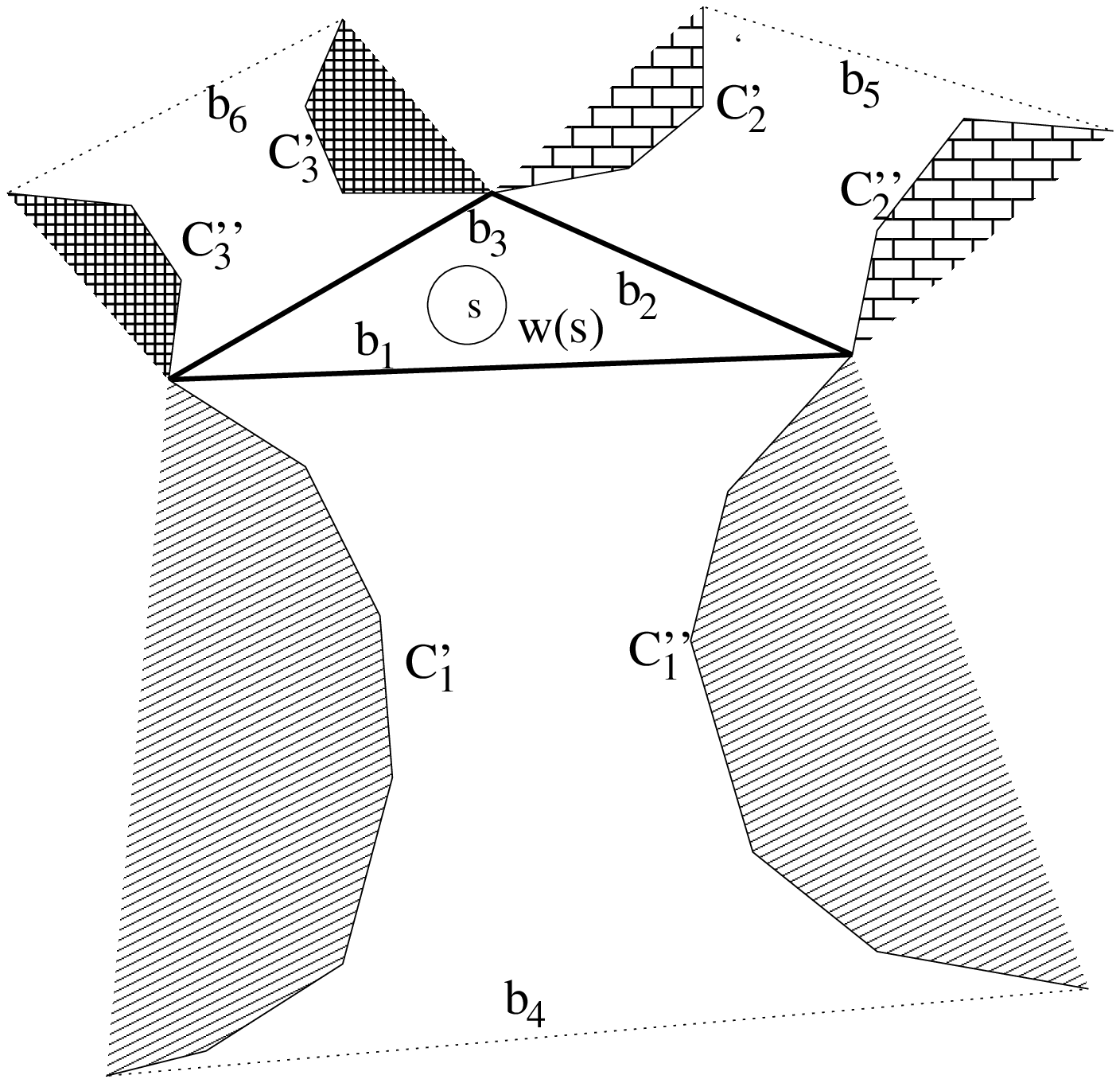}}
\subfigure[Boundary cycle $C_1', b_4, C_1'', b_2, b_3$]{\epsfxsize=200pt \epsfbox{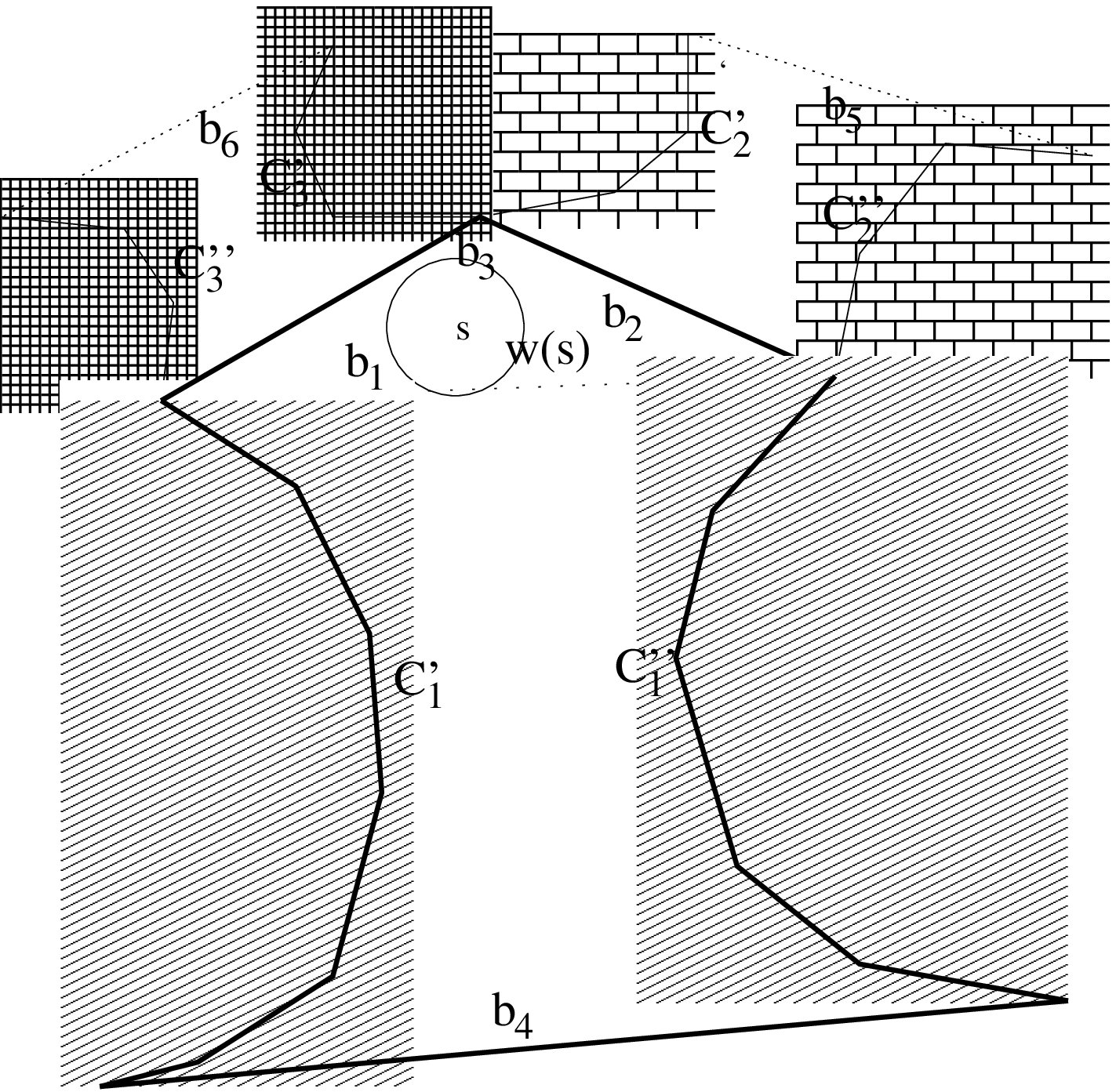}}
}
\centerline{
\subfigure[Boundary cycle $C_1', b_4, C_1'', b_2, C_3', b_6, C_3''$]{\epsfxsize=200pt \epsfbox{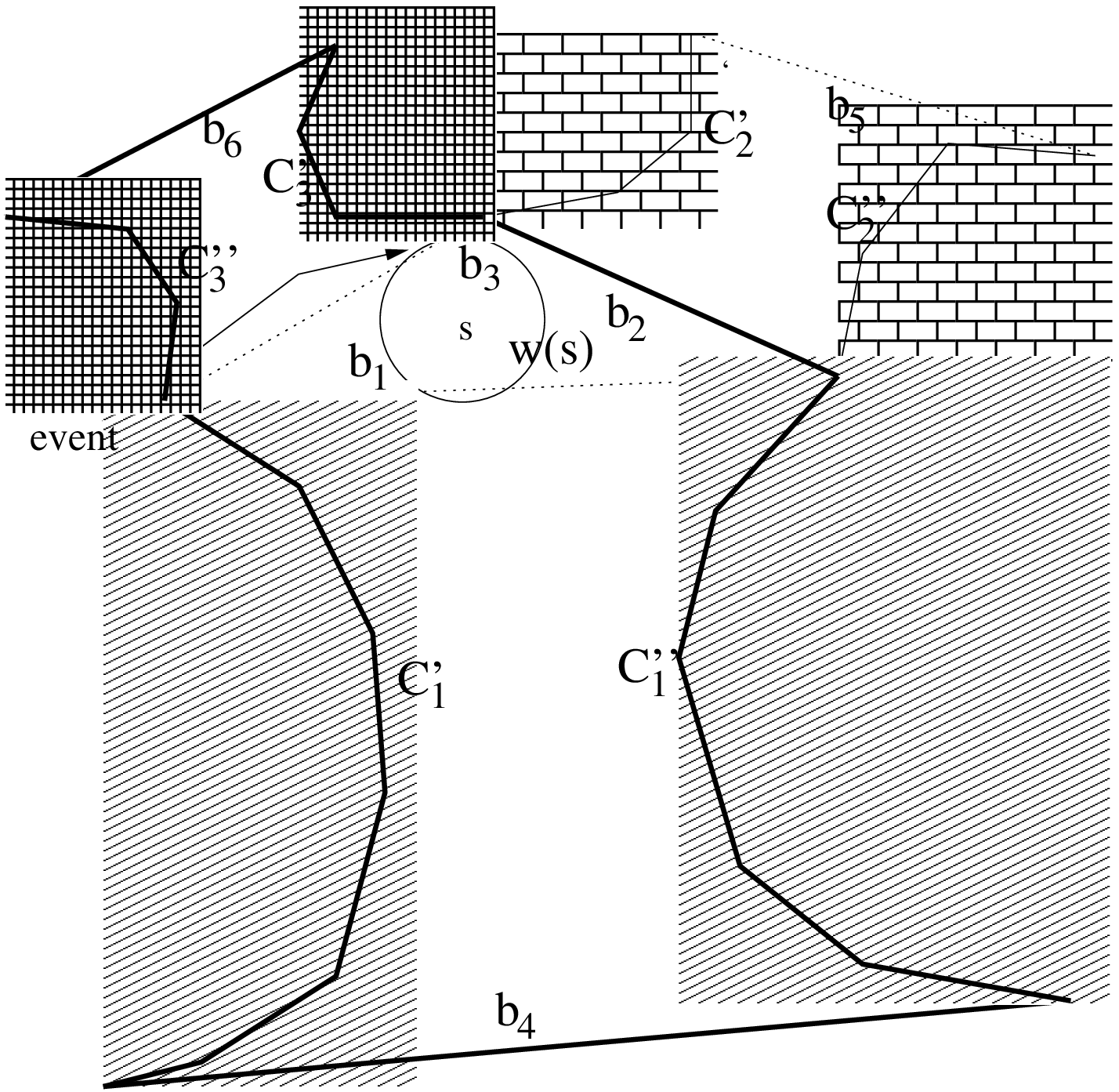}}
}
\caption{\label{fig:evolveboundcycle} Boundary cycle as the initial wavefront expands}
\end{figure}

As shown in Fig. \ref{fig:evolveboundcycle}(a), $e_1, e_2, e_3$ are the edges of a junction $J$ and the corresponding $b_1, b_2, b_3$ are the bounding edges that are reachable from the initial wavefront segment $w(s)$.
The initial wavefront comprises of $w(s)$ only.
The bounding edges $b_1, b_2, b_3$ together form the boundary cycle $BC$.
The first event that occurs when the wavefront, which is a circle, strikes $b_1$ of $J$.
See Fig. \ref{fig:evolveboundcycle}(b). 
At that stage, the bounding edge $b_1$ in $BC$ is replaced by convex chain $C_1'$, bounding edge $b_4$, and the convex chain $C_1''$.
The resultant boundary cycle is $C_1', b_4, C_1'', b_2, b_3$.
As the $w(s)$ progresses, the boundary cycle $BC$ further changes as shown in Fig. \ref{fig:evolveboundcycle}(c).
In general, as the wavefront expands, if the just struck edge $e$ bounds an untraversed junction $J$ then $e$ is replaced by the other two edges of $J$ in the boundary cycle under consideration.
And if the edge $e$ is a bounding edge of an untraversed corridor $C$, then $e$ is replaced by the other bounding edge of $C$ and corridor convex chains of $C$. 
\hfil\break

\begin{figure}
\centerline{\epsfysize=220pt \epsfbox{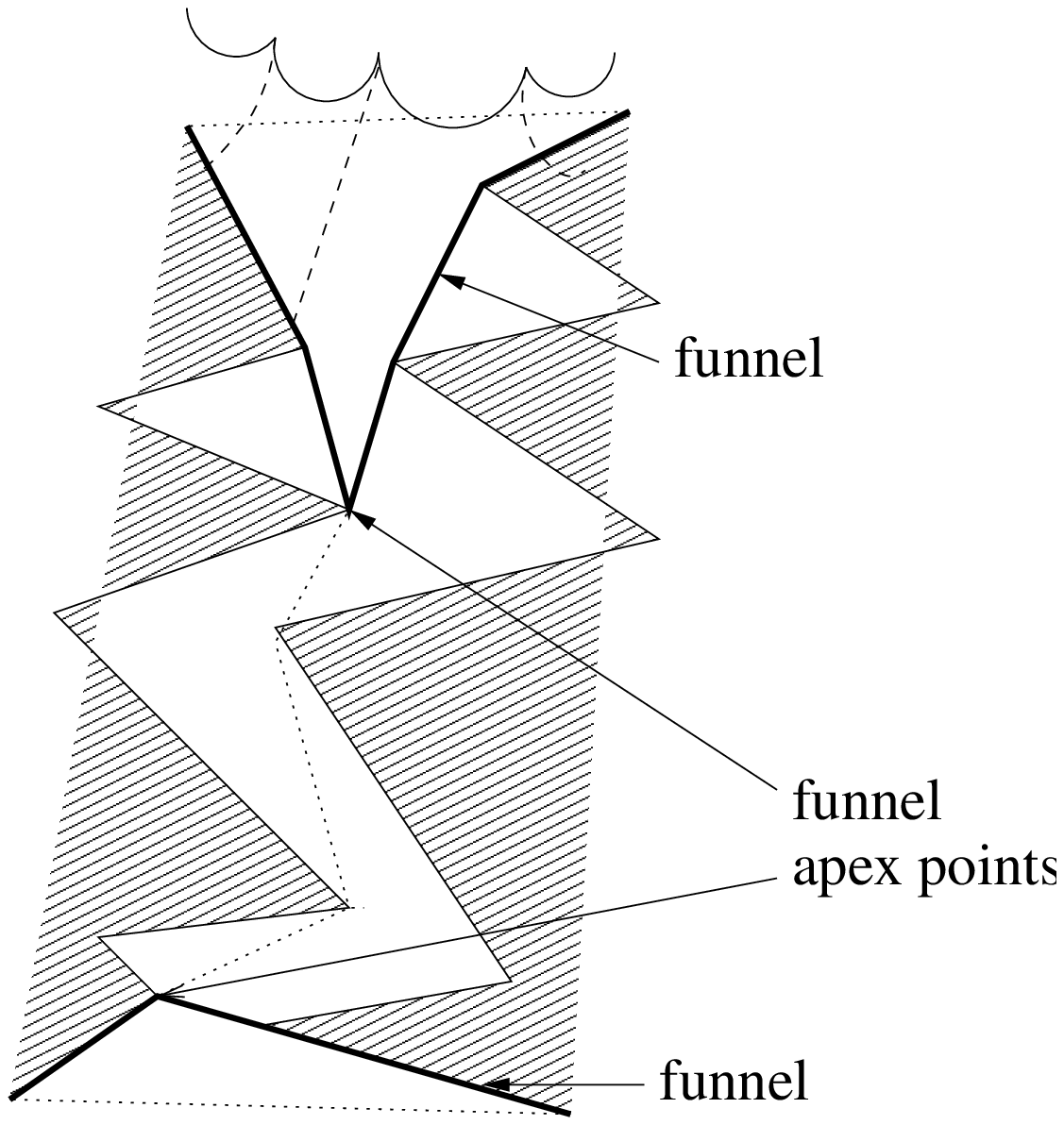}}
\caption{\label{fig:closedcorri}{Wavefront progression in closed corridors}}
\end{figure}

In the case of closed corridors, after the wavefront strikes the first apex point of the funnel, segments are initiated from the other apex point (provided the shortest distance to that has not already been determined) $p$, when the wavefront expands from $s$ after a distance that equals the shortest distance from $s$ to $p$.
See Fig. \ref{fig:closedcorri}.
\hfil\break

\begin{figure}
\center{
\subfigure[Junction $J$ before boundary split]{\epsfxsize=330pt \epsfbox{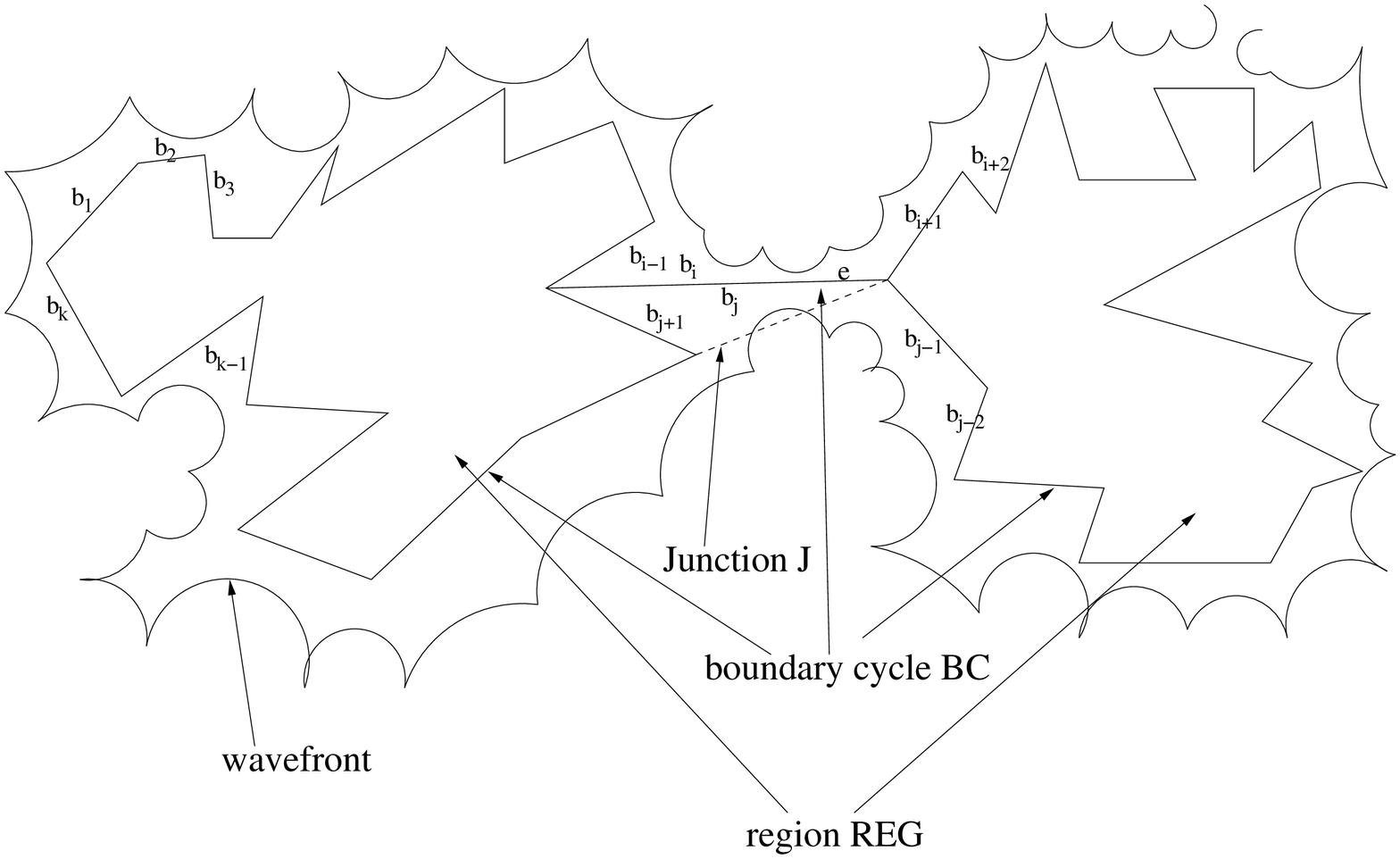}}
\subfigure[Junction $J$ after boundary split]{\epsfxsize=330pt \epsfbox{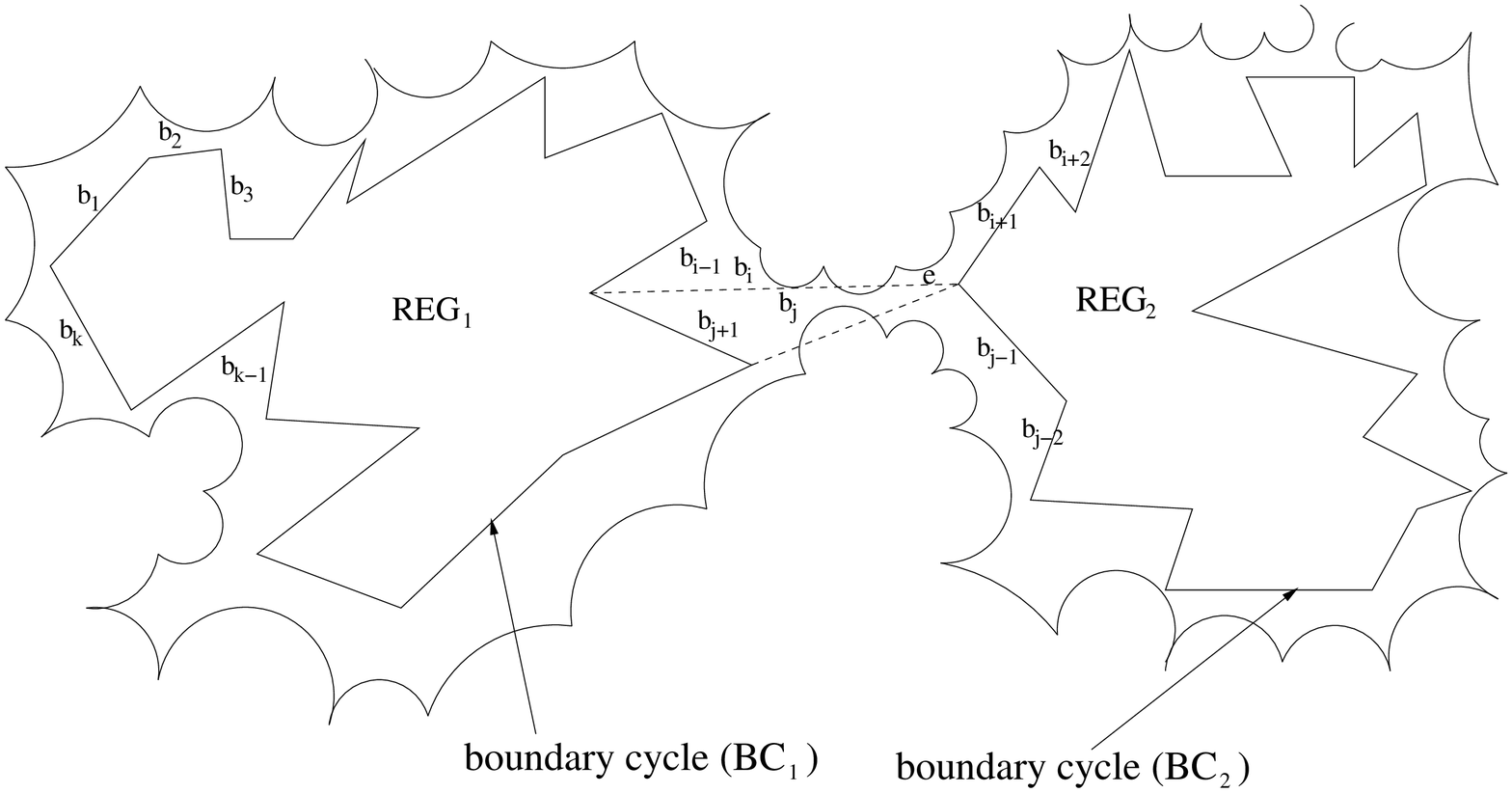}}
}
\caption{\label{fig:boundarysplit} Boundary split}
\end{figure}

As the wavefront progresses, a boundary cycle may split.
See Fig. \ref{fig:boundarysplit}.
This is possible when an edge $e$ in a boundary cycle $BC$ is reachable from the wavefront from both of its sides. 
In other words, $e$ appears twice in $BC$ as a bounding edge, and the boundary cycle splits into two when $e$ is struck from either side.
Consider a boundary cycle $BC = b_1 , b_2 \ldots b_i = (u_i , v_i) \ldots  b_j = (v_i, u_i) \ldots b_k, b_1$.
Note that both $b_i$ and $b_j$ represent the same edge, say $e$.
When $e$ is struck, the boundary cycle $BC$ splits into two boundary cycles:  $BC_i= b_1, \ldots, b_{i-1}, b_{j+1}, \ldots, b_k$ and $BC_j=b_{i+1}, \ldots, b_{j-1}$.
For any two boundary cycles $BC', BC''$, and bounding edges $b' \in BC', b'' \in BC''$, if  $b'$ and $b''$ do not correspond to the same edge, then we say that $BC', BC''$ are {\it disjoint boundary cycles}. 
Therefore, $BC_i$ and $BC_j$ are disjoint boundary cycles.
The corridor bounding edges in all boundary cycles together are denoted with $\partial B$.
A sequence of contiguous bounding edges in $\partial B$ is termed as a {\it section of boundary}.

\paragraph{Associations and I-curves} \hfil\break

Since there are $O(m)$ corridors, $O(m)$ junctions, and $O(n)$ wavefront segments, the interactions between the wavefront and the junction/corridor boundaries can be $O(nm)$.  
This has been perceived to be the bottleneck in implementing the wavefront method.
To this end, we keep track of corridor convex chains, corridor enter/exit boundaries that are reachable from sections of wavefront, with the wavefront progress.
\hfil\break
 
Consider a corridor convex chain or enter/exit boundary $g$ in a boundary cycle.
Let $W_g$ be the set of wavefront segments from each of which $g$ is reachable.
Every wavefront segment $w \in W_g$ is {\bf associated} with $g$ (or, $g$ is associated with segment $w$) if and only if a point on $g$ has shortest Euclidean distance to $s$ via the center of $w$. 
The  association is defined by the relation: ${\cal A} \subseteq {\cal G} \times {\cal S}$, where ${\cal G}$ is the set of corridor convex chains or enter/exit boundaries in the polygonal domain, and ${\cal S}$ is the set of waveform segments formed during the course of algorithm.
Note that the wavefront segments in $W_g$ need not be contiguous in wavefront.
\hfil\break

We maintain these associations for determining the wavefront progression at which the wavefront strikes $\partial B$.
For every section of wavefront $W$, given its association with a section of boundary $BS$ in a boundary cycle, we need only to compute the shortest distance between $BS$ and $W$, which is more efficient than computing the shortest distance between $W$ and whole of $\partial B$.
Again, we update these associations locally whenever an event changes either $BS$ or $W$. 
The following definition of I-curves helps in initiating and updating associations of wavefront segments or sections of wavefronts.
\hfil\break

The Voronoi diagram of a given set of points $S$ is the partition of the plane into regions so that all the points interior to a region are closer to one and only one point of $S$. 
To determine the associations for a wavefront segment, essentially we require Voronoi regions corresponding to each wavefront segment.
Consider any two adjacent wavefront segments, $w(v_a), w(v_b)$, in the wavefront.
The curve that separates the Voronoi regions of $v_a$ and $v_b$ is termed as an {\bf I-curve$(w(v_a), w(v_b))$}.
Every point $p$ on I-curve($w(v_a), w(v_b)$) has (at least) two shortest distance paths to $s$: one is via $v_a$ (the line segment from $p$ to $v_a$ together with a shortest path from $v_a$ to $s$) and the other is via $v_b$ (the line segment from $p$ to $v_b$ together with a shortest path from $v_b$ to $s$).  
In other words, I-curve($w(v_a), w(v_b)$) is a curve separating Voronoi regions belong to $v_a$ and $v_b$ in a weighted Voronoi diagram defined over $v_a$ and $v_b$, with the respective weights $d(v_a, s), d(v_b, s)$ at $v_a$ and $v_b$.
\hfil\break

Let $w(v_i), \ldots, w(v_j), w(v_{j+1}), \ldots, w(v_k), w(v_{k+1}), \ldots, w(v_l)$ be successive wavefront segments.
When an I-curve($w(v_j), w(v_{j+1})$) intersects another I-curve($w(v_k), w(v_{k+1})$) at a point $p$, due to non-crossing property of shortest paths (Lemma \ref{lem:spnoncrossing}), there is no need to progress wavefront segments $w(v_{j+1}), \ldots, w(v_k)$  any further from $p$.
To avoid overlaps of segments within the wavefront, we capture I-curve intersection as an event point and update the wavefront by removing $w(v_{j+1}), \ldots, w(v_k)$ from the wavefront.
Therefore, I-curves are also useful in identifying the wavefront changes.
\hfil\break

Given that there could be number of boundary cycles, the associations of wavefront segments with the boundary edges on every cycle is of interest in determining the interactions.
We show the following useful property:

\begin{lemma}
\label{lem:contigprop1} There exists an association {\cA} such that the sequence of boundary edges on a boundary cycle that are associated with a segment is a contiguous sequence.  
This is known as the {\it contiguity property for wavefront segments}.
\end{lemma}
\begin{proof}
Consider a sequence of bounding edges $e_i, \ldots, e_j, \ldots, e_k, \ldots, e_l$ of a boundary cycle.
Let two sections of boundary $e_i, \ldots, e_j$ and  $e_k, \ldots, e_l$ are associated with wavefront segment $w(v_a)$ and another section of boundary $e_j, \ldots, e_k$ is associated with a segment $w(v_b)$.
See Fig. \ref{fig:contigassoc}.
Also, let the shortest path from $v_b$ to $s$ does not pass through $v_a$, and vice versa.
\hfil\break

\begin{figure}
\centerline{\epsfysize=160pt \epsfbox{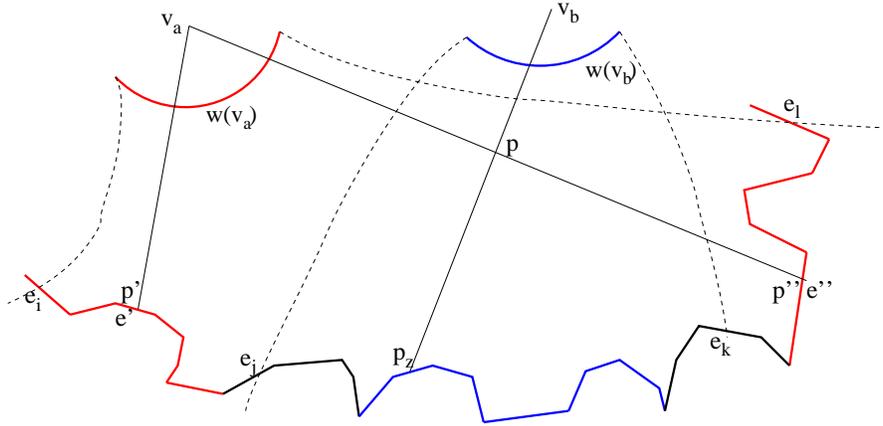}}
\caption{\label{fig:contigassoc}{Contiguity property for segments}}
\end{figure}

>From the definition of associations, for every bounding edge $e' \in \{e_i,\ldots, e_j\}$, there exists a point $p'$ on $e'$ which has a shortest path to $s$ via  $v_a$.  
Similarly, for every bounding edge $e'' \in \{e_k,\ldots,e_l\}$, there exists a point $p''$ on $e''$ which has a shortest path to $s$ via $v_a$; and for every bounding edge $e_z \in \{e_j,\ldots,e_k\}$, there exists a point $p_z$ on $e_z$ which  has a shortest path to $s$ via $v_b$.  
But the line segment joining $p_z$ to $v_b$, intersects either the line segment joining $p'$ to $v_a$ or the line segment joining $p''$ to $v_a$.
This is because the point $p_z$ occurs on the section of boundary between points $p'$ and $p''$.
Suppose the line segment $p_zv_b$ intersects the line segment $p''v_a$ at a point $p$  (argument for the other case is symmetric).
Because of the non-crossing property of shortest paths (Lemma \ref{lem:spnoncrossing}), the intersection point $p$ has at least two shortest paths to $s$: one via $v_a$, and, another via $v_b$.  
In turn, there exists two shortest paths from $p''$ (resp. $p_z$) to $s$: one via $v_a$, and the other via $v_b$.  
Hence $w(v_b)$ can be associated with $e''$, and, $w(v_a)$ can be associated with $e_z$.
\end{proof}

Fig. \ref{fig:assocsegcorri} shows the associations between a section of wavefront $W$ and a section of boundary $B$ with I-curves between adjacent segments of $W$.
\hfil\break

\begin{figure}
\centerline{\epsfxsize=290pt \epsfbox{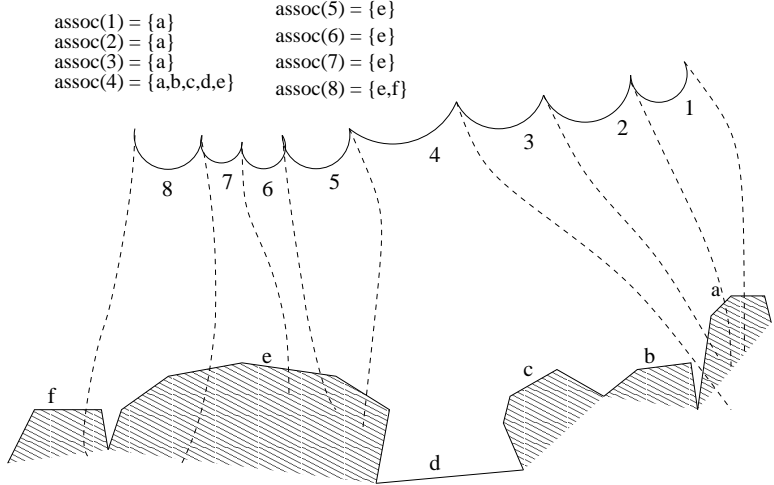}}
\caption{\label{fig:assocsegcorri} Associations of wavefront segments with convex chains and enter/exit boundaries}
\end{figure}

Any corridor enter/exit boundary $e$ is considered as traversed from the first time it got struck by a section of wavefront, and from there on $e$ does not participate in associations. 
In other words, each corridor enter/exit boundary is struck by the wavefront at most only once.
When a section of wavefront $W$ strikes an untraversed junction $J$ or corridor $C$, the updates to associations of $W$ are termed as {\it wavefront split}; whereas the {\it wavefront merger} procedure updates the associations of sections of wavefront when an edge of already traversed junction $J$ or corridor $C$ is struck.
\hfil\break

The following example describes both the wavefront split and merger.
Consider a junction $J=(e_1, e_2, e_3)$ in which no edge is traversed yet. 
Suppose a section of wavefront $SW_1$ struck edge $e_1$ when the wavefront is $\cW (d)$. 
Then we need to delete $e_1$ from all of its associations, and the associations of $SW_1$ needs to be updated to reflect the associations of wavefront segments in $SW_1$ to bounding edges corresponding to $e_2$ and $e_3$.
As shown in Fig. \ref{fig:wfsplitmerger}(a), the section of wavefront $SW_1$ is split into section of wavefront $SW_1'$ that is associated with $e_2$, and the section of wavefront $SW_1''$ that is associated with $e_3$, for $SW_1' \cup SW_1'' = SW_1$.
\hfil\break

\begin{figure}
\center{
\subfigure[$SW_1$ struck $e_1$]{\epsfxsize=200pt \epsfbox{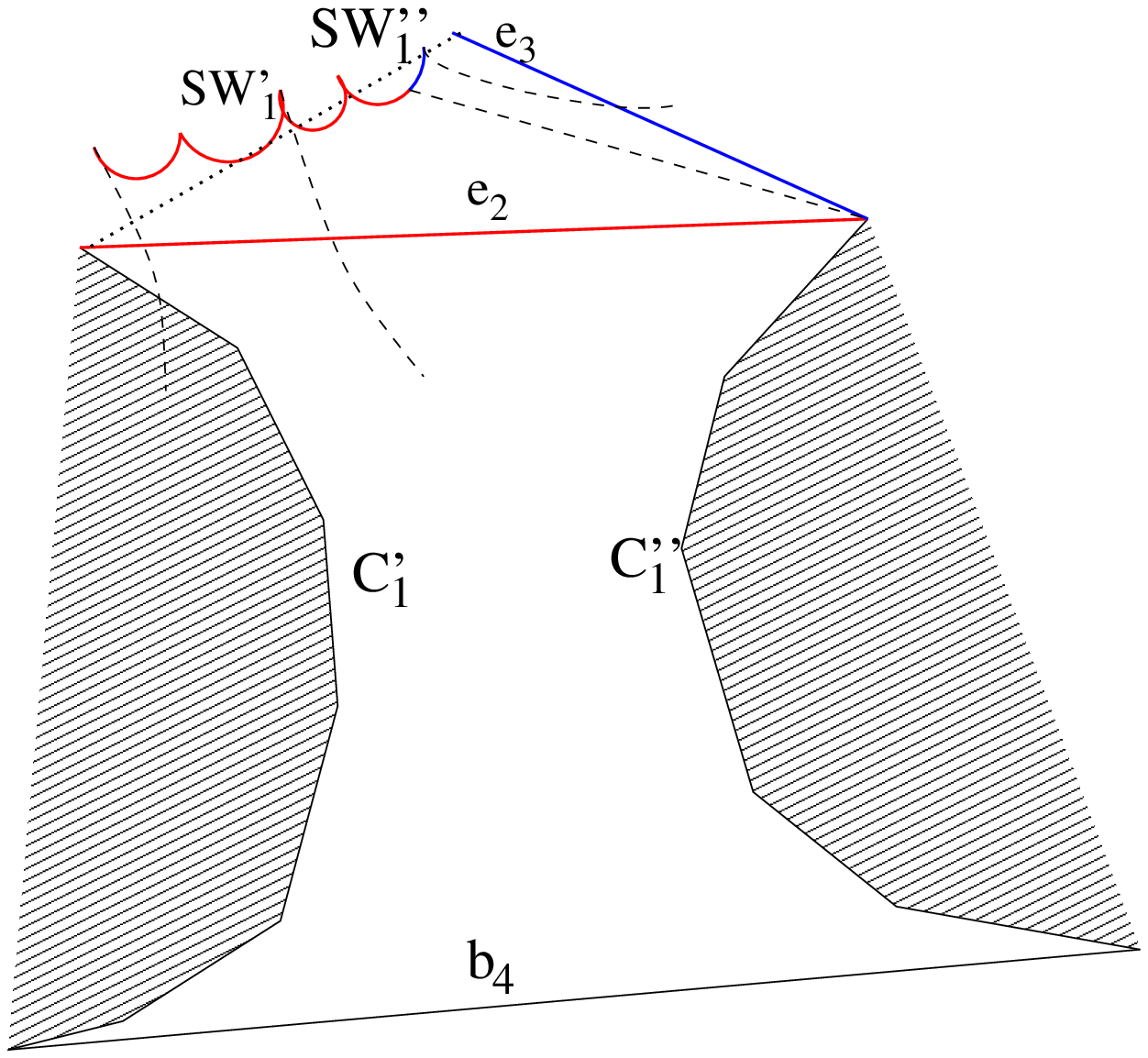}}
\subfigure[$SW_1'$ struck $e_2$ and $SW_2$ struck $e_3$]{\epsfxsize=200pt \epsfbox{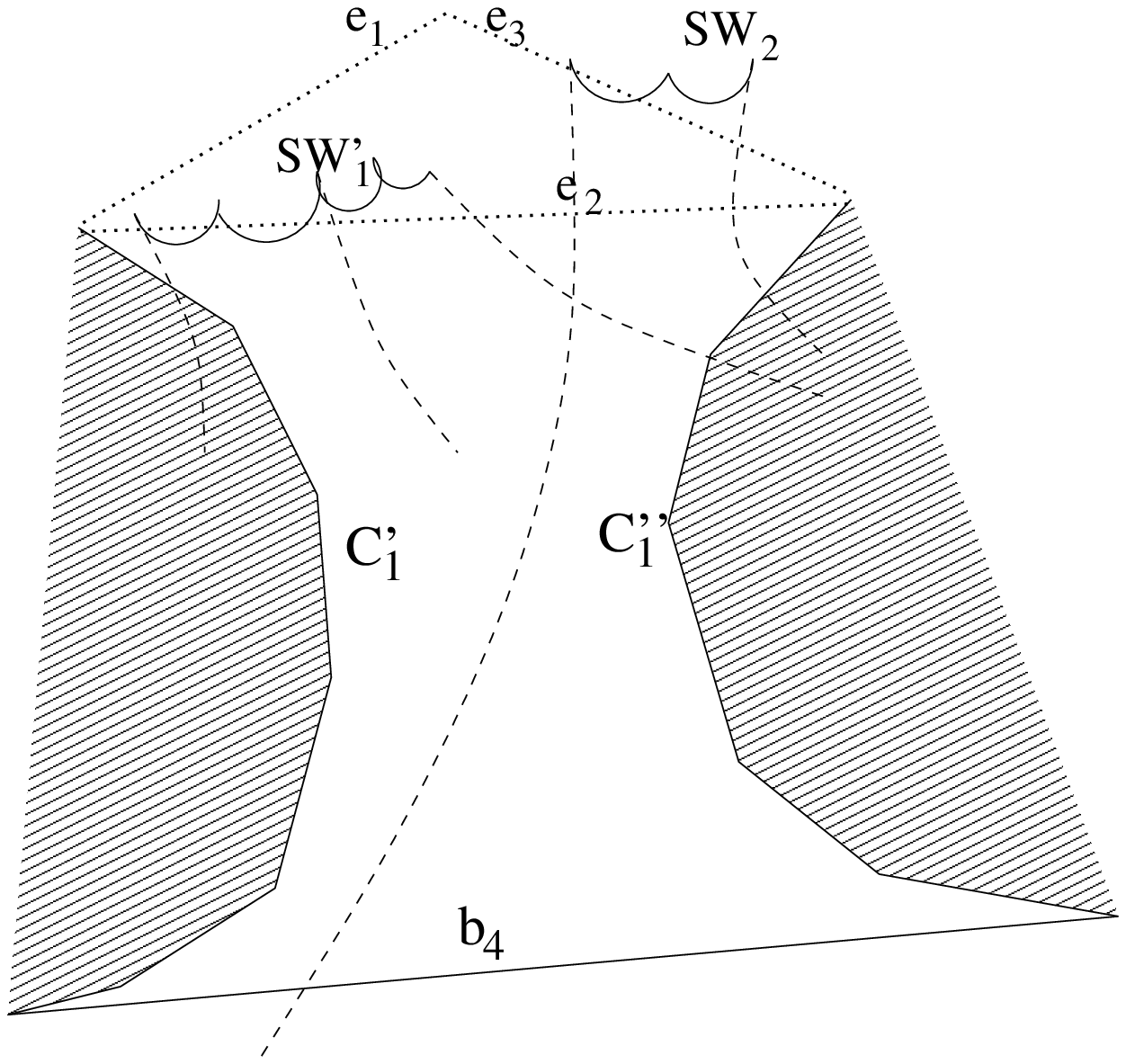}}
}
\caption{\label{fig:wfsplitmerger} Wavefront split and merger}
\end{figure}

For $d' > d$, when the wavefront is $\cW (d')$, suppose $e_2$ is struck by $SW_1'$ as shown in Fig. \ref{fig:wfsplitmerger}(c).
Now based on $SW_1'$ I-curves, the section of boundary associated with $SW_1'$ comprise $C_1' \cup b_4 \cup C_1''$.
For $d'' > d'$, when the wavefront is $\cW (d'')$, suppose $e_3$ is struck by a section of wavefront $SW_2$ (see Fig. \ref{fig:wfsplitmerger}(c)).  
This causes a boundary split.
Both the edges $e_1$ and $e_2$ are not considered as existing any more as they were respectively struck by $SW_1$ and $SW_1'$ in the past.
Then based on $SW_2$ I-curves, we know that $C_1'$ and $b_4$ are reachable from $SW_2$.
Since there is a section of boundary that is common to both $SW_1'$ and $SW_2$, we merge these two sections of wavefront. 
We merge two sections of wavefront whenever there is a boundary split.
The wavefront merge procedure computes the association of $C_1' \cup b_4 \cup C_1''$ with bunches in $SW_1' \cup SW_2$, based on the proximity of each of $C_1', b_4, C_1''$ with $SW_1'$ and $SW_2$.
Note that each section of wavefront shown in these pictures may comprise bunches which are not necessarily contiguous along the wavefront.
This is true even for the resultant section of wavefront that is associated with $C_1' \cup b_4 \cup C_1''$.

\begin{property}
A waveform merger is required whenever there is a boundary split.
\end{property}

\ignore{
\begin{figure}
\centerline{\epsfysize=200pt \epsfbox{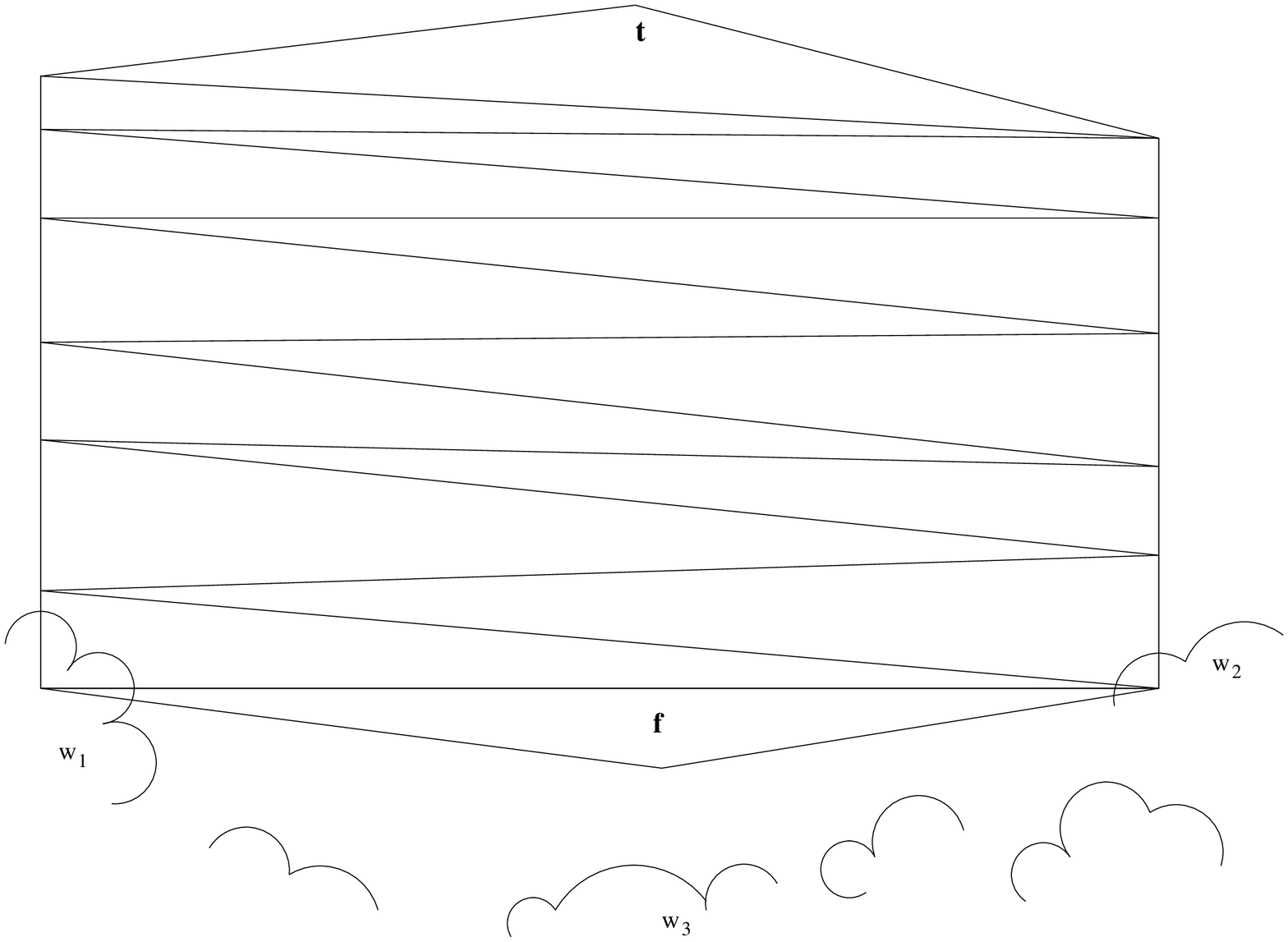}}
\caption{\label{fig:wfmerger} Wavefront merger}
\end{figure}

Consider an example that describes the wavefront merger.
Suppose a long facet $f$ reached by sections of wavefront $w_1$ and $w_2$ far from each other, as illustrated in 
See Fig. \ref{fig:wfmerger}.
Let an edge of face $f$ is struck by $w_1$ first.
When $w_2$ strikes another edge of $f$, we merge $w_1$ and $w_2$.
Let the new section of wavefront due to this merger be $w_{12}$.  
This is valid since there is no requirement for bunches in a section of wavefront to be contiguous. 
When a section of wavefront $w_3$, strikes a face $f'$ whose bounding edge is associated with $w_{12}$, we merge $w_{3}$ with $w_{12}$.  
We continue in this way of merging the sections of wavefront as the algorithm proceeds.  
\hfil\break
}

We associate either a section of boundary with a segment, or a section of wavefront with a $g \in \partial B$, such that neither a $g$ nor a $w(v)$ participates in more than one association.   
This limits the number of interactions between the wavefront and $\partial B$ to $O(n+m)$.

\paragraph{Bunches} \hfil\break

The shortest distance computations and updates of associations are proportional to the number of associations involved. 
Given that $n$ could be much larger than $m$, it would be interesting to explore whether the number of interactions/associations is a function of number of obstacles. 
Also, $O(n)$ wavefront segments could interact between themselves.
To reduce the number of such possible events, either we need to reduce the number of wavefront segments initiated or show that the number of segment intersections are limited to $o(n^2)$.
The latter is possible due to the following structure:

\begin{definition}
Let $v_j, v_{j+1}, \ldots, v_k$ be a maximal list of successive vertices along a corridor convex chain $CC$ such that for every two neighboring vertices $p, q$ in this ordered list $d(s, q) = d(s, p) + d(p, q)$.  
Then, a {\bf bunch} $B(v_j, v_k)$ is defined as the sequence of wavefront segments $w(v_j), w(v_{j+1}), \ldots, w(v_k)$.
\end{definition}

The I-curves among segments within a bunch are straight-lines.
Two I-curves are said to {\it diverge} if they are not intersecting, and will not intersect in future as the wavefront expands.

\begin{lemma}
\label{lem:divergeprop} The intra-bunch I-curves among the segments within a bunch diverge.
\end{lemma}

\begin{figure}
\centerline{\epsfxsize=180pt \epsfbox{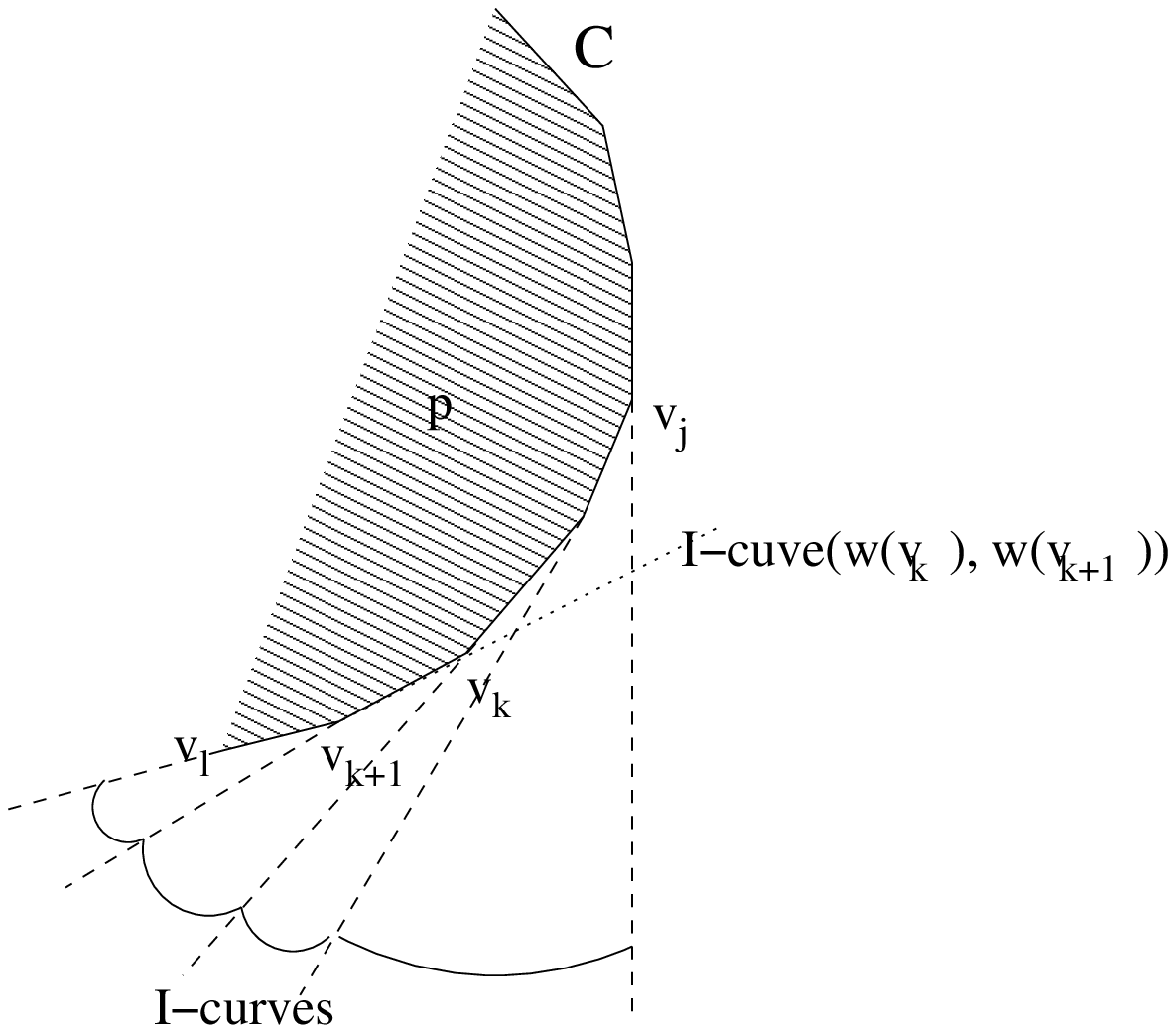}}
\caption{\label{fig:divergingicurves} Diverging I-curves}
\end{figure}

\begin{proof}
Consider a bunch $B$ with wavefront segments $w(v_j), \ldots, w(v_l)$.
Suppose $v_j$ lies on corridor convex chain $C$.
See Fig. \ref{fig:divergingicurves}.
Let $p$ be a point interior to the hull formed by the vertices of $C$, and $v_j, \ldots, v_l$ occur in clockwise (resp. counter-clockwise) direction w.r.t. $p$.
Since any two non-parallel lines can intersect at one point at most, and any edge $v_kv_{k+1}$ for $j< k \le l-1$ extended in coutner-clockwise (resp. clockwise) direction w.r.t. $p$ intersects all lines induced by edges in the set $\{v_jv_{j+1}, \ldots, v_{k-1}v_k\}$, the I-curves will not intersect in future on wavefront expansion.
\end{proof}

Let $CC, CC'$ be two corridor convex chains.
And, let the vertices of $CC$ be $ \{v_1, v_2, \ldots, v_i, v_{i+1},$ $\ldots, v_j=v_z, v_{j+1}, \ldots, v_k, \ldots, v_{n'}$\}.
Let $v_{zout}'$ be a vertex of $CC'$.
Suppose $w(v_{zout}')$ strikes $CC$ at a vertex $v_z$ along the tangent from $v_{zout}'$ to $v_z$.
The tangent $v_{zout}'v_z$ defines a sequence $S$ of $CC$ vertices, $v_{j+1}, \ldots, v_{n'}$, none of which are visible from $v_{zout}'$.
We initiate a bunch $B(v_z, v_{n'})$, whose wavefront segment centers are $v_j, v_{j+1}, \ldots, v_{n'}$, and insert $B(v_z, v_{n'})$ into the wavefront.
\hfil\break

\begin{figure}
\centerline{\epsfxsize=300pt \epsfysize=280pt \epsfbox{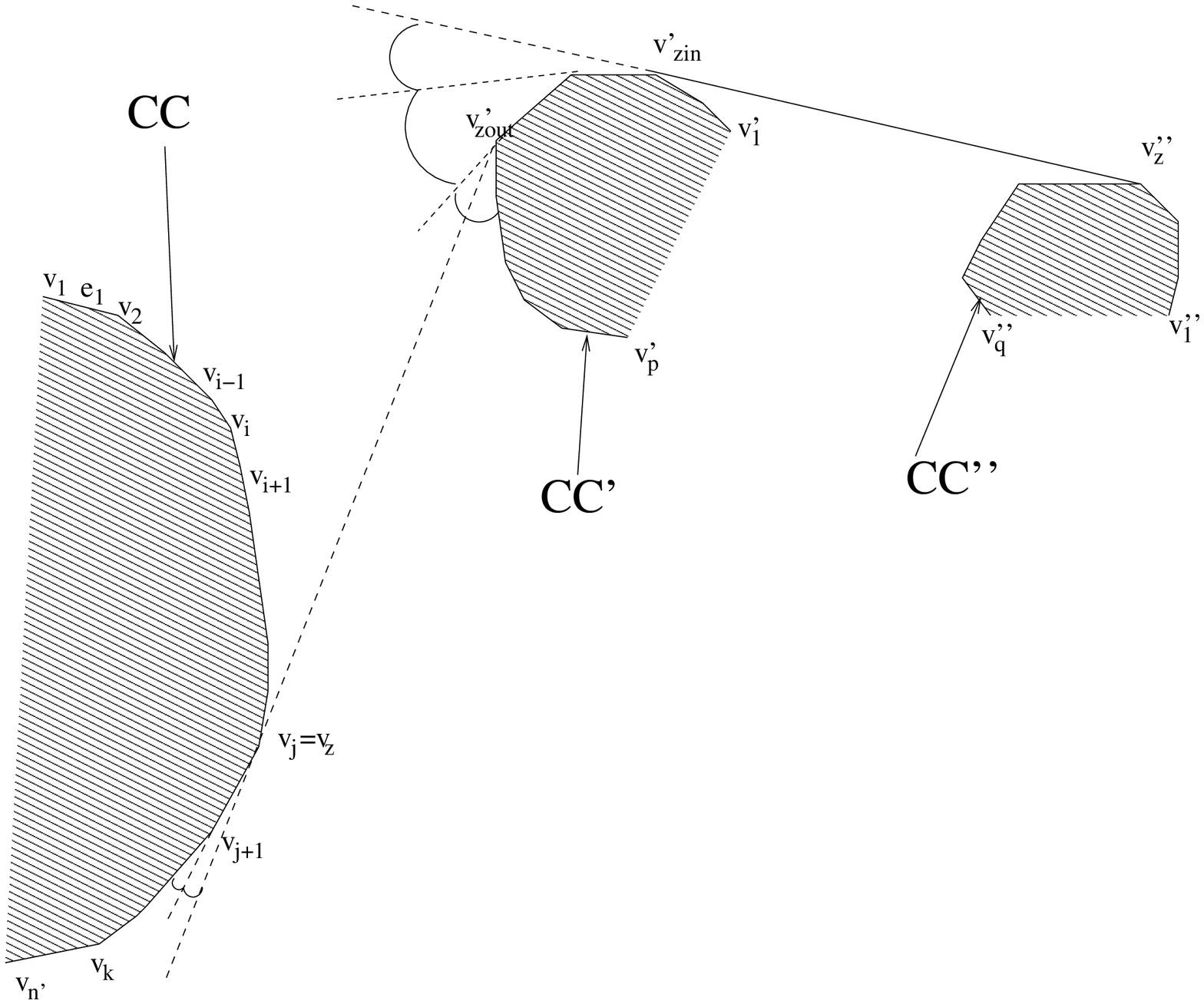}}
\caption{\label{fig:bunchinit}{Bunch initiation}}
\end{figure}

The initial wavefront consists of a (degenerate) bunch with one segment, $w(s)$.
Bunches are initiated whenever a wavefront segment strikes a corridor convex chain tangentially.
As shown in Fig. \ref{fig:bunchinit}, when a segment from a bunch initiated from a convex chain $CC''$ strikes another convex chain $CC'$ with $v_z''v_{zin}'$ as a common tangent, we initiate a bunch $B'(v_{zin}', v_{p}')$ from $v_{zin}'$.
Similarly, when a segment $w(v_{zout})$ of a bunch $B(v_{zin}', v_p')$ strikes with $v_{zout}'v_z$ as a common tangent between $CC$ and $CC'$, a new bunch $B(v_z, v_{n'})$ is initiated.
This process continues until we reach the corridor in which $t$ resides.
\hfil\break

We consider a closed corridor as a special case of open corridor.
See Fig. \ref{fig:closedcorri}.
Let $v_1, v_2$ be the apex points of two funnels.
Suppose $v_2$ is not yet struck.
After the wavefront strikes $v_1$, (two) bunches are initiated from $v_2$ whenever the wavefront expands from $v_1$ a distance that equals the shortest distance from $v_1$ to $v_2$.
This distance is precomputed for each closed corridor.
The bunches initiated from $v_2$ correspond to two corridor convex chains that  originate from $v_2$ and define the funnel with apex $v_2$.
Let $\{v_2, u_1, u_2, \ldots, u_k\}$ and $\{v_2, u_1', u_2', \ldots, u_l'\}$ be the sequence of vertices along the two convex chains of funnel with apex $v_2$.
Then the two bunches that are initiated comprise the wavefront segments $w(v_2), w(u_1) , w(u_2), \ldots w(u_k)$ and $w(u_1') , w(u_2') \ldots w(u_l')$ i.e., with the wavefront segment originating at $v_2$ is included in only one bunch.
\hfil\break

Hence at any stage of the algorithm, the wavefront is formed by a set of bunches.

\begin{property}
\label{prop:bunchseg} The wavefront $\cW (d)$, for any $d$, is composed of a set of bunches.
\end{property}

At any stage of the wavefront progression, a bunch $B(v_z, v_{n'})$ consists of a sequence of segments that are already initiated, followed by the rest of the uninitiated wavefront segments (if there are any) that possibly belong to $B(v_z, v_{n'})$ in future. 
However, for any vertex $v$ in $\{v_z, \ldots, v_{n'}\}$, $w(v)$ is part of $B(v_z, v_{n'})$ if and only if the wavefront strikes $v$ when the wavefront is at Euclidean distance $d(s, v_z) + d(v_z, v)$ from $s$.
Each bunch is maintained in a balanced tree data structure.
\hfil\break

\begin{figure}
\centerline{
\subfigure[Both the bunch and $w(v_z)$ are initiated at $v_z$]{\epsfxsize=180pt \epsfbox{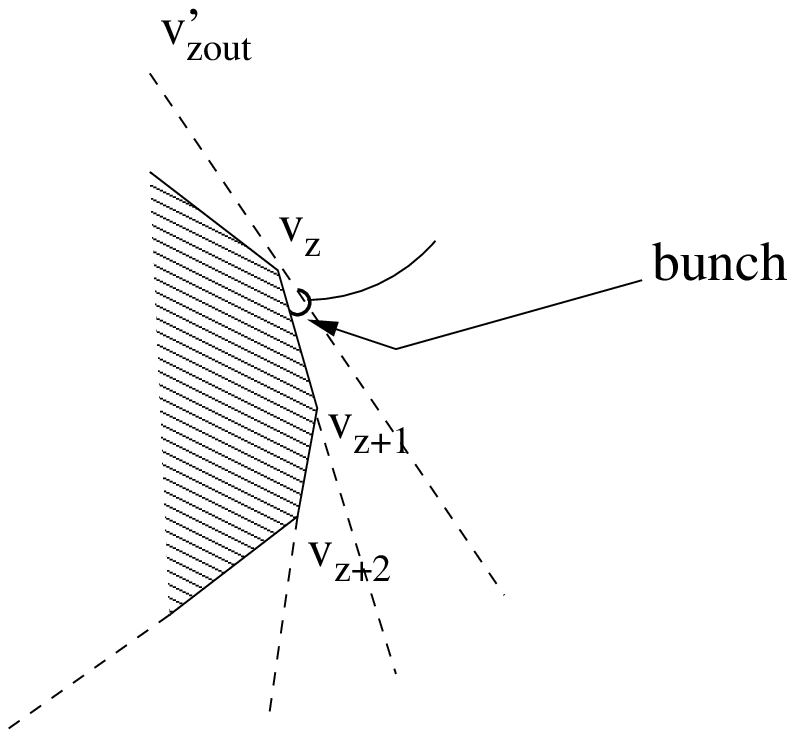}}
\subfigure[$w(v_{z+1})$ is initiated in the bunch]{\epsfxsize=180pt \epsfbox{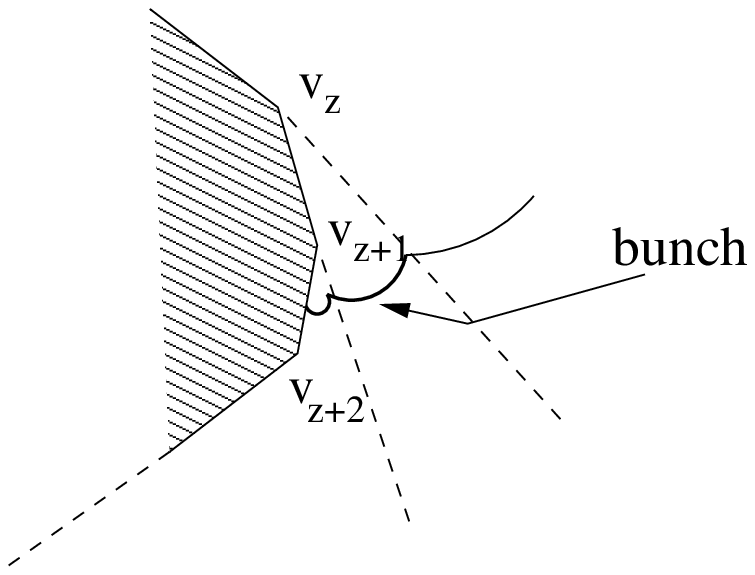}}
}
\centerline{
\subfigure[$w(v_{z+2}$ is initiated in the bunch]{\epsfxsize=180pt \epsfbox{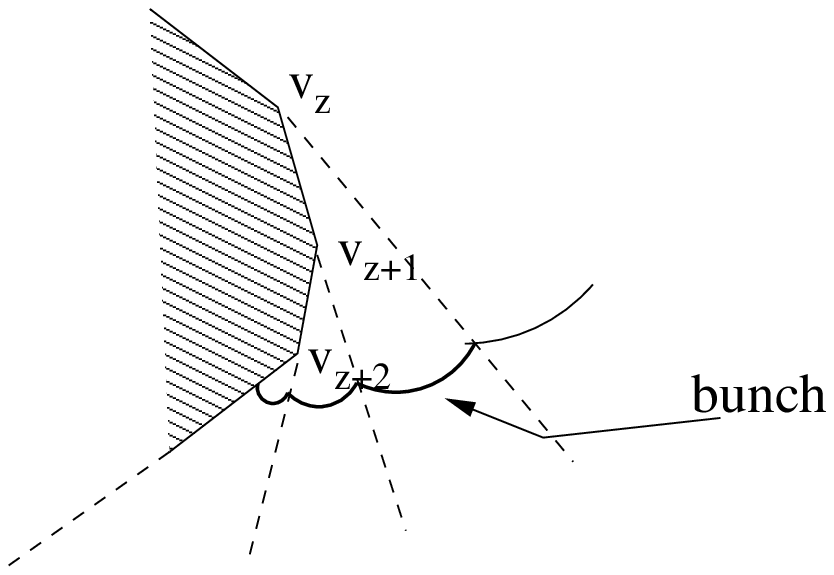}}
}
\caption{\label{fig:bunchprogr}{Initiation of wavefront segments within a bunch}}
\end{figure}

As Fig. \ref{fig:bunchprogr}(a) shows, when $v_z$ is struck by the wavefront such that the line joining $v_z$ and $v_{zout}'$ is a tangent to $UH$, we initiate a bunch $B(v_z, v_{n'})$ with $w(v_z)$.
With the wavefront expansion, when $w(v_z)$ strikes $v_{z+1}$, we initiate $w(v_{z+1})$ in bunch $B(v_z, v_{n'})$.
See Fig. \ref{fig:bunchprogr}(b).
Similarly, $w(v_{z+2})$ is initiated in bunch $B(v_z, v_{n'})$ with further wavefront expansion.
See Fig. \ref{fig:bunchprogr}(c).
Although the bunch $B(v_z, v_{n'})$ is inserted into the wavefront when $w(v_z)$ is initiated, the rest of the wavefront segments are in $B(v_z, v_{n'})$ when the wavefront strikes their corresponding centers.
The satellite data associated with the nodes of the data structure corresponding to every bunch $B$ facilitates in determining which wavefront segments within $B$ are initiated. 

\paragraph{Bunch I-curves and Associations} \hfil\break

Since segments within a bunch does not interact between themselves, only the interactions between bunches, and the interaction of bunches with $\partial B$ are of interest. 
Given that both the number of bunches and $|\partial B|$ are a function of number of corridors, which is again a function of number of obstacles, the total number of events is a function of $m$. 
Similar to associations of segments with the boundary sections, we are interested in associations between bunches and the boundary sections.
\hfil\break

\begin{figure}
\center{
\subfigure[Associations of $B_1, w(s)$, and $B_2$ \hfil\break are respectively \{$b_1$\}, \{$b_1, b_2$\}, \{$b_2$\}]{\epsfxsize=220pt \epsfbox{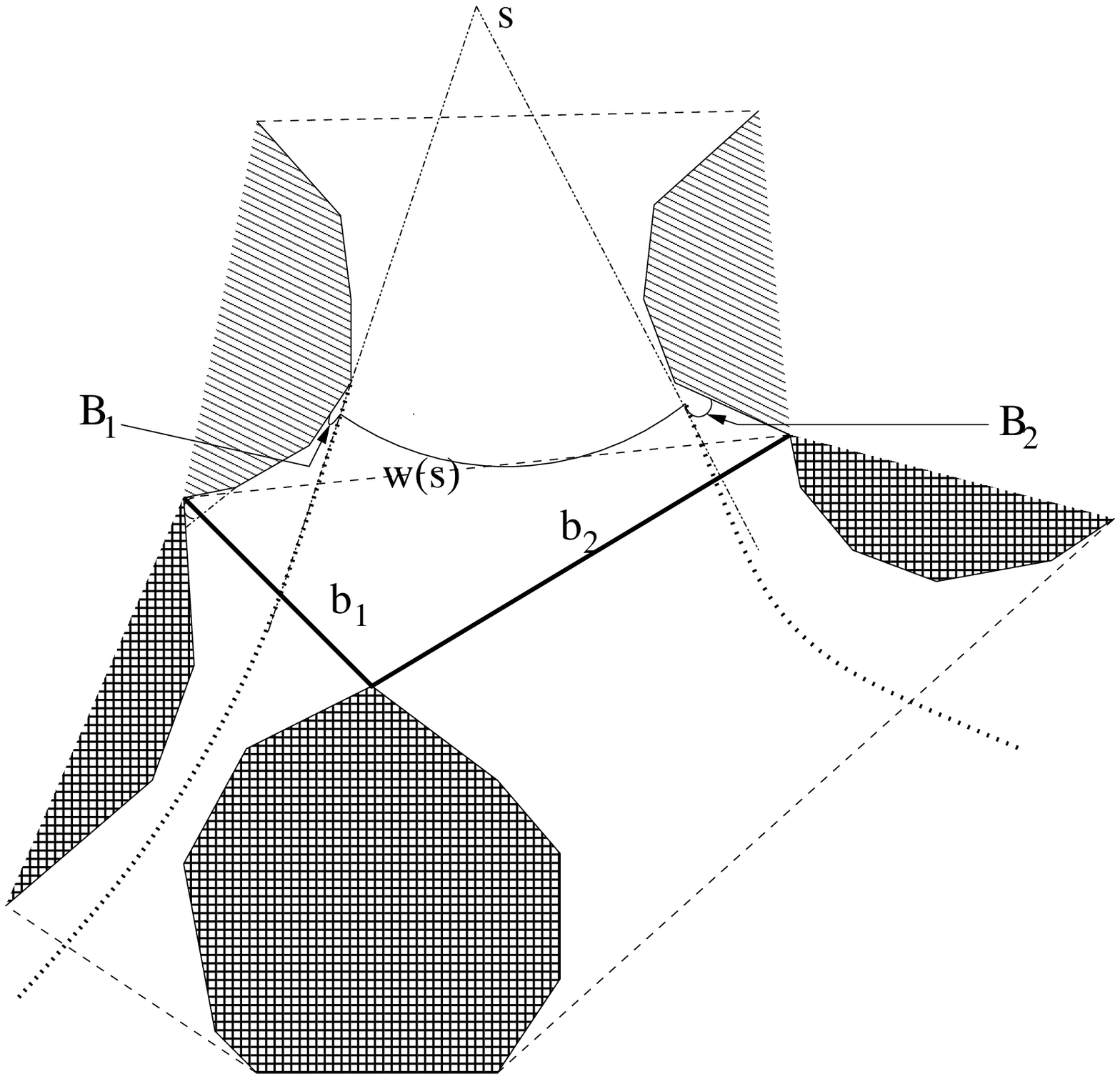}}
\subfigure[Associations of $B_1, w(s)$, and $B_2$ \hfil\break are respectively \{$C_2', b_3$\}, \{$b_3, C_2'', C_3'$\}, \{$b_4, C_3''$\}]{\epsfxsize=240pt \epsfbox{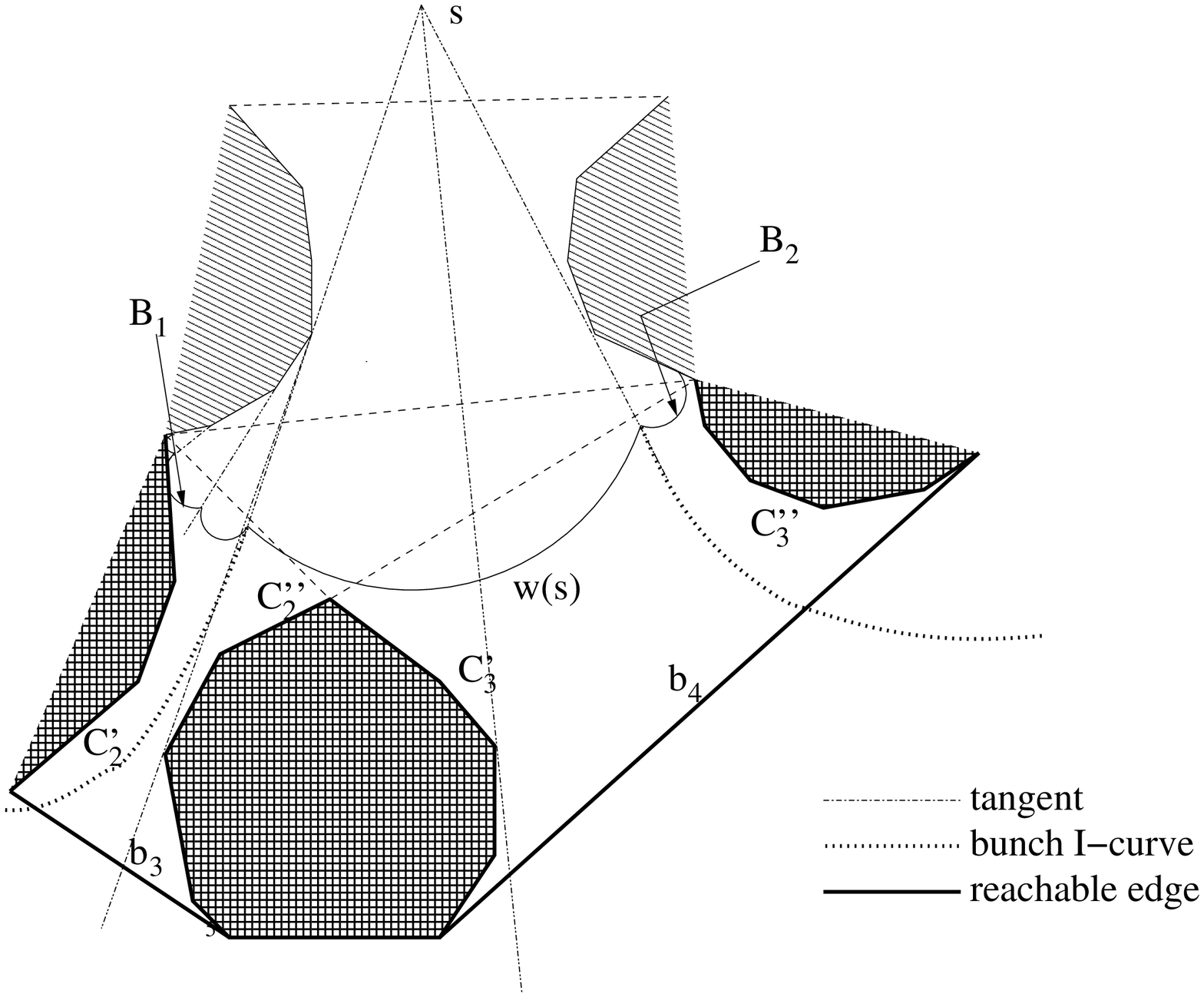}}
}
\caption{\label{fig:bunchcorriassoc} Associations of bunches with $\partial B$}
\end{figure}

Consider a corridor convex chain or enter/exit boundary $g$ in a boundary cycle.
Let $B_g$ be the set of bunches from each of which $g$ is reachable.
Every bunch $S \in B_g$ is {\bf associated} with $g$ (or, $g$ is associated with bunch $S$) if and only if a point on $g$ has shortest Euclidean distance to $s$ via the center of a segment $w \in S$.
The  association is defined by the relation: ${\cal A} \subseteq {\cal G} \times {\cal B}$, where ${\cal G}$ is the set of corridor convex chains or enter/exit boundaries in $\cP$, and $B$ is the set of bunches formed during the course of algorithm.
Again, note that it is not necessary for any two bunches in $B_g$ to be contiguous in the wavefront.
The example in Fig. \ref{fig:bunchcorriassoc}(a) shows an association of bunches with $\partial B$, and how the associations are updated in Fig. \ref{fig:bunchcorriassoc}(b).  

\begin{lemma}
\label{lem:contigprop2} There exists an association {\cal A} such that the sequence of boundary edges on a boundary cycle that are associated with a bunch is a contiguous sequence.
This is known as {\it contiguity property for bunches}.
\end{lemma}

The proof for the above Lemma is similar to Lemma~\ref{lem:contigprop1}.
\hfil\break

To help in associating bunches with $\partial B$, we define Voronoi regions corresponding to bunches. 
Let $w(v_i), w(v_{i+1}), \ldots, w(v_k)$ be the segments of a bunch $B(v_i, v_k)$.
Let $w(v_q), w(v_{q+1}), \ldots, w(v_r)$ be the segments of a bunch $B(v_q, v_r)$.
Let the bunches $B(v_i, v_k), B(v_q, v_r)$ be adjacent along the wavefront such that the wavefront segments $w(v_k)$ and $w(v_q)$ are adjacent.
Then the I-curve($v_k, v_q$) is the {\bf inter-inter-bunch I-curve} between $B(v_i, v_k), B(v_q, v_r)$, denoted as I-curve($B(v_i, v_k), B(v_q, v_r)$). 
Although the intra-bunch I-curves are diverging rays (Lemma \ref{lem:divergeprop}), inter-bunch I-curves could be higher-order curves.
\hfil\break

\begin{figure}
\centerline{\epsfysize=250pt \epsfbox{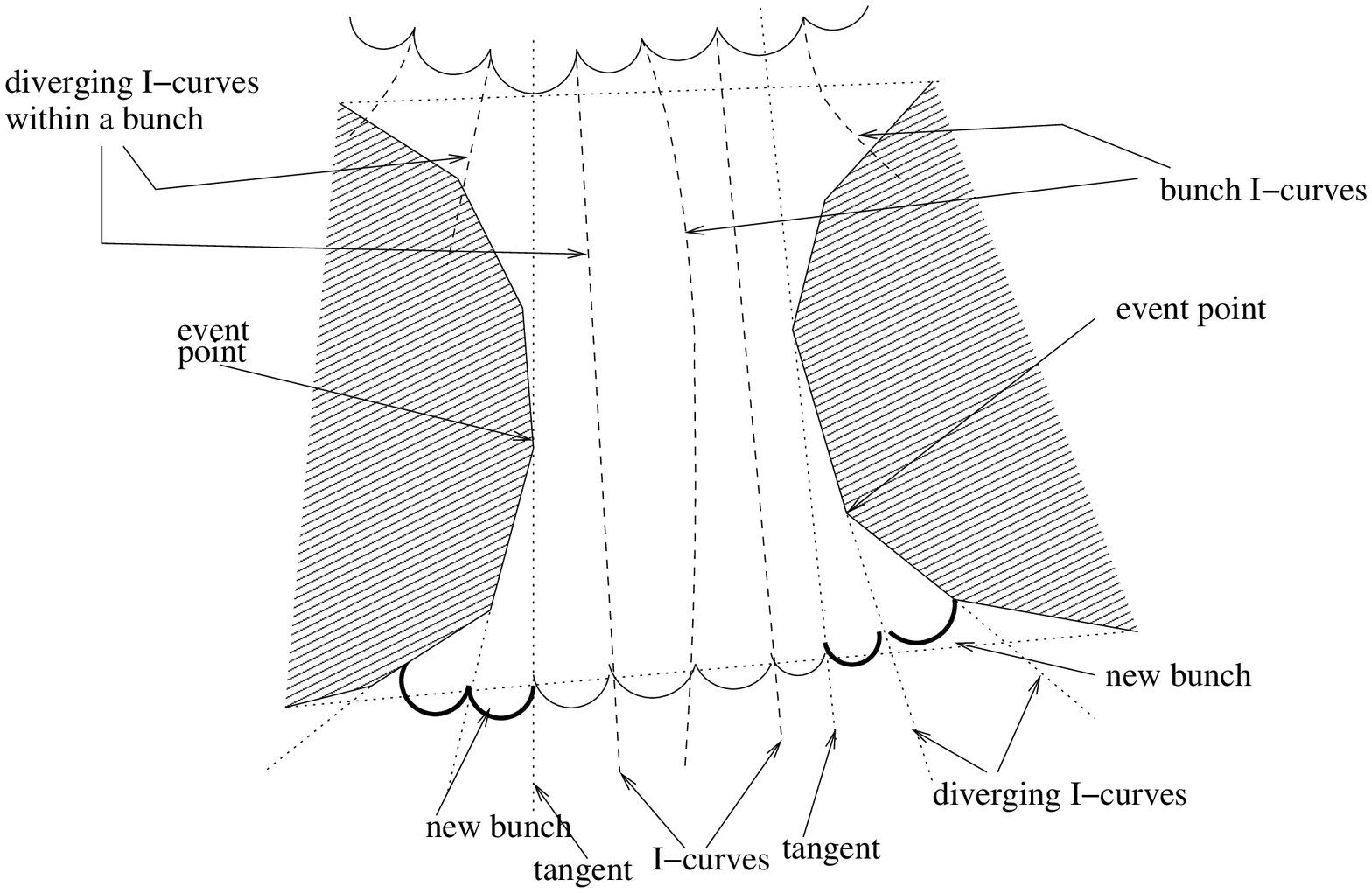}}
\caption{\label{fig:opencorri}{Open corridor with bunches}}
\end{figure}

An example in Fig. \ref{fig:opencorri} shows bunches, I-curves among wavefront segments in a bunch, inter-bunch I-curves.
The following two definitions help in maintaining the associations:

\begin{definition}
Let $B$ be the set of bunches associated either with a corridor convex chain or corridor enter/exit boundary $C$, i.e., $\forall b \in B, (C, b) \in \cA$.  
When $|B|>1$,  a {\bf waveform-section} for $C$, denoted by $WS(C)$, is the sequence of bunches in $B$.
Note that the bunches in $WS(C)$ need not be contiguous in the wavefront.
\end{definition}

\begin{definition}
Let a bunch $b$ be associated to a set $S$, where each element of $S$ is a corridor convex chain or enter/exit boundary.  
When $|S| > 1$, a {\bf boundary-section} for bunch $b$, denoted with $BS(b)$, is the contiguous sequence of corridor convex chains or enter/exit boundaries in $S$.
\end{definition}

Consider a corridor convex chain or enter/exit boundary $C$ for which a waveform-section, $WS(C)$, is defined.
The above partitioning ensures that if a bunch $b$, which has not yet struck $C$ before, strikes a point $p \in C$  before striking any other bounding edge in $\partial B$, then the bunch $b$ is guaranteed to be in $WS(C)$.  
Furthermore, consider a bunch $b$ for which a boundary-section, $BS(b)$, is defined.
On wavefront expansion, if the bunch $b$ strikes $\partial B$, it would do so only by striking a point $p \in BS(b)$.  
In general, a section of wavefront represents either a sequence of bunches (not necessarily contiguous) in a waveform-section, or the bunch associated to a boundary-section.  
\hfil\break

The {\it RV} defined below is used in initiating/updating boundary-/waveform-sections and it helps in the analysis.  

\begin{definition}
Suppose a bunch $a_r$ is associated with a contiguous sequence of corridor convex chains/exit boundaries, say $B_S$, i.e. $(e,a_r) \in {\cal A}, \forall e \in B_S$.  
If there exists at least one corridor convex chain or enter/exit boundary $C \in S$ such that $C$ is solely associated with bunch $a_r$ then {\bf $RV(a_r)$} is defined as $B_S$; otherwise, $RV(a_r)$ is $\phi$.
\end{definition}

For a bunch $b$, suppose $|RV(b)| = 0$.
Hence the bunch $b$ is associated solely to at least one corridor convex chain or enter/exit boundary, say $C$.  
Then the waveform-section of $C$, $WS(C)$, is obtained as the (not necessarily contiguous) sequence of bunches associated with $C$.  
For a bunch $b$, we define boundary-section for $b$, $BS(b)$, only if $|RV(b)| > 1$.
When defined, $BS(b)$ is the sequence of corridor convex chains/exit boundaries in $RV(b)$.  
The example in Fig. \ref{fig:wsbsassoc} shows $RV$s together with associations, boundary-sections, and waveform-sections.
\hfil\break

\begin{figure}
\centerline{\epsfxsize=430pt \epsfbox{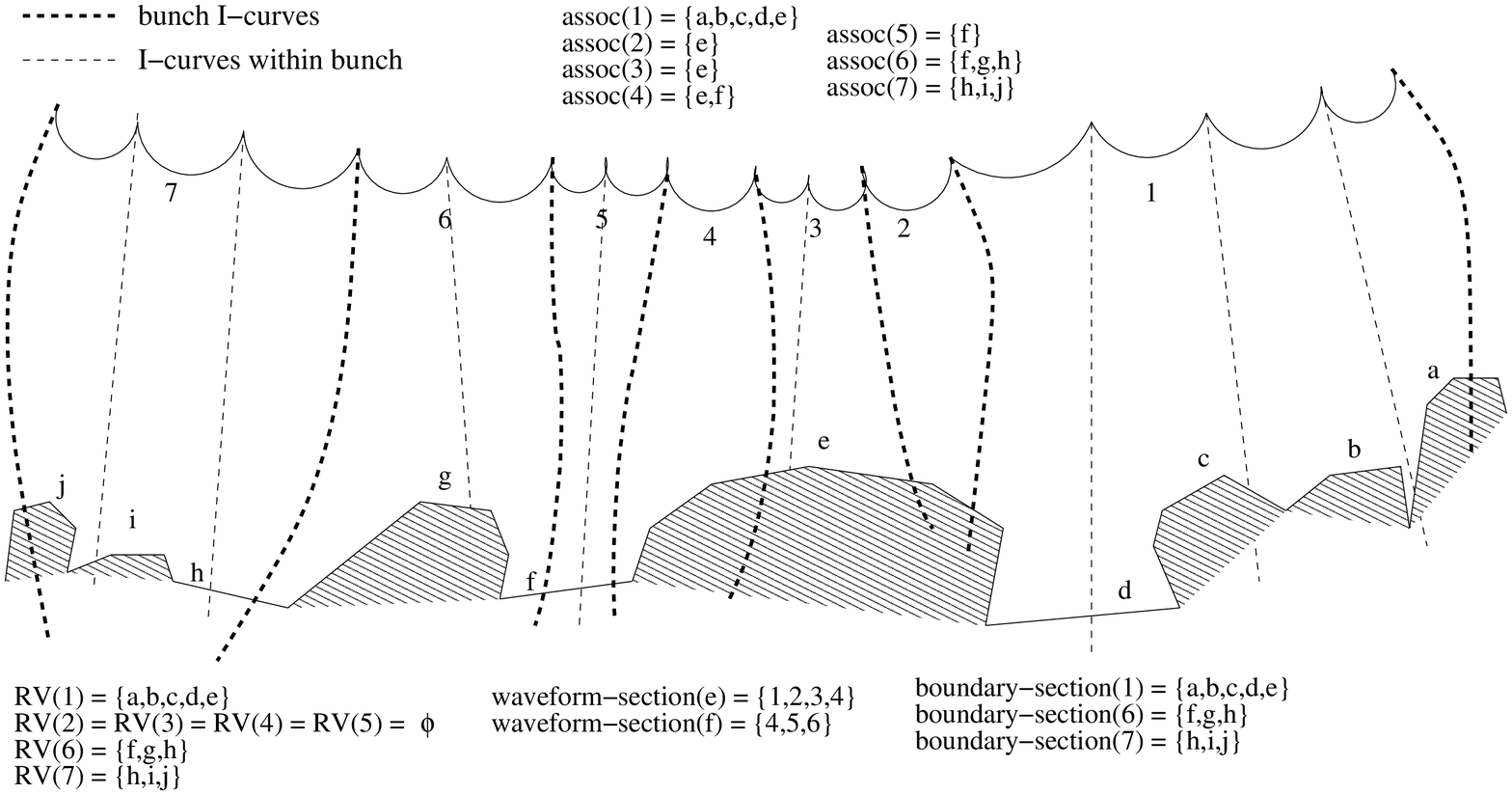}}
\caption{\label{fig:wsbsassoc} Associations, boundary-sections and waveform-sections of a section of wavefront}
\end{figure}

Fig. \ref{fig:wsbsassoc} shows $RV$s, associations, boundary- and waveform-sections for a section of boundary and a section of wavefront.
\hfil\break

As shown in Fig. \ref{fig:evolveboundcycle}, when the wavefront is initiated the boundary-section $BS$ of $w(s)$ is $b_1b_2b_3$.
We compute the shortest distance $d'$ between $BS$ and $w(s)$.
When $w(s)$ strikes $b_1$, a waveform-section $WS_1$ comprising the boundary $C_1', b_4, C_1'$ is associated with $w(s)$ is initiated, and $b_1$ is deleted from $BS$. 
At the next event point i.e., when $w(s)$ strikes $b_3$, a new waveform-section $WS_2$ comprising the boundary $C_3'', b_6, C_3'$ is associated with $w(s)$ is initiated, and $b_3$ is deleted from $BS$. 
We continue initiating/updating/deleting waveform- and boundary-sections until the wavefront strikes $t$.
\hfil\break

Boundary-sections are primarily useful when the boundary splits.
The data structure that stores the boundary sections' facilitates in finding the convex chain vertices that are visible from the vertices of wavefront segments.
This in turn helps in initiating bunches.

\section{Algorithm Outline}
\label{sect:algooutline}

Initially, the wavefront is a segment $w(s)$ with radius $\epsilon$, and the boundary-section (also, the initial boundary cycles) is the triangle in which $s$ resides.
The event points are based on the interaction of wavefront with the boundary cycles. 
The events occur as the wavefront progresses and are maintained in a min-heap, with the corresponding shortest distance at which an event occurs as the key. 
Based on the type of event, an event handler procedure is invoked to handle the event.
The event points are categorized into the following types:  \hfil\break
\hfil\break

\begin{tabular}{lp{4.8in}}
Type-I &
Occurs when a segment $w(v)$ in a waveform-section $WS$ strikes the associated corridor convex chain or enter/exit bounding edge $C$.
This event is determined by computing the shortest distance between $WS$ and $C$. 
The event handler accomplishes the following:
\vspace{-0.1in}
\begin{enumerate} \itemsep -2pt
\item Updating relevant waveform-sections and their associations.
\item Updating relevant boundary-sections and their associations.
\item Handling relevant boundary splits if there are any.
\item Using the new associations, computing the shortest distances and pushing corresponding event points to min-heap.
\item If the event is occurred due to wavefront strike of an enter/exit bounding edge of a corridor $C'$, then Type-III events are pushed to the min-heap: these event points correspond to the wavefront progression at which new bunches from the convex chains of $C'$ are initiated.
\end{enumerate}
\end{tabular}

\begin{tabular}{lp{4.8in}}
Type-II &
Occurs when a segment $w(v)$ in a bunch $B$ strikes the associated boundary-section $BS$.
This event is determined by computing the shortest distance between $BS$ and $w(v)$.
The handling procedure is same as the Type-I handler.  \hfil\break
\end{tabular}

\begin{tabular}{lp{4.8in}}
Type-III &
Occurs when a segment $w(v)$ in a section of wavefront $SW$ strikes a corridor convex chain $C$ at a point of tangency $p$ such that the line segment $vp$ is a tangent from $v$ to $C$. 
Either a Type-I or a Type-II event could cause Type-III event.
The handler procedure does the following:
\vspace{-0.1in}
\begin{enumerate} \itemsep -2pt
\item A new bunch $B$ is initiated from $p$.
\item The associations of bunch $B$ are determined and the relevant waveform-sections are updated.
\item Also, shortest distance corresponding to new associations are pushed to min-heap.
\end{enumerate}
\end{tabular}

\begin{tabular}{lp{4.8in}}
Type-IV &
Occurs when upper convex hull boundary approximations $UH_i, UH_j$ of any two bunches $b_i, b_j$ within a waveform-section $WS$ intersect for the first time. 
This is determined by computing the shortest distance between pairs of bunches within $WS$. 
As there could be $O(m)$ bunches within a $WS$, this could lead to $O(m^2)$ computations.
However, the number of events are reduced by utilizing the special structure in saving the waveform-sections (detailed in Subsection \ref{subsubsect:wstdatastr}). 
The event handler primarily does the following:
\vspace{-0.1in}
\begin{enumerate} \itemsep -2pt
\item Suppose $SW$ is the section of wavefront between $b_i$ and $b_j$ such that $t$ is not located in the region enclosed by $SW \cup \{b_i\} \cup \{b_j\}$.  Since no wavefront segment in $SW$ can cause a shortest path to $t$ without crossing some wavefront segment not in $SW$, utilizing the non-crossing property of shortest paths (Lemma \ref{lem:spnoncrossing}), bunches in $SW$ are deleted from the wavefront.
\item Updating relevant waveform- and boundary-sections.
\item Pushing new shortest distances with the updated waveform- and boundary-sections.
\end{enumerate}
\end{tabular}
\hfil\break
\hfil\break
The detailed descriptions of the determination and handling procedures are given in Section \ref{sect:evtpts}. 
We continue processing the events scheduled from the min-heap till the sink $t$ is struck.
When this happens, we compute the shortest path and distance from $s$ to $t$.

\section{Data Structures}
\label{sect:datastr}

This Section describes all the required data structures.

\subsection{Bunch Hull Tree ({\it BHT})}
\label{subsect:bhtdatastr}

\begin{figure}
\center{
\subfigure[A bunch]{\epsfxsize=170pt \epsfbox{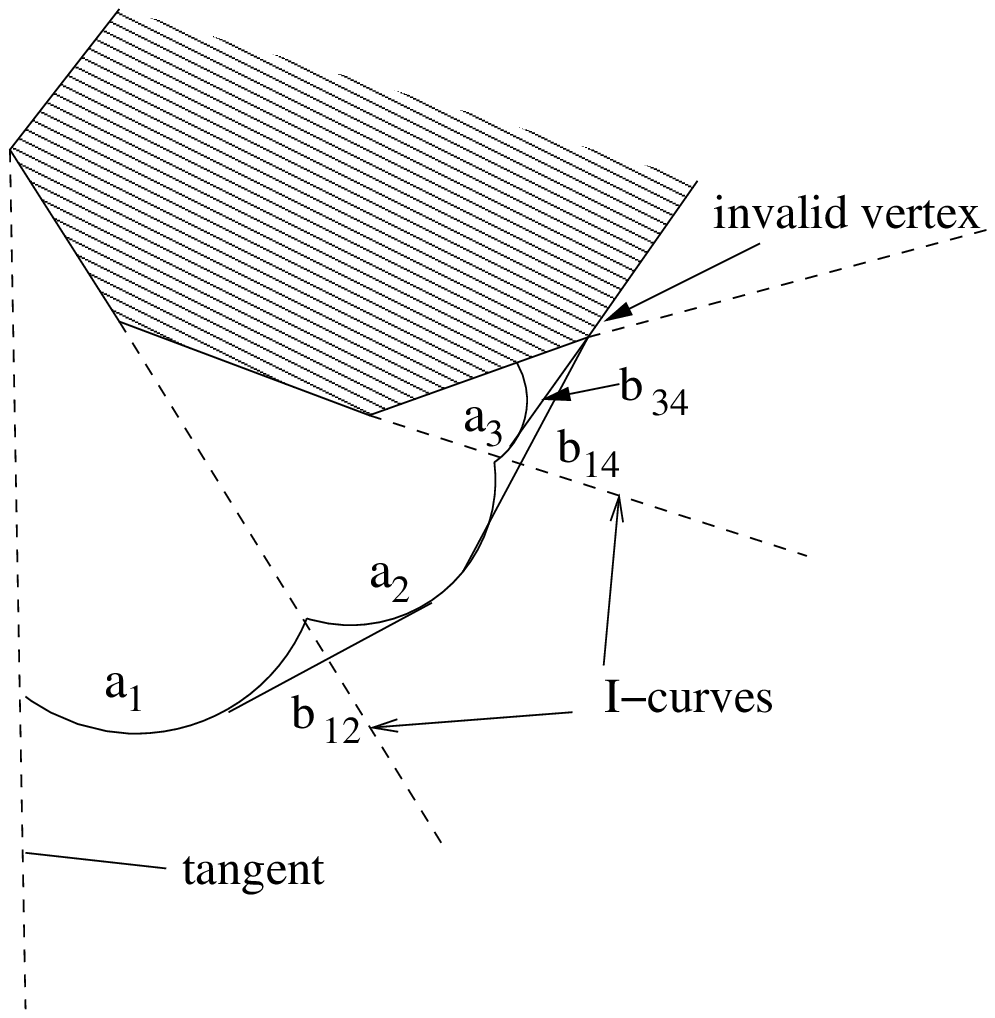}}
\subfigure[Corresponding $BHT$]{\epsfxsize=170pt \epsfbox{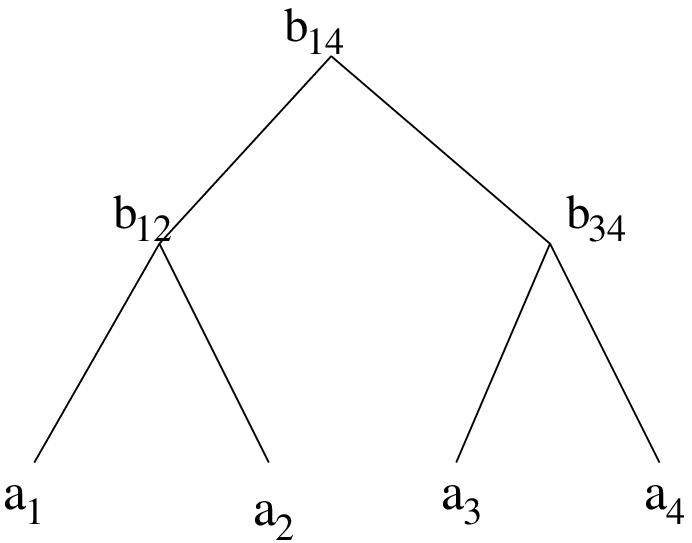}}
}
\caption{\label{fig:bhtexample} A bunch and its $BHT$}
\end{figure}

Each bunch is stored in a balanced tree structure, termed as a {\it bunch hull tree ($BHT$)}.
Suppose that there is no bunch originated from a corridor convex chain $CC$.
Let a vertex $v_j$ of $CC$ be struck by a wavefront segment $w(v_{zout}')$ such that the line $v_{zout}'v_j$ is a tangent to $CC$.
Also, let $v_j \in \cW (d)$. 
As shown in Fig. \ref{fig:bunchinit}, a new bunch $B(v_j, v_{n'})$ is initiated i.e., a bunch hull tree $BHT(v_j, v_{n'})$ corresponding to $B(v_j, v_{n'})$ is constructed.
A wavefront segment $w(v)$ in a bunch is termed as {\it initiated} if $d(s, v)$ is determined.
Since the wavefront struck $v_j$ and hence $d(v_j, s)$ is known, $w(v_j)$ is initiated.
At a given $\cW (d)$, the wavefront segments of a bunch that are initiated are termed as {\it valid segments} and the other wavefront segments of bunch are termed as {\it invalid segments} (see Fig. \ref{fig:bunchprogr}).
The invalid segments of $B(v_z, v_{n'})$ may possibly become valid in future as the wavefront progresses.
\hfil\break

The leaves of $BHT(v_j, v_{n'})$ consists of $w(v_j), \ldots, w(v_{n'})$, wherein $w(v_j)$ is the only valid segment and no wavefront segment initiated from the remaining vertices.
An example is shown in Fig. \ref{fig:bhtexample}.
Each tree node $p$ represents an upper hull of the wavefront segments at the leaves of the subtree at $p$.
Let $q, r$ be the child nodes of an internal node $p$ of $BHT(v_j, v_{n'})$.
Also, let $UH_q, UH_r$ be the upper hulls at $q$ and $r$ respectively.
The node $p$ stores the common tangent segment $br_p$ between $UH_q$ and $UH_r$.
This common tangent is known as a {\it bridge} between $UH_q$ and $UH_r$.
The upper hull at a node $p$ of the tree can be determined from $UH_q, br_p,$ and $UH_r$. 
\hfil\break

When initiating a $BHT$, only one wavefront segment within that $BHT$ is initiated.
Consider a leaf node $l_{v_k}$ of $BHT(v_j, v_{n'})$ that corresponds to $w(v_k)$, where $v_k$ is a vertex of $CC$ and belongs to \{$v_j, \ldots, v_{n'}$\}. 
>From the definition of a bunch, $w(v_k)$ is added to $BHT(v_j, v_{n'})$ if and only if $d(s, v_k) = d(v_k, v_{k-1}) + d(v_{k-1}, v_{k-2}) + \ldots + d(v_{j+1}, v_j) + d(v_j, s)$.
Suppose that this is the case.
To facilitate in automatic implicit insertion of $w(v_k)$ to $BHT$, the data member {\it wpupdate} associated with $l_{v_k}$ stores the negated Euclidean distance along $CC$ between $v_j$ and $v_k$.
In other words, for $w(v_k)$ to be inserted to $BHT(v_j, v_{n'})$, the wavefront segments at $w(v_{j+1}), \ldots, w(v_{k-1})$ needed to be initiated in that order and $w(v_{k-1})$ requires to strike $v_k$.  
The root node $r$ stores $d(s, v_j)$ in a data member, {\it shortestdist}.

\begin{property}
\label{lem:wpupdate} 
Let $\cW (d)$ be the wavefront that caused a bunch $B(v_j, v_{n'})$ to be initiated.
Let $v_k \in \{v_j, \ldots, v_{n'}\}$.
Suppose $w(v_k)$ is struck by $w(v_{k-1}) \in B(v_j, v_{n'})$ when the wavefront is $\cW (d')$.
Also, let $r$ be the root node and $l_{v_k}$ corresponds to $w(v_k)$ in $BHT(v_j, v_{n'})$.
The wavefront segment $w(v_k)$ in $B(v_j, v_{n'})$ is valid if and only if $(d'-r.shortestdist+l(v_k).wpupdate) > 0$.
\end{property}

The root node of $BHT(v_j, v_{n'})$ refers vertex $v_{zout}'$ in {\it tangentstart}.

\begin{lemma}
\label{lem:spsdtos} The variables stored at roots and leaves of all the $BHT$s are sufficient to compute the shortest distance and shortest path from any valid segment $w(v_k)$ to $s$.
\end{lemma}
\begin{proof}
Let the leaf node $l_{v_k}$ of $BHT(v_j, v_{n'})$ corresponds to $v_k$. 
Let $\cW (d)$ be the wavefront that caused the initiation of bunch $B(v_j, v_{n'})$.
Suppose $w(v_k)$ is struck by $w(v_{k-1})$ when the wavefront is $\cW (d')$.
Then $d'$ is the $d(s, v_k)$.
The shortest path from a point $p$ on $w(v_k)$ to $s$ consists of line segment $pv_k$, edges along $v_k$ to $v_{k-1}$, $v_{r-1}$ to $v_{r-2}$, \ldots, $v_{j+1}$ to $v_z$, and the tangent from $v_z$ to $root.tangentstart$, including the shortest paths computed from $v_{zout'}$ in similar fashion, until reaching $s$.
\end{proof}

Each internal node $t_v$ stores the maximum $wpupdate$ of its children.
Since each leaf node $l_{v_k}$ stores the {\it negated Euclidean distance} along $CC$ between $v_j$ and $v_k$, a negative $wpupdate$ at $t_v$ indicates that all wavefront segments stored in the leaves of the subtree rooted at $t_v$ are invalid. 
\hfil\break

The important data members of a $BHT(v_j, v_{n'})$ are recapitulated w.r.t. a leaf node $l_{v_k}$, an internal node $t_v$, and the root $r$:
\hfil\break

\begin{tabular}{lp{3.7in}}

$l_{v_k}.wpupdate$ & Negated Euclidean distance along $CC$ between $v_j$ and $v_k$ \\

$t_v.wpupdate$ & maximum $wpupdate$ of its children \\

$r.tangentstart$ & refers to $v_{zout}'$, given that $w(v_{zout}')$ is the wavefront segment that struck $v_j$ to cause $B(v_j, v_{n'})$ \\

$r.shortestdist$ & stores $d(s, v_j)$

\end{tabular}

\paragraph{Initialization}  \hfil\break

Let a vertex $v_z$ of $CC$ is struck by a wavefront segment $w(v_{zout}')$ such that the line $v_{zout}'v_z$ is a tangent to $CC$.
The following cases needs to be considered:
\begin{enumerate}[(1)]  \itemsep -2pt
	\item Neither $w(v_z)$ nor any $w(v_k)$ from $CC$ is initiated. 
	\item $w(v_z)$ is not initiated but there exists a $w(v_k)$ for some $v_k \in \{v_{z+1}, \ldots, v_{n'}\}$.
	\item $w(v_z)$ is a valid segment in some $BHT$.
	\item $w(v_z)$ is an invalid segment in some $BHT$.
\end{enumerate}

\begin{figure}
\centerline{
\subfigure[Case (1)]{\epsfysize=240pt \epsfbox{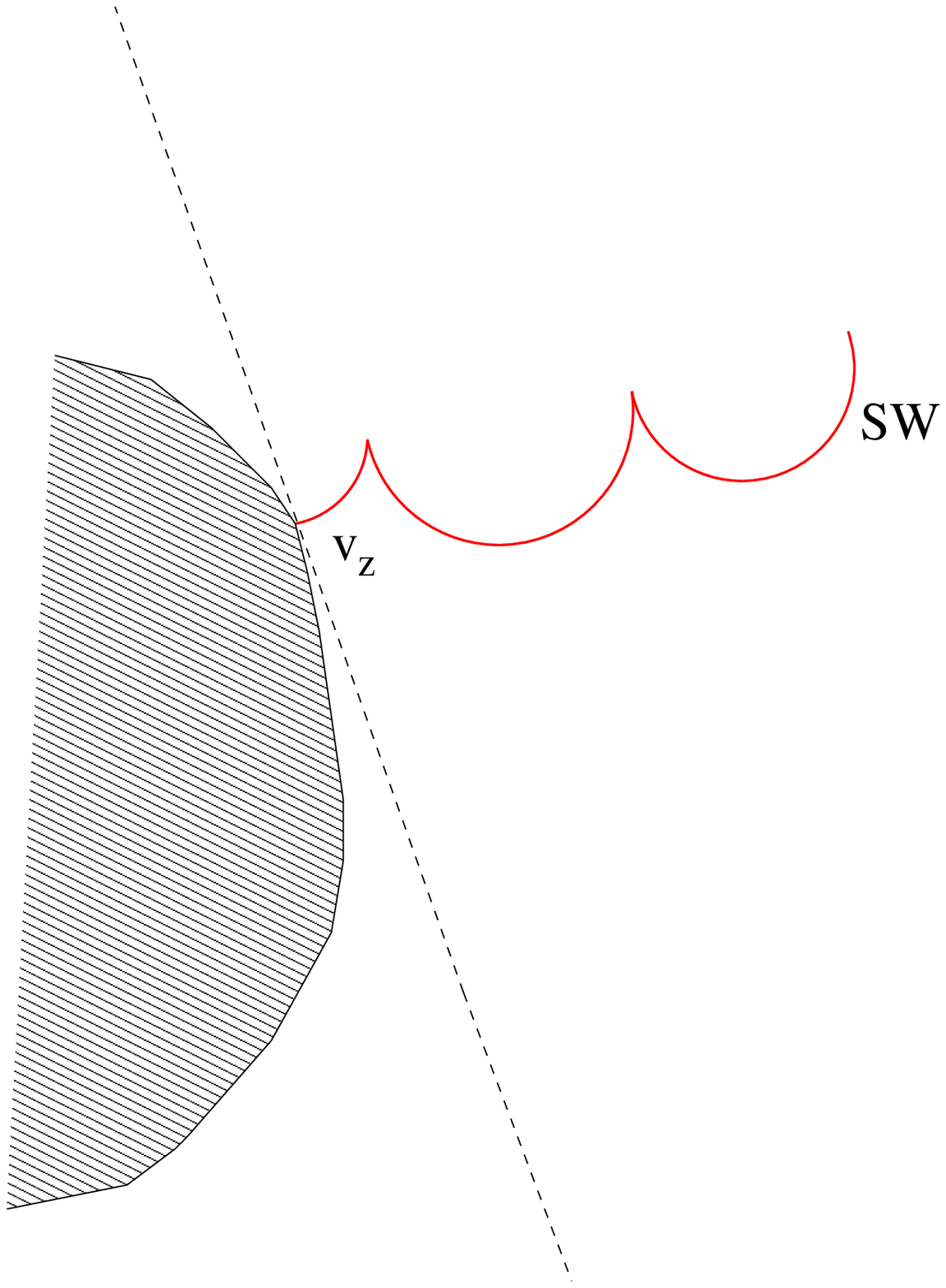}}
\subfigure[Case (2)]{\epsfysize=240pt \epsfbox{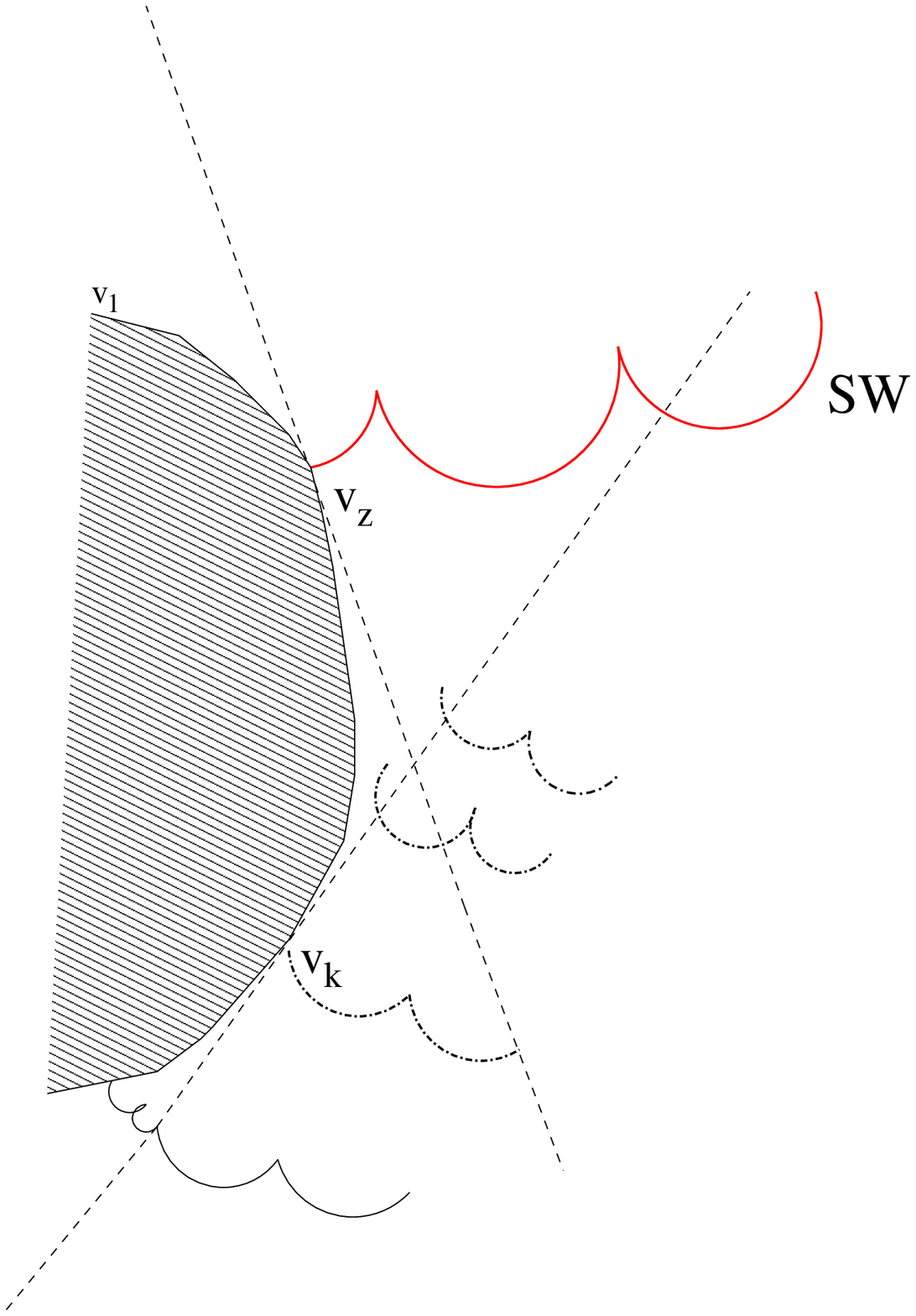}}
}
\centerline{
\subfigure[Case (3)]{\epsfysize=240pt \epsfbox{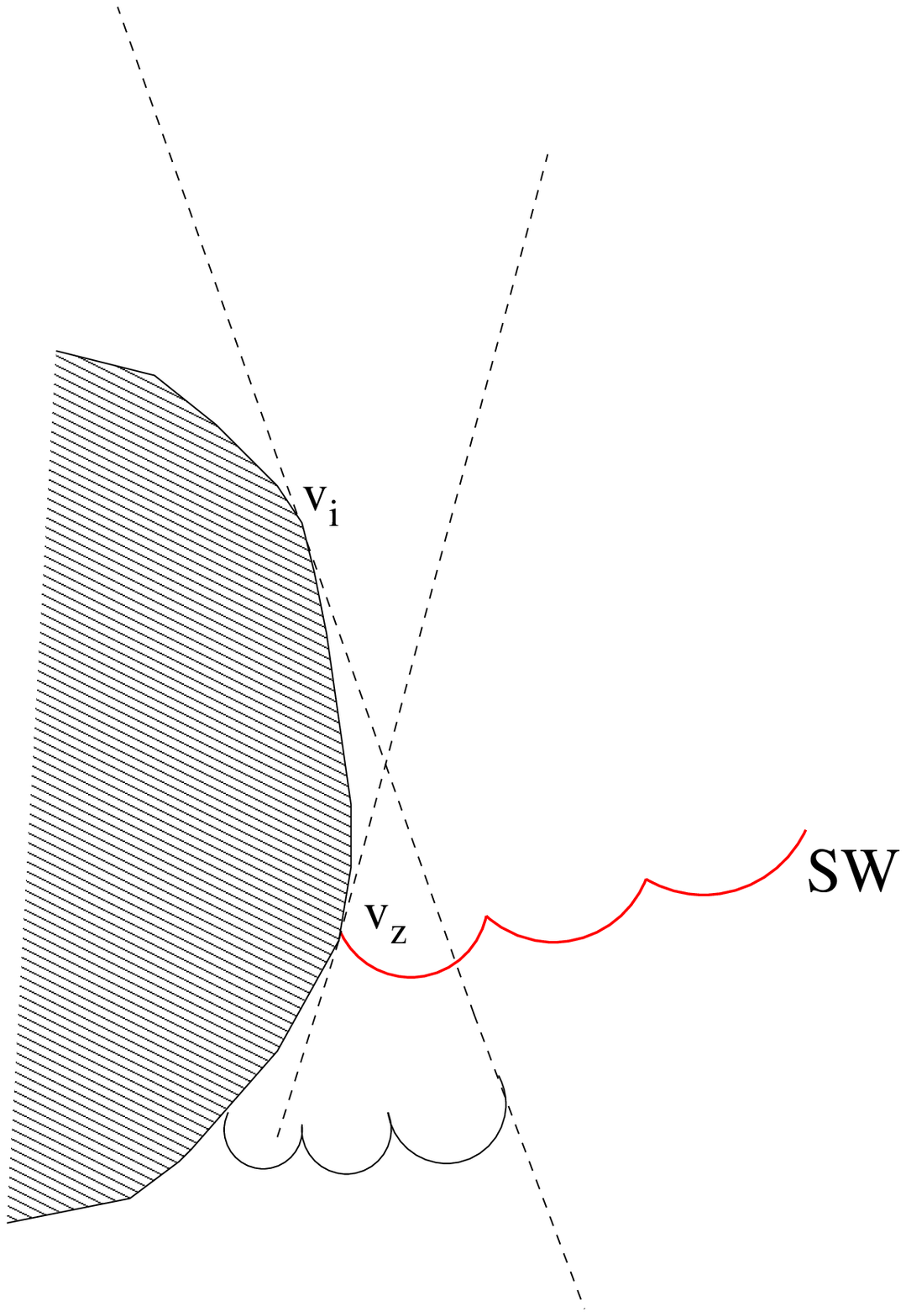}}
\subfigure[Case (4)]{\epsfysize=240pt \epsfbox{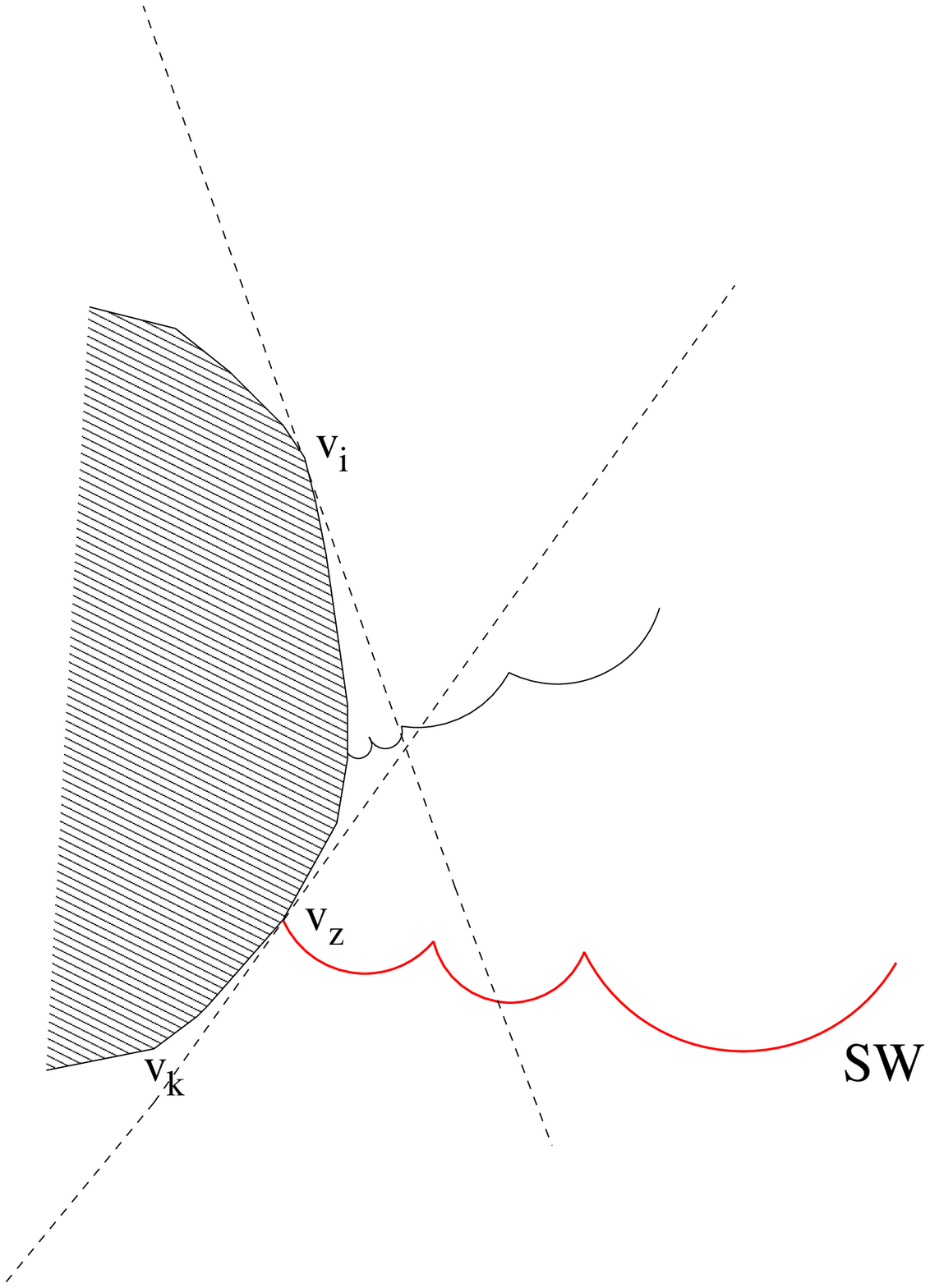}}
}
\caption{\label{fig:initbht} $BHT$ Initialization}
\end{figure}

For Case (1) (see Fig. \ref{fig:initbht}(a)), we create a new $BHT(v_z, v_{n'})$.
For the Case (2) (see Fig. \ref{fig:initbht}(b)), we neither need to create nor modify any $BHT$s.  
Let $SW'$ be the section of wavefront that struck $v_k$, which in turn caused a bunch.
Since the sink vertex $t$ is in its own corridor, progressing $SW$ and/or a 
bunch initiated from $v_z$ cannot reach this corridor without crossing $SW'$; 
but, the non-crossing nature of shortest paths (Lemma \ref{lem:spnoncrossing})
removes this case.
We do nothing for the Case (3) (see Fig. \ref{fig:initbht}(c)); the argument is same as Case (2).
For the Case (4) (see Fig. \ref{fig:initbht}(d)), let $v_z$ belongs to $BHT(v_i, v_k)$, whose leaves are $l_i, l_{i+1}, \ldots, l_z\hspace{-0.06in}=\hspace{-0.06in}l_j, l_{j+1}, \ldots, l_k$.
We split $BHT(v_i, v_k)$ into two $BHT$s, $BHT(v_i, v_{z-1})$ and $BHT(v_z, v_k)$, using the split procedure mentioned below. 
However, since $B(v_i, v_k)$ cannot reach corridor in which $t$ resides without crossing $v_{zout}'v_z$, we delete $BHT(v_i, v_k)$. 
In this case, the split procedure essentially needs to create only $BHT(v_z, v_k)$ from $BHT(v_i, v_k)$.

\paragraph{Splitting}  \hfil\break

Consider a bunch $B(v_j, v_w)$ with wavefront segments $w(v_j), \ldots, w(l_k), \ldots, w(l_w)$ in that order.
The process of forming two bunches $B(v_j, v_{k-1}), B(v_k, v_w)$ whose wavefront segments are $w(v_j), \ldots, w(l_{k-1})$ and $w(v_k), \ldots, w(l_w)$ respectively, is termed as a {\it bunch split}.
This is required due to either of the following reasons:
\begin{enumerate}[(a)] \itemsep -2pt
\item Case (4) mentioned with the bunch initialization.
\item All the wavefront segments in bunch $B(v_j, v_w)$ are no more associated with the same corridor convex chain or enter/exit boundary.
\end{enumerate}

The former is discussed in Case (4) part of bunch initialization. 
\hfil\break

Consider the latter.
Suppose we require the procedure to form two $BHT$s: one with leaves $l_j, \ldots, l_{k-1}$, and the other with leaves $l_k, \ldots, l_w$ i.e., we wish to form $BHT(v_j, v_{k-1})$ and $BHT(v_k, v_w)$ from $BHT(v_j, v_w)$. 
Let $p_a$ be the least common ancestor (LCA) of nodes $l_j, \ldots, l_{k-1}$, and let $p_b$ be the LCA of nodes $l_k, \ldots, l_w$.
Then $p_a$ with its subtree is the $BHT(l_j, l_{k-1})$, and $p_b$ with its subtree is the  $BHT(l_k, l_w)$.
For both of these $BHT$s, we update the $wpupdate$ data member for each internal node occurring along the leftmost and rightmost branches i.e., along the split paths.
As $wpupdate$ of no leaf node is changed, $wpupdate$ of no other internal node needs to be changed.
For $BHT(l_j, l_{k-1})$, the $tangentstart$ and $shortestdist$ members are same as the root of $BHT(v_j, v_w)$.
For $BHT(l_k, l_w)$, $tangentstart$ of root refers to $v_z$ with a special flag enabled to denote that this bunch is formed during bunch splits which helps in finding shortest paths passing through $v_k$; and, $shortestdist$ stores the shortest Euclidean distance along the boundary of $CC$ from $v_j$ to $v_k$ added with $d(s, v_j)$. 

\begin{property}
\label{lem:ptangencycorr} At any $\cW (d)$, any vertex $v_l$ is a leaf in at most one $BHT$.
\end{property}

\begin{lemma}
\label{lem:bunchinittime} The bunch initiation takes $O(n)$ time.
\end{lemma}
\begin{proof}
There are at most $O(n)$ leaf nodes in a $BHT$.
At every node, we spend $O(1)$ time in initializing data members.
Hence, the Case (1) takes $O(n)$.
There is nothing to do in Cases (2) and (3).
As the Lemma \ref{lem:bunchsplittime} shows, Case (4) takes $O(\lg{n})$ time.
\end{proof}

\begin{lemma}
\label{lem:bunchsplittime} The split operation takes $O(\lg{n})$ time.  
\end{lemma}
\begin{proof}
The number of leaves in a bunch are $O(n)$.
Computing LCAs, $p_a$ and $p_b$, take $O(1)$ time. 
As $BHT$s are balanced trees, updating $wpupdate$ along the split paths take $O(\lg{n})$ time.
Updates required at the root node take $O(1)$ time.
\end{proof}

\begin{lemma}
\label{lem:numbunches} The total number of bunches at any point of execution of the algorithm are $O(m)$.
\end{lemma}
\begin{proof}
In Case (1), we create a new bunch.  
In Cases (2) and (3), we are not creating any new bunches.
In Case (4), we are deleting $T_{old}$ immediately after creating $T_{new}$.  
Therefore, at any instance at most two bunches  corresponding to a corridor convex chain are alive,  one moving towards/across one corridor enter/exit boundary, and the second towards/across the other corridor enter/exit boundary.  
Since there are $O(m)$ corridors with each having at most two convex chains, there can be at most $O(m)$ bunches at any instance during the entire algorithm.  
A split at a junction/corridor causes a bunch to divide into (at most) three bunches; once an enter/exit boundary $b$ caused a bunch to split, the same $b$ cannot cause split of any other bunch, causing $O(m)$ splits in total.
\end{proof}

\subsection{Waveform- and Boundary- Section Trees}

The computation of event points require computing the shortest distance for the wavefront to strike the boundary $\partial B$ either at a point on the corridor convex chain or at a corridor entry/exit boundary.  
The Euclidean distances need only be computed between:
\begin{enumerate} [(i)] \itemsep -2pt
\item Waveform-section $WS$ and the associated corridor convex chain or corridor enter/exit boundary $C$. 
\item Boundary-section $BS$ and the associated bunch $B$.
\end{enumerate}
To compute these distances, a naive approach would in case (i) compares the distance of each bunch in $WS$ with $C$, in case (ii) compares the distance of each corridor convex chain or corridor enter/exit boundary in $BS$ with $B$.
This would lead to a time complexity which is quadratic in number of corridors.
In order to improve the time complexity, in case (i), we consider a convex hull approximation of $WS$; in case (ii), we consider a convex hull approximation of $BS$.
The convex-hull approximation of $WS$ (resp. $BS$) is defined by the minimum area convex figure enclosing the $WS$ (resp. $BS$). 
Using the convex representation, we can efficiently determine when the wavefront next strikes the boundary.
\hfil\break

Since the convex-hull approximations are to be maintained as the wavefront changes dynamically, a hierarchical representation of the convex-hull of boundary- and waveform-sections is constructed and maintained dynamically using a balanced tree structure.
The data structures to store the waveform-sections and the boundary-sections are described in Subsections \ref{subsubsect:wstdatastr} and \ref{subsubsect:bstdatastr}.
The shortest distance computations using these trees are described in Section \ref{sect:sdcomp}.

\subsubsection{Waveform Section Tree ({\it WST})}
\label{subsubsect:wstdatastr}

\begin{figure}
\center{
\subfigure[A waveform-section]{\epsfxsize=270pt \epsfbox{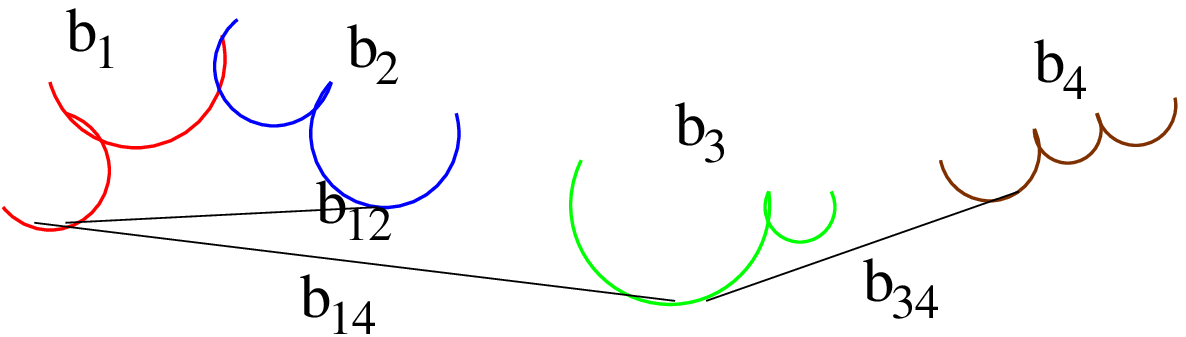}}
\subfigure[Corresponding $WST$]{\epsfxsize=140pt \epsfbox{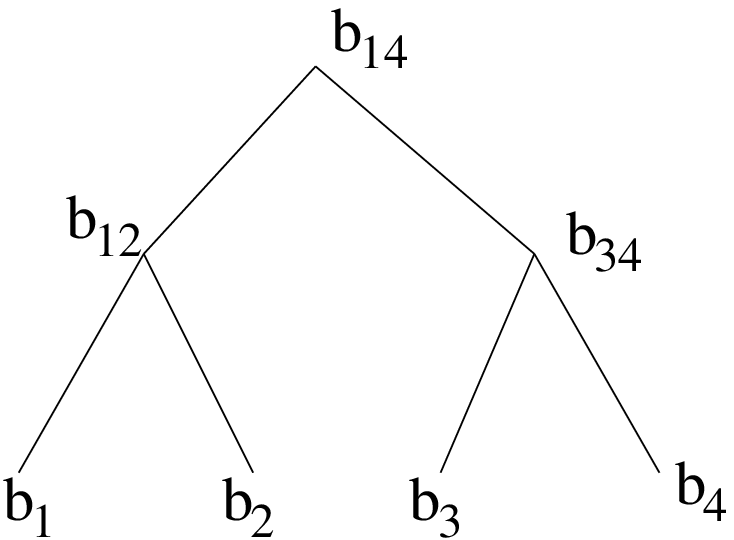}}
}
\caption{\label{fig:wstexample} A waveform-section and its $WST$}
\end{figure}

The convex hull approximation of a section of wavefront is maintained as a balanced tree.
This tree is termed as a {\it waveform-section tree, or $WST$}. 
Each leaf of a $WST$ refers to a $BHT$.
As $BHT$ is a convex approximation of a bunch, it is immediate that each leaf node refers to an upper hull.
At every internal node $p$ of a $WST$, is maintained  an upper hull of the 
wavefront segments (bunches) at the leaves of the subtree rooted at $p$.
Let $W(u)$ be the upper hull (implicitly) stored at a node $u$ of a hull tree.
Let $u$ be any internal node of $WST$.
Let $P(u)$ be the path from the root to the node $u$.
Each internal node $u$ maintains $offset(u) = d-\sum_{w \in P(u)} offset(w)$.
Whenever $offset(u) < 0$ then $W(u)$ is invalid.
See Fig. \ref{fig:wstexample}.
\hfil\break

As is standard, \cite{Overmars81},  at each internal node 
we will maintain bridges between the two hulls at the child nodes. 
Let $q, r$ be the child nodes of an internal node $p$ of $WST$. 
Let $S_q = \{b_i, b_{i+1}, \ldots, b_{j-1}\}$ and $S_r = \{b_j, b_{j+1}, \ldots, b_k\}$ be the bunches at the leaf nodes of the subtrees rooted at $q$ and $r$ respectively.
Formally, we define the bridge $br_p$ at node $p$ as a line segment joining two points, $p_j$ on a valid segment of $b_j$ and $p_k$ on a valid segment of $b_k$, such that both the $p_j$ and $p_k$ are points of tangencies.
Let $LINE(p_j, p_k)$ be the line obtained by extending the line segment $p_jp_k$ infinitely at both of its ends.
The bridge $br_p$ is defined such that the bunches $\bigcup_{l \in \{i, \ldots, k\}} b_l$, all belong to one of the closed half-planes defined by $LINE(p_i, p_k)$.
The upper hull $UH_p$ at node $p$ is computed from the upper hulls $UH_q, UH_r$ at nodes $q$ and $r$ respectively.
The information stored at internal nodes is same as in Overmars et al. \cite{Overmars81}.
\hfil\break

Let $p_j$ lies on wavefront segment $w(v_j)$ of $b_j$, and $p_k$ lies on wavefront segment $w(v_k)$ of $b_k$.
The slope of the lines $v_jp_j$ and $v_kp_k$ are saved at $p$ to facilitate in constructing bridge $br_p$ at a given wavefront propagation.
Consider the maintenance of bridges at internal nodes with the wavefront expansion.

\begin{figure}
\centerline{
\centerline{\epsfysize=200pt \epsfbox{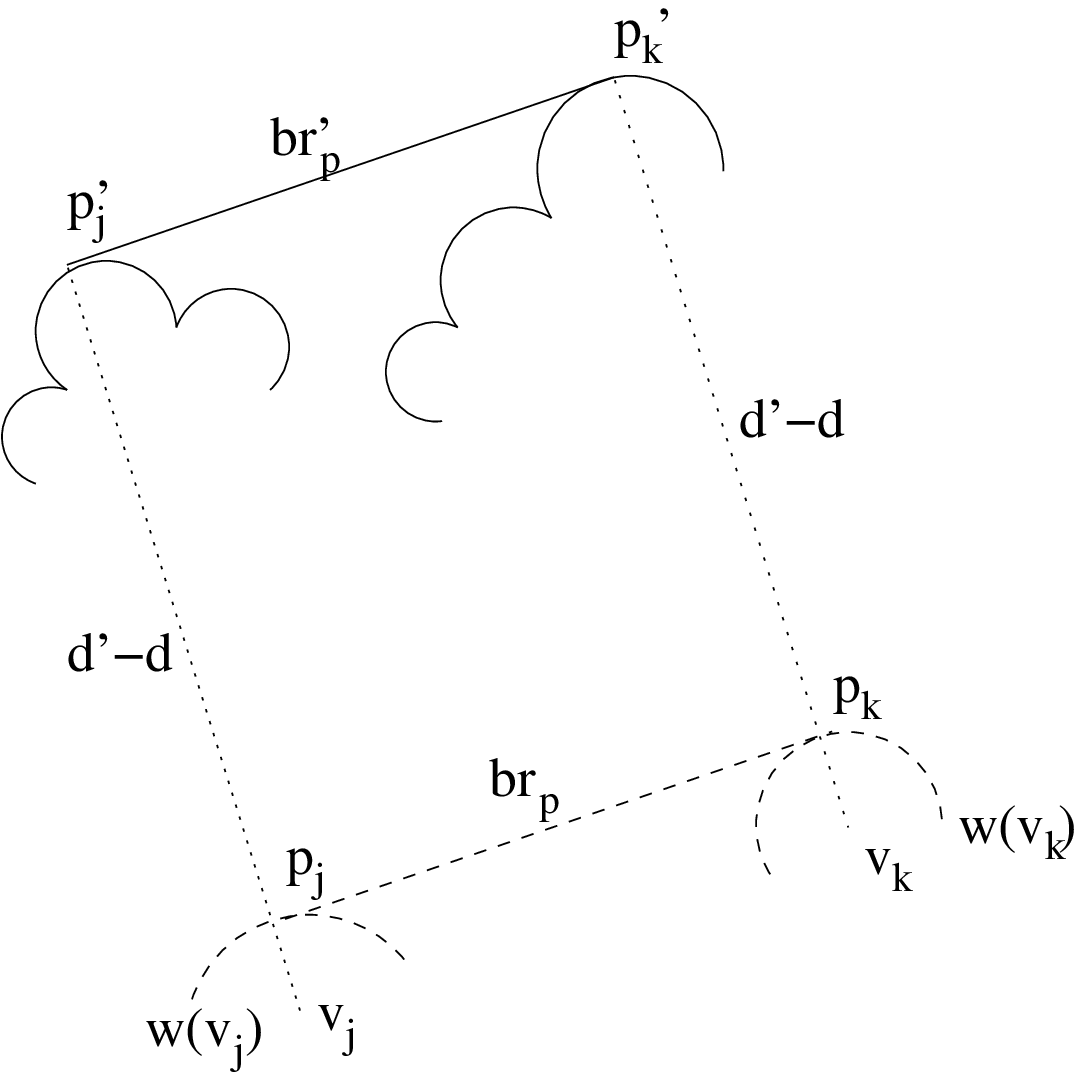}}
}
\caption{\label{fig:wstbridge} Bridge movement with the wavefront expansion}
\end{figure}

\begin{property}
\label{prop:wstbridgemovement} 
For $p_j \in w(v_j), p_k \in w(v_k)$, let $br_p=(p_j, p_k)$ be a bridge at node $p$ of a $WST$ when the wavefront is $\cW (d)$, and $p_j, p_k$ are the point of tangencies on $w(v_j)$ and $w(v_k)$ respectively.
Let $p_j'$ be a point at Euclidean distance $(|v_jp_j|+d'-d)$ from $v_j$ along the vector $\overrightarrow{v_jp_j}$.
Also, let $p_k'$ be a point at Euclidean distance $(|v_kp_k|+d'-d)$ from $v_k$ along the vector $\overrightarrow{v_kp_k}$.
If $p_j', p_k' \in \cW (d')$, then the line segment $p_j'p_k'$ is a bridge $br_p'$ of $WST$ at node $p$ when the wavefront is $\cW (d')$.
See Fig. \ref{fig:wstbridge}.
\end{property}

As the wavefront expands, the bridge $br_p$ moves as fast as any point on a bunch $b_r$, for $i \le r \le k$, expands.
Therefore, the traversal of $p_j$ (or $p_k$) could happen due to some bunch $b_r$, where $r < i$ or $r > k$ i.e., $b_r$ would be located in a subtree other than the one rooted at $p$.  
We handle the bridge maintenance in this case by introducing the notion of dirty bridges in Section \ref{sect:dirtybridges}. 

\paragraph{Primitive Operations on $WST$} \hfil\break

In the Overmars et al. \cite{Overmars81} structure each leaf node represents a point.
However, our data structure contains a bunch at the leaf node.
Our data structure operations involve insertion of bunches, deletion of bunches, tree splitting, and merging of trees.
Since all the points defining a bunch are contiguous geometrically and do not overlap with other bunches stored in $WST$, the time complexity to insert/delete a $BHT$ from $WST$ is upper bounded by the time complexity to insert/delete a point from Overmars structure.
Although the bunch itself is maintained as a hull tree, note that we never need to traverse $BHT$ while inserting/deleting/splitting/merging hull trees in $WST$.
At each node of $WST$, we maintain the same information as in Overmars et al. \cite{Overmars81} structure.  
\hfil\break


\hfil\break

Splitting a convex hull tree is performed on the balanced tree as specified in Preparata et al. \cite{Prep85}.
The offset at each node on the split path and in the two trees can be computed along the path. 
The combination of two convex hulls is computed by the bridge construction procedure, taking offset of the hulls into account.
\hfil\break

Merging two trees is performed by finding the height in the tree at which the two trees are to be merged. 
The offset of the two hulls to be merged at that node can be computed and bridge construction follows the procedure in Preparata et al. \cite{Prep85}.

\begin{lemma}
\label{lem:wsttimeandspace}
The upper hull tree of a set of bunches can be maintained dynamically at the worst-case cost of $O((\lg{m})(\lg{n}))$ per insertion/deletion/merge/split. 
And, the data structure uses $O(n)$ space.
\end{lemma}
\begin{proof}
>From Lemma \ref{lem:numbunches}, the number of bunches at the leaves of the hull tree at any stage of the algorithm are $O(m)$.
The analysis is same as given in Preparata et al. \cite{Prep85}, except that we need to spend $O(\lg{m})$ to locate the leaf node $l$ at which we are interested in inserting/deleting a bunch.
And for constructing bridges at any node along the path from the from $l$ to root, we spend $O(\lg{n})$ time. 
\end{proof}

\subsubsection{Boundary Section Tree ({\it BST})}
\label{subsubsect:bstdatastr}

\begin{figure}
\center{
\subfigure[A boundary-section]{\epsfxsize=270pt \epsfbox{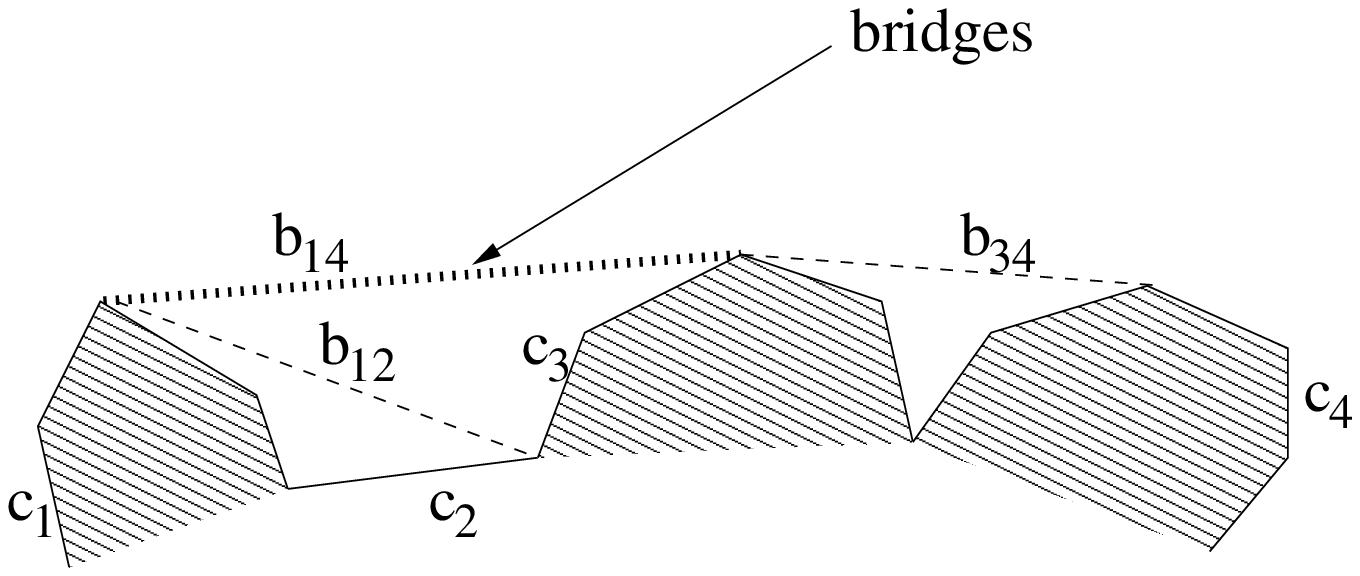}}
\subfigure[Corresponding $BST$]{\epsfxsize=140pt \epsfbox{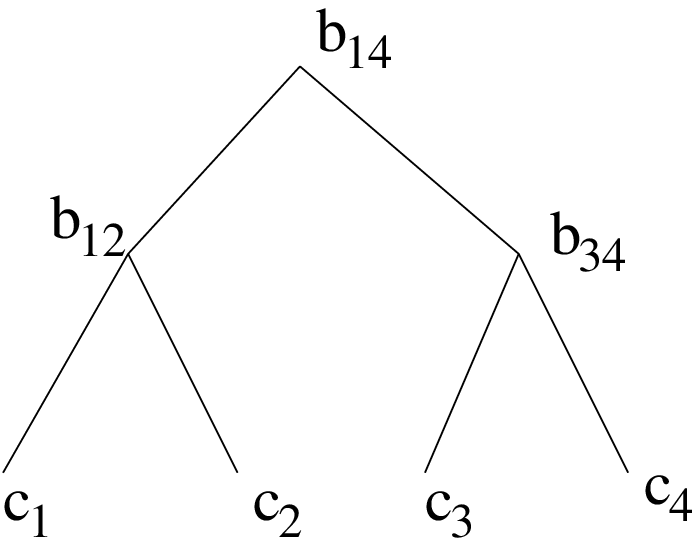}}
}
\caption{\label{fig:bstexample} A boundary-section and its $BST$}
\end{figure}

The convex hull approximation of a section of boundary is maintained as a balanced tree.
This tree is termed as a {\it boundary-section tree or $BST$}.
Each leaf of a $BST$ refers to an upper hull of a corridor convex chain or enter/exit boundary.
Every node $p$ of a $BST$ is an upper hull of the corridor convex chains or enter/exit boundaries stored at the leaves of the subtree rooted at $p$.
The bridges at internal nodes are defined same as in the case of $WST$.
See Fig. \ref{fig:bstexample}.
We remember the endpoints of bridges as in Overmars et al. \cite{Overmars81}. 
\hfil\break

In Overmars et al. \cite{Overmars81} structure, each leaf node contains a point.
However, our data structure contains an upper hull at each leaf node (node that enter/exit boundary is a degenerate upper hull).
Our data structure operations involve insertion and deletion of upper hulls.
Since all the points defining any upper hull are contiguous geometrically and do not overlap with other upper hulls stored at leaf nodes of a $BST$, the time complexity to insert/delete an upper hull from $BST$ is upper bounded by the time complexity to insert/delete a point from Overmars et al. structure.  
At each node of $BST$, we maintain the same information as in Overmars et al. structure.  

\begin{lemma}
\label{lem:bsttimeandspace}
The upper hull tree of a contiguous set of corridor convex chains or enter/exit bounding edges can be dynamically maintained at the worst-case cost of $O((\lg{m})(\lg{n}))$ per insertion/deletion/merge/split. 
And, the data structure uses $O(n)$ space.
\end{lemma}
\begin{proof}
Since the number of upper corridor convex chains and enter/exit boundaries are $O(m)$, the proof is same as in Lemma \ref{lem:wsttimeandspace}. 
\end{proof}

\subsection{Boundary Cycle List ({\it BCL})}
\label{subsect:bcldatastr}

To help in boundary splits, all the boundary cycles are saved in a data structure.   
We use a doubly-linked circular list, termed as {\it Boundary Cycle List $L_{BC}$} to denote each boundary cycle $BC$.
Every node in $L_{BC}$ refers to either a $BST$ or $WST$.
Let $BS$ be a contiguous boundary in $BC$. 
Let $ST_1$ is either a $BST$ or $WST$, and $ST_2$ is another $BST$ or $WST$. 
Let the sections of contiguous boundaries in $ST_1$ and $ST_2$ respectively be $BS_1$ and $BS_2$. 
Also, let $BS_1 \cup BS_2 = BS$.
For every two such sections of boundary $BS_1$ and $BS_2$, there exists two nodes in $L_{BC}$ that are adjacent and referring to $ST_1$ and $ST_2$.
\hfil\break

Each event point corresponding to the strike of wavefront with the boundary, pushes to event heap a reference to an entry $p$ in list $L_{BC}$, so that $p$ refers to a $BST$ or $WST$ involved in the occurrence of this event.
New $BST$s or $WST$s can be inserted/deleted to this list in $O(1)$ time as we refer to the node at which change is happening. 
Also when boundary splits, to partition one BCL into two takes $O(1)$ time only as the computations are local.
>From Lemma \ref{lem:numbunches}, the space requirement is $O(m)$.

\subsection{Corridors and Junctions}
\label{subsect:ccjdatastr}

The partitioned polygonal domain $\cP'$ into corridors and junctions is modeled as a graph $G(V, E)$.
The set $V$ of vertices in $G$ comprise the set of junctions in $\cP'$, and the set $E$ of edges in $G$ comprise the set of corridors in $\cP'$.
Let $e \in E$ be an edge in $G$ corresponding to a corridor $C$ in $\cP'$, and let $v \in V$ be a vertex in $G$ corresponding to a junction $J$ in $\cP'$.
Then $v$ lies on $e$ if and only if $C$ and $J$ share an edge in $\cP'$.
In other words, two vertices $v_1, v_2 \in V$ are adjacent in $G$ whenever there is a corridor that is adjacent to junctions corresponding to these vertices.
Note that the maximum degree of $G$ is three.
\hfil\break

Each edge in $G$ refers to a corridor, which comprise at most two convex chains and two enter/exit boundaries.
Each of these convex chains is stored in an array data structure, which is envisaged as a tree.

\section{Shortest Distance Computations}
\label{sect:sdcomp}

\subsection{$BST$ and its Association}
\label{subsect:bstsdcomp}

The shortest Euclidean distance between a boundary-section $BS$ and its associated bunch $B$ is the minimum amount of wavefront expansion required for some segment within the bunch $B$ to strike a point in the boundary-section $BS$.
The computation starts at the root node of $BST$ corresponding to boundary-section $BS$.
\hfil\break

Let $\{l_i, \ldots, l_j, \ldots, l_k\}$ be the upper hull vertices at node $l_p$, and let $\{r_{i'}, \ldots, r_{j'}, \ldots, r_{k'}\}$ be the upper hull vertices at node $r_p$ so that $l_p$ and $r_p$ are the left and right children of node $p$.
Also, let the line segment $l_jr_{j'}$ is the bridge at node $p$.
Then we call the hulls formed by the ordered list of vertices $\{l_i, \ldots, l_j\}$, $\{r_{j'}, \ldots, r_{k'}\}$ as the outer hulls at nodes $l_p$ and $r_p$ respectively.
Similarly, we call the hulls formed by the ordered list of vertices $\{l_j, \ldots, l_k\}$, $\{r_{i'}, \ldots, r_{j'}\}$ as the inner hulls at nodes $l_p$ and $r_p$ respectively.
The definitions of outer and inner hulls are applicable to (implicit) hulls stored at any node of $BST$, other than the root.
\hfil\break

Starting from the root node of $BST$, we recurse on its nodes as described below.
Consider the shortest distance computation involved at node $p$ of $BST$.
\hfil\break

(a) Suppose that bunch $B$ associated with $BST$ does not intersect with the upper hull $UH_p$ at $p$.
Then the shortest distance from $B$ is to either of the following:
\begin{enumerate}[(1)] \itemsep -2pt
	\item \label{item:brisd} $br_p$ 
	\item \label{item:brilp} outer hull at node $l_p$
	\item \label{item:brirp} outer hull at node $r_p$
\end{enumerate}
We compute the shortest distance between the vertices of $B$ and each of these. 
In Case (\ref{item:brisd}), we found the required shortest distance.
In Case (\ref{item:brilp}), we recurse on outer hull at node $l_p$; in Case (\ref{item:brirp}), we recurse on outer hull at node $r_p$.
The recursion bottoms whenever the shortest distance is to a bridge or to a corridor convex chain or to an enter/exit boundary.
Suppose recursively traversing $BST$ lead the shortest distance to happen at a leaf $l$, which refers to a corridor convex chain.
Since we stored each corridor convex chain as a tree, we continue recursing over the tree referred by leaf $l$ until we find an edge with which we can compute the shortest distance from $B$.  
The case in which leaf refers to a corridor enter/exit boundary, is similar. 
\hfil\break

(b) Suppose that bunch $B$ associated with $BST$ does intersect with the upper hull $UH_p$ at $p$. 
Then $B$ intersects with either of the following:
\begin{enumerate}[(1)] \itemsep -2pt
	\setcounter{enumi}{3}
	\item \label{item:brisd2} $br_p$  
	\item \label{item:brilp2} outer hull at node $l_p$
	\item \label{item:brirp2} outer hull at node $r_p$
\end{enumerate}
Computing shortest distances from valid segments in $B$ to each of these, we determine which among these are traversed. 
In Case (\ref{item:brisd2}), the bridge $br_p$ is split to recurse on the inner hulls at $l_p$ and $r_p$. 
In Case (\ref{item:brilp2}), we recurse on the outer hull at $l_p$; in Case (\ref{item:brirp2}), we recurse on the outer hull at $r_p$.
\hfil\break

Now consider the sub-procedure required in (a) and (b): finding the shortest distance $d$ between bunch $B$ and a line segment $L$ (or a degenerate point $p$).
To this end, we exploit the unimodal property of shortest distance between the wavefront segments in $B$ and $L$. 
We do binary search over the valid wavefront segments in $B$ to find $d$, which is similar to computing the shortest distance between two convex chains.

\begin{lemma}
\label{lem:sdcbt}
The shortest distance between the bunch $B$ and a bridge or exit boundary located at a node of $BST$ or an edge located at a leaf node of a corridor convex chain (which is a leaf node of $BST$) can be computed in $O((\lg{m})(\lg{n}))$ time amortized over the number of bridges in all the $BST$s. 
\end{lemma}
\begin{proof}
Since there are $O(m)$ useful corridors and corridor boundaries together, there are $O(m)$ elements in all the $BST$s together at any point of execution of the algorithm.
Hence, every $BST$ is of $O(\lg{m})$ depth during the entire algorithm.
\hfil\break

Finding the shortest distance between the bunch $B$ and a bridge/edge takes $O(\lg{n})$ time by doing binary search over the possible $O(n)$ intra-bunch I-curves.
\hfil\break

Suppose that there is no change in both $BST$ and $B$.
Then no case mentioned under (a) and (b), takes more than $O((\lg{n})(\lg{m}))$ time, as in the worst case we need to traverse $BST$ along its depth doing binary search over the I-curves of $B$ during this traversal.
In Case (\ref{item:brisd2}), we may traverse the inner hulls of both the nodes at every stage.
But that traversal can be charged to the bridge that was intersecting $B$.
Once a bridge is split it won't split again, as each bridge is treated similar to an edge being struck by the wavefront.
\hfil\break

Suppose the shortest distance from bunch $B$ is found to a leaf node of $BST$, which is a corridor convex chain, $CH$.
We can compute the shortest distance between $BHT$ and $CH$ in $O(\lg{n})$ time.
The exit boundary stored at a leaf node of $BST$ is treated in the same way  as a bridge of $BST$.

\end{proof}

\subsection{$WST$ and its Association}
\label{subsect:wstsdcomp}

The shortest Euclidean distance between a waveform-section $WS$ and its associated corridor convex chain or enter/exit boundary $C$ is the minimum amount of wavefront expansion required for some wavefront segment within the waveform-section $WS$ to strike $C$.
\hfil\break

Let us consider the case in which $WS$ is associated with a corridor convex chain $C$.
Let the tree corresponding to $C$ be $HT$.
Since the intra-bunch I-curves diverge, inter-bunch I-curves in the given waveform-section play a major role.  
The computation starts at the root node of $WST$ corresponding to the waveform-section $WS$. 
At a node $p$ of $WST$ if corridor convex chain $C$ does not intersect the bridge at $p$, then the shortest distance from $C$ to $WS$ is either to the bridge at $p$ or lies to an upper hull at its left/right child subtree root.
In the first case, the shortest distance is noted as the distance between the bridge and the corridor convex chain $C$.  
If the convex chain stored at left (resp. right) descendant node is nearer, recurse on the left (resp. right) subtree.  
In the case that the corridor convex chain $C$ intersects the bridge at $p$ of $WST$, the bridge is split to find the shortest distance between the $C$ and the upper hulls located at the subtrees of $p$.  
To consider the case in which some point $q \in C$ lies inside an upper hull $ch_p$ stored at node $p$ of $WST$, we denote the distance to a bridge of $WST$ from $q$ as positive (resp. negative) whenever $q$ lies outside (resp. inside) $ch_p$.  
Let $br$ be the closest bridge to $C$ located at a leaf node of $WS$, and this leaf node corresponds to a bunch $B$ with bunch hull tree as $BHT$.  
Then, the computation recurse on $BHT$ similar to the traversal of $WST$ explained above, considering only bridges over the valid segments in bunch $B$.  
This is accomplished by noting that for a non-root node $p$ of $BHT$, the negative $root.wpupdate+p.wpudate$ indicates that the rooted subtree with $p$ as root has only invalid segments as leaves.
\hfil\break

The other case in which the waveform-section $WS$ is associated with an exit boundary $e$ is same as the computation where the corridor convex chain $CH$ is having only one edge.
\hfil\break

Suppose the computed shortest distance for a waveform-section $WS$ to reach its associated corridor convex chain or enter/exit boundary $C$ is represented as an event point $evt$ in the event heap.
Let two inter-bunch I-curves of waveform-section $WS$ intersect before the $evt$ occurs.
This require updates to both the $WST$ and the shortest distance, which are attached as satellite data associated with $evt$. 
However, the intersection of two inter-bunch I-curves results in a Type-IV event, which in turn takes care of these updates.

\begin{lemma}
\label{lem:sdcwt} 
The shortest distance between the corridor convex chain or enter/exit boundary $C$ associated with the waveform-section $WS$, and a bridge/segment located at a node of $WST$ (or, a node of $BHT$ in $WST$) can be computed with $O((\lg{m})(\lg{n}))$ time complexity amortized over the number of bridges in all $WST$s.
\end{lemma}
\begin{proof}
We analyze the case in which $WS$ is associated with a corridor convex chain; the other case when $WS$ is associated with an exit boundary is similar to this.
According to the definition of $WST$, the waveform is formed with bunches at the leaves.
The total number of bunches at any point of execution of the algorithm are $O(m)$ (from Lemma \ref{lem:numbunches}).  
Hence, any $WST$ is of $O(\lg{m})$ depth during the entire algorithm.
For a bridge $br$ at node $p$ of $WST$, we need to compute the shortest distance between $C$ and $br$, between $C$ and the two upper hulls implicitly stored at $p$'s immediate child nodes to find the closest bridge.  
Since $WST$ is of $O(\lg{m})$ depth and the tree corresponding to $C$ is of $O(\lg{n})$ depth, every shortest distance computation takes $O((\lg{m})(\lg{n}))$.
This complexity includes the computation required in determining whether the closest point $pt \in C$ lies inside/outside the upper hull stored at node $p$.
Once a bridge $br$ in $WST$ is determined as closest to $C$, $br$ will never be considered again.  
\hfil\break


We next consider the complexity of computing shortest distances' within bunch hull trees.
Let $BHT$ is a bunch hull tree located at a leaf node $p_l$ of $WST$.
Also, let the bridge at node $p_l$ is determined as closest to a bridge $br'$ in $HT$ among all the bridges that are not struck in $HT$.  
Then we need to traverse the bridges over the valid segments of that $BHT$; similar to the traversal of $WST$ bridges, this is of $O(\lg{n})$ amortized complexity.  
\end{proof}

\section{IIntersect Procedure}
\label{sect:IIntersectProc}

Let $SB = \{S_1, S_2, \ldots, S_k\}$ be the sequence of contiguous sections of boundary such that $S_i$ represents either a boundary-section associated to a bunch, or a corridor convex chain (or enter/exit boundary) associated with a waveform-section. 
The IIntersect procedure finds the corridor convex chain or enter/exit boundary in $SB$ with which a given I-curve($w(c_1), w(c_2))$ first intersects.
We find the intersection of I-curve($w(c_1), w(c_2)$) with each of $S_1, S_2, \ldots, S_k$ in that order (the cost is amortized).
If the I-curve($w(c_1), w(c_2)$) does not intersect with $SB$, the procedure returns the same.
Since an I-curve is not a straight-line, it may intersect a $SB$ multiple times.  
However, we are interested in the first point of intersection along the given I-curve($w(c_1), w(c_2)$) with $SB$.
\hfil\break

First, consider finding the intersection of I-curve($w(c_1), w(c_2)$) with $S_i \in SB$.
Let $HT$ be the hull tree of $S_i$.
We denote the upper hull at a node $p$ of $HT$ with $H_p$, outer upper hull at $p$ with $OH_p$, and the inner inner hull at $p$ with $IH_p$. 
The computation starts at the root node of $HT$.
Consider the computation at a node $p$ of $HT$.
Let $q, r$ be the left and right children of $p$. 
The bridge $br_p$ at $p$ be $p_1p_2$. 
The following are the possible places at which I-curve($w(c_1), w(c_2)$) intersects: 
\begin{enumerate}[(1)] \itemsep -2pt
	\item \label{one} $OH_{r}$
	\item \label{two} $OH_{q}$
	\item \label{three} $br_p$
	\item \label{four} None 
\end{enumerate}

>From the definition of I-curve, if I-curve intersects bridge $br_p$ then $p_1$ (resp. $p_2$) is closer to $c_1$ and $p_2$ (resp. $p_1$) is closer to $c_2$.  
This fact is used in differentiating Case (\ref{three}) from the first two Cases.
Let the angle $c_1p_1$ (resp. $c_2p_1$) makes at $c_1$ (resp. $c_2$) is less than the angle made by $c_1p_2$ (resp. $c_2p_2$) at $c_1$ (resp. $c_2$).   
The conditions in which Case (\ref{one}) (resp. Case (\ref{two})) occur is: $|c_1p_1|<|c_2p_1|,|c_1p_2|<|c_2p_2|$ and the line segment $c_1p_1$ (resp. $c_2p_1$) is in the left (resp. right) half-plane defined by the line $c_2p_2$ (resp. $c_1p_2$); or, $|c_1p_1|>|c_2p_1|,|c_1p_2|>|c_2p_2|$ and the line segment $c_1p_1$ (resp. $c_2p_1$) is in the left (resp. right) half-plane defined by the line $c_2p_2$ (resp. $c_1p_2$).
The Case (\ref{three}) occur if $p_1$ is closer to $c_1$ (resp. $c_2$) than $c_2$ (resp. $c_1$).
The Case (\ref{four}) occur if we reach a leaf node and do not find the intersection point.
\hfil\break

For Case (\ref{one}) (resp. Case (\ref{two})), it is guaranteed that the I-curve does not first intersect with any of the elements located in the right (resp. left) subtree of $r$. 
So we recursively traverse left (resp. right) subtree of $r$.
For the Case (\ref{three}), we need to split the bridge $br$ to further find the intersection of I-curve with either of $IH_{q}, IH_{r}$.
Since we do not know which inner hull is intersected by the I-curve, we traverse both left and right inner hulls at $r$.
To facilitate traversing inner hulls, we determine the visible cones from both the centers $c_1$ and $c_2$ by computing tangents to inner hulls.
\hfil\break

The worst-case occurs when we need to compute the tangents to inner hulls along the depth of $HT$, leading to the time complexity of $O((\lg{m})(\lg{m}))$. 
Since the I-curve intersects at most $O(\lg{m})$ bridges and no bridge splits more than once, the amortized time complexity is $O(\lg{m})$. 
Since the number of boundary- and waveform-sections are bonded by $O(m)$ i.e., since the number of elements in $SB$ is $O(m)$, finding the intersection of the given I-curve with the $SB$ amortized over splits, mergers and Type-IV intersections takes $O(\lg^2{m})$ amortized time.

\section{Merging}
\label{sect:merging}

\begin{figure} [htop]
\centerline{
\subfigure[Before merge: $c_1, \ldots c_8$ are associated with $A$, and some of them are reachable from $B$]{\epsfxsize=420pt \epsfbox{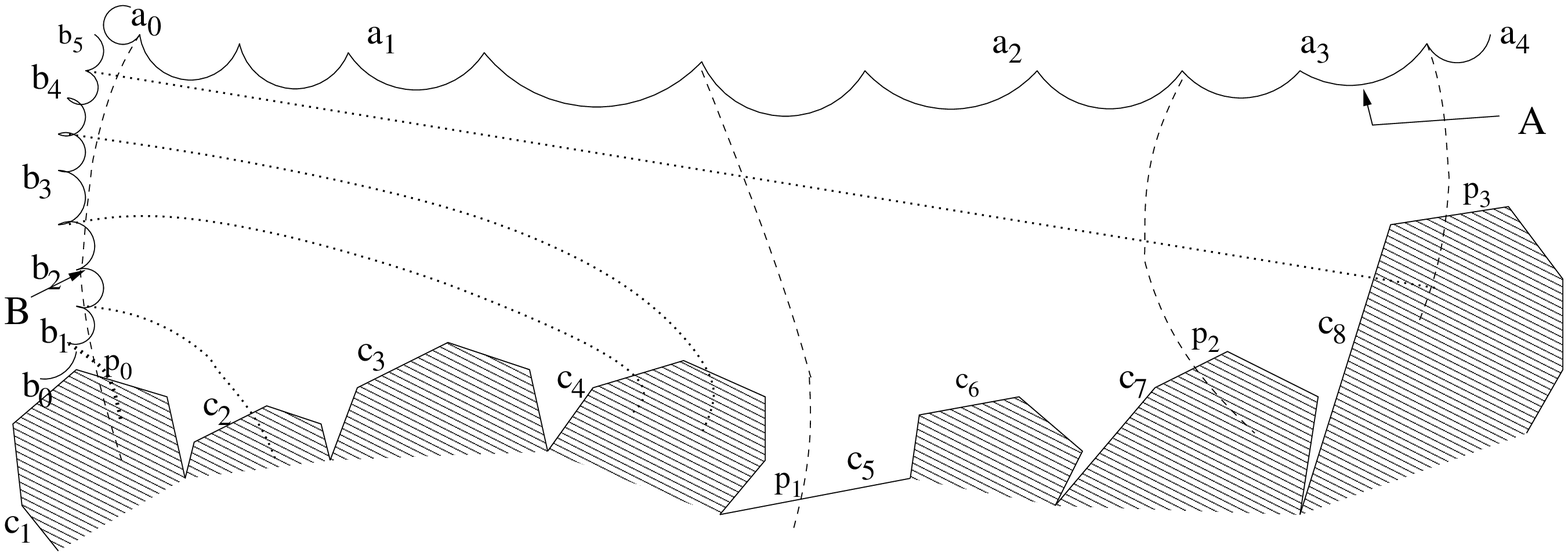}}
}
\centerline{
\subfigure[After merge: $c_1, \ldots, c_4$ are associated with $A$, and $c_4, \ldots, c_8$ are associated with $B$]{\epsfxsize=420pt \epsfbox{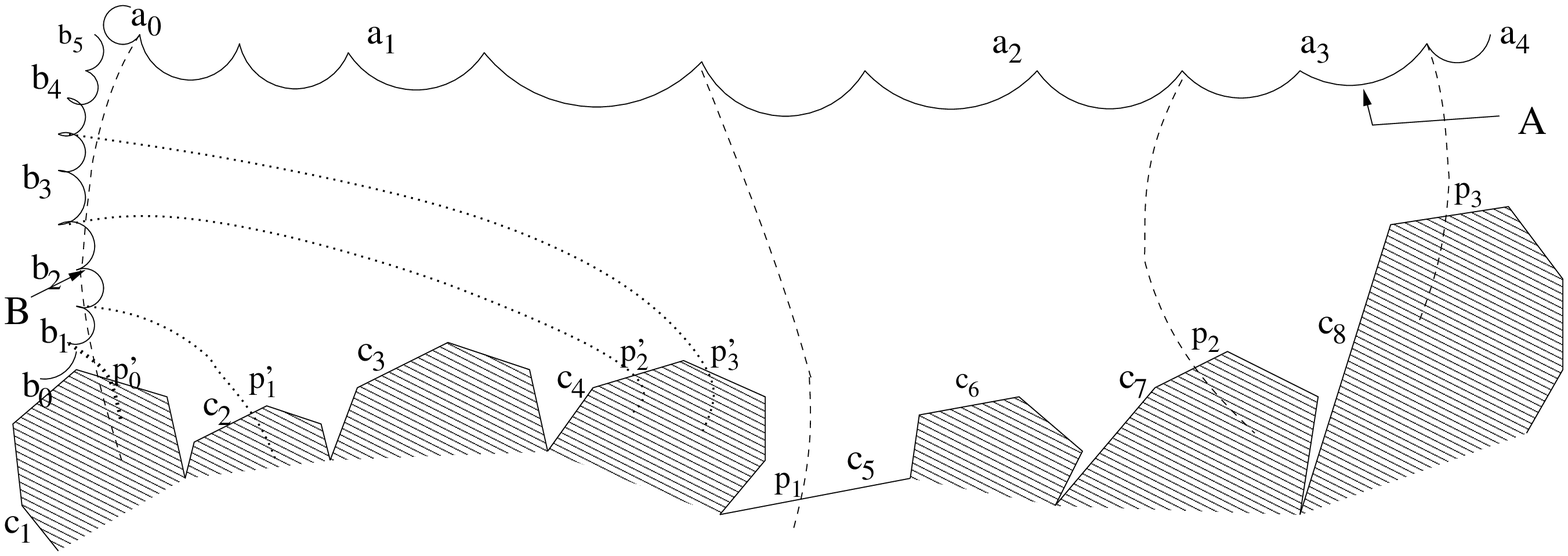}}
}
\caption{\label{fig:merge} Modification to associations of two sections of wavefronts in Merge procedure}
\end{figure}

Consider a section of boundary $SB_A$ associated with a section of wavefront $A$. 
Also, consider a section of boundary $SB_B$ associated with a section of wavefront $B$. 
Suppose $B$ struck $e$, causing the reachable edges of $B$ to include new bounding edges from $SB_A$.
Then we need to update the associations of $B$ by considering the proximity of bunches in $A \cup B$ with the corridor convex chains or enter/exit boundaries in $SB_A$ such that:
\begin{enumerate} \itemsep -2pt
	\item the sections of boundary associated with both $A$ and $B$ are contiguous,
	\item each element in $SB_A$ is associated with either $A$ or $B$,
	\item at most one element in $SB_A$ is associated with both $A$ and $B$.
\end{enumerate}
The process of updating the associations of $A$ and $B$ and initiating/updating boundary- and waveform-sections for all elements in $SB_A \cup SB_B$ is termed as the {\it merger} of sections of wavefronts' $A$ and $B$.
\hfil\break

Consider Fig. \ref{fig:merge}(a).
Here, the section of wavefront $A$ comprising bunches $a_1, \ldots, a_3$, and $SB_A$ comprise corridor convex chains or enter/exit boundaries $c_1, \ldots, c_8$.
Specifically, $BS(a_1) = \{c_1, \ldots, c_5\}$, $BS(a_2) = \{c_5, c_6, c_7\}$, and $BS(a_3) = \{c_7, c_8\}$.
Suppose an event caused part of $SB_A$ to be reachable from a waveform-section $B$ consisting of bunches $b_1, \ldots, b_5$.
The associations of $SB_A \cup SB_B$ after the wavefront merger is shown in Fig. \ref{fig:merge}(b) with inter-bunch I-curves: $BS(b_1) = \{c_1, c_2\}$, $BS(b_2) = \{c_2, c_3, c_4\}$, $WS(c_4) = \{b_2, b_3, a_1\}$,  $BS(a_1) = \{c_4, c_5\}$, $BS(a_2) = \{c_5, c_6, c_7\}$, and $BS(a_3) = \{c_7, c_8\}$.
Also, the corresponding boundary- and waveform-section trees are initiated/updated.
As mentioned below, new associations follow all three characterizations:
\begin{enumerate} \itemsep -2pt
	\item the contiguous section of boundary $c_1, \ldots, c_4$ is associated with $B$, and the contiguous section of boundary $c_4, \ldots, c_8$ is associated with $A$,
	\item every element in $c_1, \ldots, c_8$ is associated with either $A$ or $B$, and
	\item only $c_4$ is associated with both $A$ and $B$.
\end{enumerate}

The inter-bunch I-curves, bunches, and the corresponding sections of boundary associated together define Voronoi regions.
For example, before invoking the merge procedure (Fig. \ref{fig:merge}(a)), the Voronoi regions are:
I-curve($a_0, a_1$), $a_1$, I-curve($a_1, a_2$), section of boundary along $SB_A$ between $p_0$ and $p_1$; 
I-curve($a_1, a_2$), $a_2$, I-curve($a_2, a_3$), section of boundary along $SB_A$ between $p_1$ and $p_2$;
I-curve($a_3, a_4$), $a_3$, I-curve($a_3, a_4$), section of boundary along $SB_A$ between $p_2$ and $p_3$.
After invoking the merge procedure (Fig. \ref{fig:merge}(b)), the Voronoi regions are: 
I-curve($b_0, b_1$), $b_1$, I-curve($b_1, b_2$), section of boundary along $SB_A$ between $p_0'$ and $p_1$; 
etc.,;
I-curve($b_3, b_4$), $b_4, a_1$, I-curve($a_1, a_2$), section of boundary along $SB_A$ between $p_3'$ and $p_1$;
etc.,;
At most one Voronoi region is bounded by one inter-bunch I-curve from $A$ and another from $B$.
In updating the Voronoi regions, some of the inter-bunch I-curves from $A$ and $B$ are not used.
For example, although I-curve($b_4, b_5$) intersects with $c_8$, section of boundary consisting of $c_5, \ldots, c_8$ is closer to $A$ than the bunches in $B$. 
Although this example shows all the bunches in $A$ (resp. $B$) are contiguous along the wavefront, it is not a requirement.
\hfil\break

Consider any corridor convex chain or enter/exit boundary, say $c_3$ in $SB_A$, that is associated with bunch $a_1$.
For wavefront segment $w(a_1') \in a_1$, let $d'$ be the Euclidean distance between $w(a_1')$ and $c_3$.
Given the intersection of inter-bunch I-curves of $B$ with $SB_A$, we know that $c_3$ is closer to bunch $b_2$ among all the bunches in $B$.
For a wavefront segment $w(b_2') \in b_2$, let $d''$ be the Euclidean distance between $w(b_2')$ and $c_3$.
If $d'' < d'$, then the association of $c_3$ is changed to $b_2$.
\hfil\break

Herewith we describe the procedure to associate a section of $SB_A$ to bunches in $B$.
For each of the corridor convex chain or enter/exit boundary $c$ occurring in $SB_A$, the proximity of $c$ to its association in $A$ requires to be compared against the proximity of $c$ with every bunch in $B$, and re-associated when necessary.  
Before $B$ strike $SB_A$, suppose that there is no Type-IV event due to bunches in $B$.
Then the intersection points of inter-bunch I-curves from $B$ are in sorted order along $SB_A$.
For example, in Fig. \ref{fig:merge}, I-curve($b_3, b_4$) is intersecting $c_4$;
if $c_1, \ldots, c_4$ are associated with $b_1, b_2, b_3$, then the section of boundary $c_4, \ldots, c_8$ is definitely not associated with $b_1, b_2, b_3$.
This unimodal property accommodates binary search over the inter-bunch I-curves of $B$ in updating the associations.
Since the inter-bunch I-curves are higher order curves, we use the IIntersect procedure described in Section \ref{sect:IIntersectProc} in determining which element of $SB_A$ is intersected by an I-curve under consideration.
\hfil\break

The following Lemma says that there exists a contiguous section of boundary that can be associated with a section of wavefront.

\begin{lemma}
\label{lem:contigprop3} There exists an association {\cA} such that the sequence of boundary edges on a boundary cycle that are associated with a section of wavefront is a contiguous sequence.
This is known as the {\it contiguity property for sections of wavefront}.
\end{lemma}
\begin{proof}
Immediate from Lemmas \ref{lem:spnoncrossing} and \ref{lem:contigprop2}.
\end{proof}

This property says that there always exists an implicit I-curve between $A$ and $B$ that separates the contiguous section of boundary that is associated with $A$ from the contiguous section of boundary that is associated with $B$.
Suppose all the elements that occurred before $c \in SB_A$ along $SB_A$ are associated with $B$, and $c$ needs to be associated with both $A$ and $B$ (or, with $A$ only).
Then no element in $SB_A$ that occur after $c$ needs to be associated with $B$. 
This is considered in halting the binary search procedure.
\hfil\break

Here is the description of ASSOCATOB procedure.
Let the bunches in $B$ be $b_1, \ldots, b_q$.
Starting from $b_1$, the procedure finds a bunch $b_j$ in $B$ such that between the points of intersection of inter-bunch I-curve($b_{j-1}, b_j$) and the inter-bunch I-curve($b_j, b_{j+1}$), there exists a contiguous section of boundary $SB' \in SB_A$ and $|SB'| > 1$ (line 4 of ASSOCATOB).
If such an I-curve is found, $RV(b_j)$ is set to $SB'$; and, $RV(b_1), \ldots, RV(b_j)$ are set to $\phi$ (lines 5-8).
In the next iteration, the procedure starts from $b_j$.
The iterations terminate when an element in $RV(A)$ needs to be left with its current association (line 11). 
Then it adjusts the $RV$s of $A$; and, updates $BST$s and $WST$s of bunches in both $A$ and $B$ (lines 12-13).
After the change in associations, the event points corresponding to the shortest distance between boundary-sections and their associated bunches, and the shortest distance between waveform-sections and their associated corridor convex chains/exit boundaries are updated in the event min-heap.  
Also, for every corridor convex chain $C$ whose associations are changed within the procedure, similar to the procedure described in Type-I/Type-II events, we find the tangent to $C$ and push the corresponding Type-III event to min-heap (line 14).
Note that the Type-III event generated herewith could cause Case (4) of $BHT$ initialization (See Subsection \ref{subsect:bhtdatastr}.).
\hfil\break

Also, suppose $C \in SB_A$ is the corridor convex chain or enter/exit boundary encountered by the binary search procedure in ASSOCATOB, and $C$ is found to be associated with $A$.
Then the procedure does not process elements in $SB_A$ that occur after $C$.
Let $A^R$ (resp. $B^R$) denote the sequence bunches in $A$ (resp. $B$) ordered in reverse direction.
Since the procedure does not know from which end of the sequence of bunches in $B$ that it requires to start associating with $SB_A$, we need to invoke with arguments: $A, B$; $A, B^R$; $A^R, B$; $A^R, B^R$.  
Although explicitly we do not state, all these eight invocations of ASSOCATOB are considered in MERGE procedure.
\hfil\break

\begin{figure}
\centerline{\epsfxsize=370pt \epsfbox{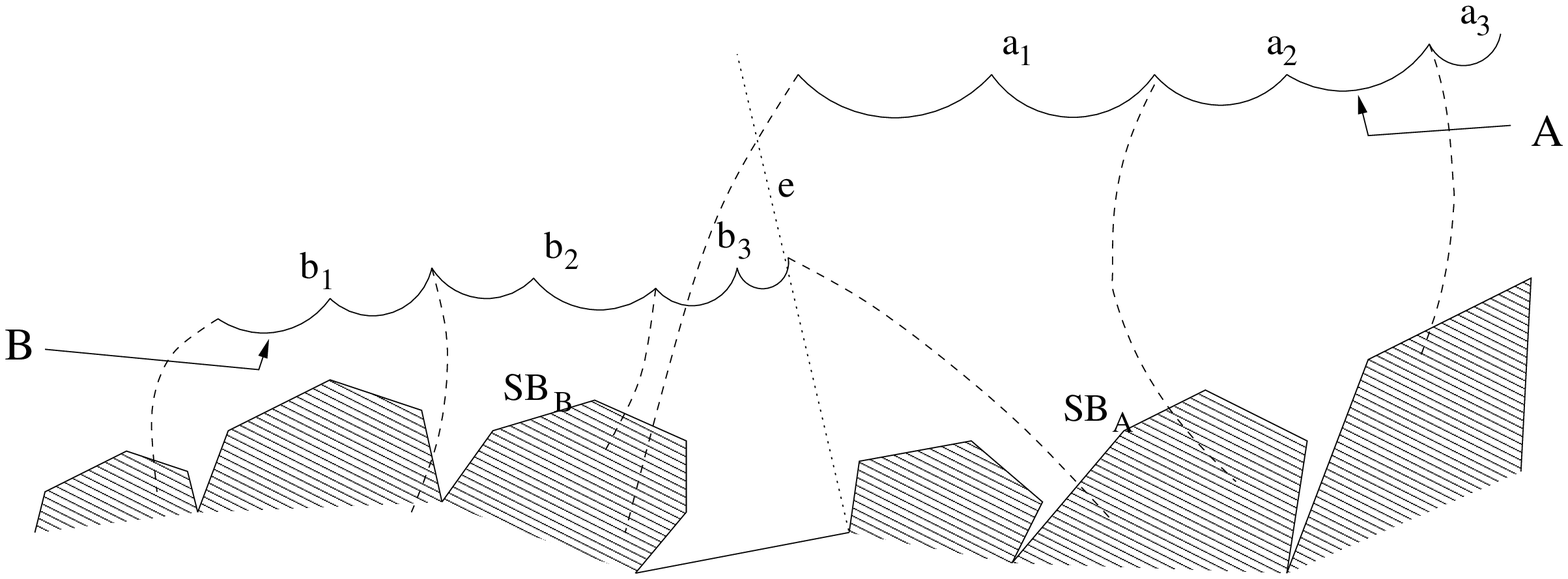}}
\caption{\label{fig:assocbtoa}Associating $SV_A \cup SV_B$ to $A \cup B$}
\end{figure}

As mentioned earlier, we invoked the merge procedure when $B$ struck an edge $e$. 
Suppose $B$ was associated with a section of boundary $SB_B' = SB_B \cup \{e\}$.
When it strikes $e$, the section of boundary $SB_B$ is reachable from $A$.
See Fig. \ref{fig:assocbtoa}.
Then the merge procedure needs to extend the three characterizations given in the beginning of this Section to $SB_A \cup SB_B$ and $A \cup B$, so that to associate $SB_A \cup SB_B$ with the bunches in $A$ and $B$.
In other words, we also require to invoke ASSOCATOB with arguments: $B, A$; $B^R, A$; $B, A^R$; $B^R, A^R$.
Considering all these possibilities, the MERGE procedure invokes ASSOCATOB eight times for merging the two given sections of wavefront.
\hfil\break

As a whole, the merger takes two sets, say $A$ and $B$, of bunches and their associations in terms of $WST$s and $BST$s as input, and outputs a set $C \subseteq A \cup B$ of bunches with initiated/updated $WST$s and $BST$s.  

\begin{algorithm}
\caption{Associates corridor convex chains or enter/exit boundaries in $RV$s of a section of wavefront $A$ to $RV$s of another section of wavefront $B$; 
also, updates $WST$s and $BST$s for the bunches in $B$}
{\bf procedure} ASSOCATOB($A, B$)
\begin{algorithmic}[1]
\REQUIRE{the re-association starts by considering the first bunch in each  of $A$ and $B$.
The bunches in both $A$ and $B$ are considered sequentially from there on; 
the procedure terminates when it finds the first convex chain or enter/exit boundary among the sequence of convex chains/exit boundaries of $RV(A)$ which cannot be associated with a bunch in $B$  
}
\vspace*{.1in}

\STATE Let \hfil\break
   $A$ : $a_1,a_2,\ldots,a_p$, a sequence of bunches \hfil\break
   $B$ : $b_1,b_2,\ldots,b_q$, a sequence of bunches \hfil\break
   $RV(A)$ : $\bigcup_{i=1}^{i=p}RV(a_i)$ = $c_1,c_2,\ldots,c_r$, contiguous sequence comprising corridor convex chains/exit boundaries in $\partial B$ (note that this sequence is not computed explicitly)
\STATE $prevb:=0$, $prevc:=0$
\REPEAT
	\STATE by doing binary search over the range $[prevb+1,q]$ find $j$, 
               where $j$ = $\min \{i:  i \in [prevb+1,q]$ and   I-curve($b_i,b_{i+1}$) 
               satisfies $c_s$ := IIntersect(I-curve($b_i,b_{i+1}),RV(A)$),
               $c_s \ne c_{prevc} \}$ and points on $c_s$ can be associated with $b_i$ 
	       instead  of $a_k$ where currently $c_{s-1}$ and $c_s \in RV(a_k)$.
	\STATE set $RV(b_{prevb+1})$,\ldots,$RV(b_{j-1})$,$RV(b_j)$ to $\phi$.
	\IF{a corridor convex chain or enter/exit boundary $c \in \{c_{prevc+1},\ldots,c_s\}$ 
            can be entirely associated with $b_j$}
		\STATE add \{$c_{prevc+1},\ldots,c_s$\} to $RV(b_j)$.
	\ENDIF
	\STATE $prevb$:=$j$, $c_{prevc}$:=$c_s$
\UNTIL{$prevb > q-1$}
\STATE remove corridor convex chains or enter/exit boundaries from $RV(A)$ so that 
       $RV(A) \cap RV(B) = \phi$.
\STATE for every bunch $b_i \in B$, if $RV(b_i)$ is modified and $|RV(b_i)| = 0$,
       and further, $b_i$ is associated to a corridor convex chain or enter/exit boundary $c$ 
       then update $WST(c)$ s.t. $WS(c)$ is the sequence of bunches associated 
       to $c$.
\STATE for every bunch $b_i \in B$, if $RV(b_i)$ is modified and $|RV(b_i)| > 1$
       update $BST(b_i)$ s.t. $BS(b_i)$ is the same as $RV(b_i)$.
\STATE for every corridor convex chain $C$ whose associations are changed, determine and push a Type-III event to min-heap. 
\end{algorithmic}
\end{algorithm}

\begin{algorithm}
\caption{Modifies $RV$s, $WST$s and $BST$s for the two given sections of the wavefront: $SW_1, SW_2$}
{\bf procedure} MERGE($SW_1,SW_2$)
\begin{algorithmic}[1]
\vspace*{.1in}
	\STATE ASSOCATOB($SW_1,SW_2$)
	\STATE ASSOCATOB($SW_2,SW_1$)
\end{algorithmic}
\end{algorithm}

\begin{theorem}
\label{lem:assocatobcorr}
The ASSOCATOB($A, B$) procedure correctly updates the associations of bunches in both the sections of wavefronts' $A$ and $B$.
\end{theorem}
\begin{proof}
Let $SB_A$ and $SB_B$ be the sections of boundary associated with bunches in $A$ and $B$ respectively.
We intend to prove three characterizations given in the beginning of this Section.
Primarily, we prove that for every bunch $b$ in $A \cup B$, section of boundary associated with $b$ is contiguous along a boundary cycle whenever $SB_A \cup SB_B$ is contiguous before invoking the procedure.
\hfil\break

The procedure assumes that both $SB_A, SB_B$ are contiguous along a boundary cycle.
To justify this assumption, note that the association of bunches are updated either in ASSOCATOB procedure or in the procedure which splits a section of wavefront.
Since the wavefront split procedure maintains the contiguity property, the input bunch associations are guaranteed to be contiguous, provided that ASSOCATOB procedure maintains the contiguity.
It also assumes that $SB_A \cup SB_B$ is contiguous before invoking the procedure.
This is valid since we always invoke the procedure with arguments which obey this precondition.
\hfil\break

First, note that there exists a contiguous association of $SB_A \cup SB_B$ to bunches in $A \cup B$, due to Lemma \ref{lem:contigprop3}. 
The re-association of $SB_A$ to $B$ starts from the first element in $SB_A$.
Consider the first bunch $b_1$ in $B$.
The procedure assigns all the corridor convex chains or enter/exit boundaries prior to the intersection of I-curve($b_1, b_2$) with $SB_A$ to $b_1$, if these boundary elements are determined closer to $b_1$ as compared to segments in $A$. 
Let this set be $SB'$.
Since the associations start from the first element in $SB_A$, and the elements in $SB_A \cup SB_B$ are contiguous,  it is guaranteed that $SB'$ together with the section of boundary already associated with bunch $b_1$ prior to the invocation of this procedure is contiguous.
Inductively, this argument can be extended to all the bunch re-associations in $B$.
The binary search (line 4) over the inter-bunch I-curves returns the correct next I-curve which does not intersect $c_{prevc}$.  
Once $j$ in $b_j$ is incremented (line 4) it never gets decremented, hence associating boundary edges from $A$ to $B$ (line 7) maintains the contiguity for each bunch.
In every iteration of the repeat loop (line 3), we iterate sequentially along the I-curves among bunches in $B$.  
Hence $SB_B$ is contiguous after the re-associations.
All the convex chains or enter/exit boundaries removed from $SB_A$ (line 11) are contiguous, hence $SB_B$ is contiguous after the invocation of ASSOCATOB.
For the bunches in $A$ whose associations are left as they were, based on the preconditions, contiguity is obviously maintained.
\hfil\break

Note that ASSOCATOB procedure does not take into account the case in which two inter-bunch I-curves of $B$ intersect before they intersect the section of boundary associated with $A$. 
The correctness for this case relies on the fact that the boundary- and/or waveform-sections of bunches in $B$ are updated whenever the Type-IV event corresponding to this inter-bunch I-curve intersection occurs.  
\hfil\break

Also, the correctness of IIntersect procedure is discussed in Section \ref{sect:IIntersectProc}.
\end{proof}

\begin{cor}
\label{cor:merge}
The MERGE($A, B$) procedure correctly updates the associations of bunches in both the sections of wavefronts' $A$ and $B$.
\end{cor}

\section{Boundary Split}
\label{sect:boundarysplit}

As described in Section \ref{sect:defsandprops}, a boundary cycle may split. 
Suppose the edge $e$ of a junction $J$ (or $C$) is associated with the wavefront from both of its sides.
This causes $e$ to appear twice in a boundary cycle $BC$.
In Fig. \ref{fig:boundarysplit}, the bounding edges $b_i$ and $b_j$ of $BC$ correspond to $e$.
Whenever edge $e$ is struck from either of its sides, the boundary cycle splits into two disjoint boundary cycles, say $BC_i$ and $BC_j$.
Let $REG_i, REG_j$ be the regions bounded by $BC_i$ and $BC_j$ respectively. 
Both the boundary cycles $BC_i$ and $BC_j$ are surrounded by the wavefront, whereas at least two edges of $J$ (or $C$) are struck.
Considering the non-crossing property of shortest paths (Lemma \ref{lem:spnoncrossing}): 
\begin{itemize} \itemsep -2pt
	\item a shortest path to $t$ could either traverse only some edges/vertices of $REG_j$, or 
	\item traverse some edges/vertices of $REG_i$, one/two edges of $J$ (or $C$), or 
	\item traverse some edges/vertices of $REG_i$, one/two edges of $J$ (or $C$), and some edges/vertices of $REG_j$ in that order.
\end{itemize}
In other words, the only way in which a shortest path traversing $REG_i$ can reach $t$ is through $J$ (or $C$) as shown in Fig. \ref{fig:gatewaydef}.
Hence we refer $J$ (resp. $C$) as a gateway, say $g$.

\begin{figure}
\centerline{\epsfxsize=350pt \epsfbox{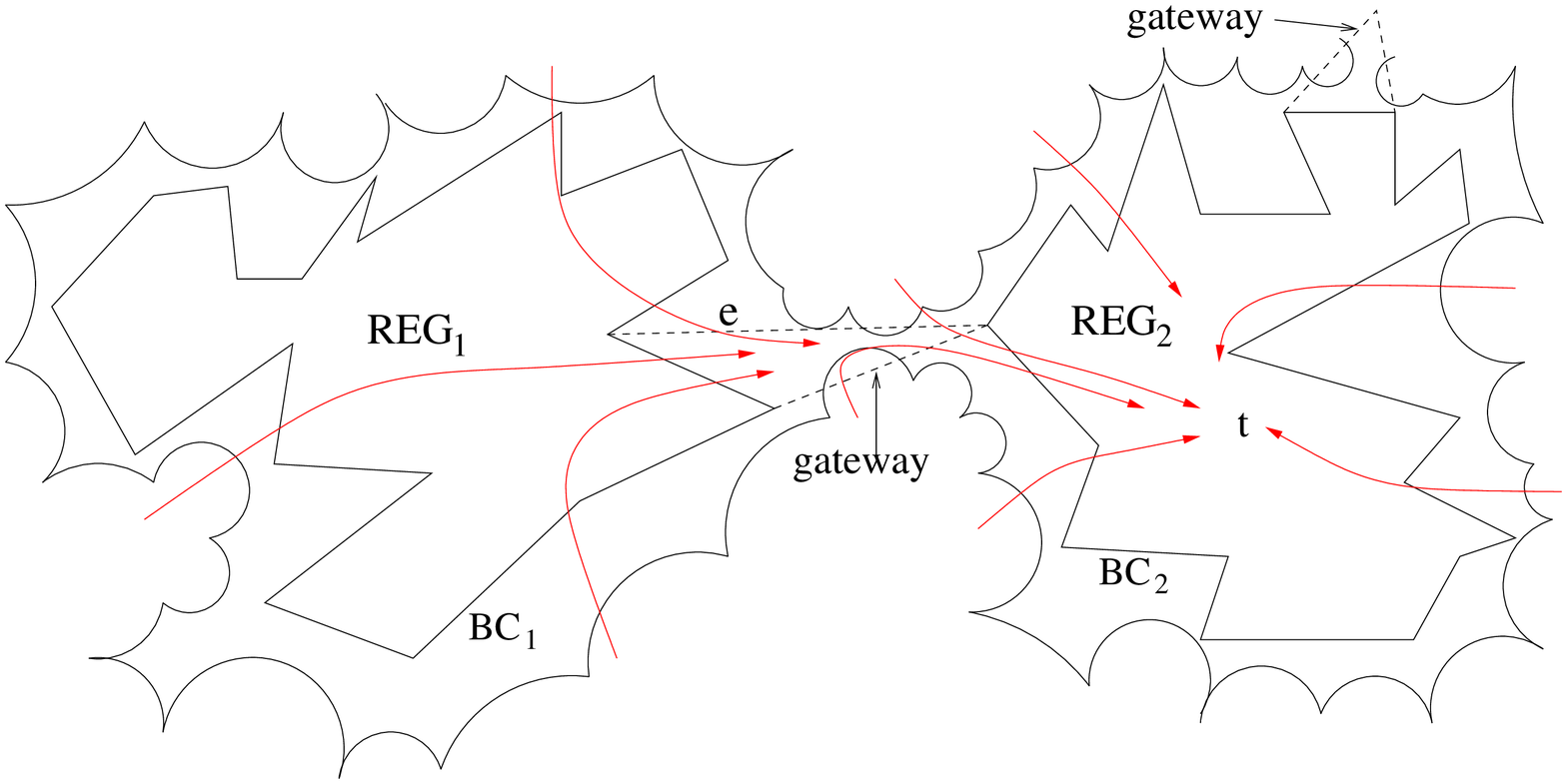}}
\caption{\label{fig:gatewaydef}{Wavefront segments that could cause a shortest path to $t$}}
\end{figure}

\begin{definition}
For a bounding edge $e \in S$ of a junction/corridor $g$, suppose $e$ is the last edge that struck, which resulted in the split of a boundary cycle $BC$ into $BC_i$ and $BC_j$, then $g$ is termed as a {\bf gateway}.
For every edge $e'$ of $g$, if $e'$ incident to both $BC_i$ and $BC_j$ then $e'$ is termed as a {\bf gateway-edge} of $g$.
\end{definition}

We denote that the gateway $g$ is {\it attached} to $REG_i$ (resp. $REG_j, BC_i, BC_j$), or $REG_i$ (resp. $REG_j, BC_i, BC_j$) is attached to $g$.
Taking Euclidean metric with the non-crossing nature of shortest paths (Lemma \ref{lem:spnoncrossing}) into consideration, each gateway is assigned a orientation, so that it suffice to expand the wavefront traversing through it over only one region attached to it.
Since in our case the destination is in $REG_j$, this is the region needed to be traversed by the wavefront expanding through $g$. 
The gateway $g$ is termed as the {\it outgoing gateway w.r.t. $REG_i$ (resp. $BC_i$)}.

\begin{lemma}
\label{lem:outgoinggw} Every region $REG$ either contains the destination $t$ or is attached with one outgoing gateway.
\end{lemma}
\begin{proof}
We prove this statement using induction.
When the algorithm starts, there is only one region, which does contain $t$.
Assume that the induction hypothesis holds for the input polygonal region having $k$ regions.
Now we extend the argument when a boundary split occurs, creating $k+1$ regions.
Let a region $REG$ is split into two regions, $REG_i$ and $REG_j$.
Suppose $g'$ is the gateway attached to both $REG_i$ and $REG_j$.
>From the induction hypothesis, we know that the region $REG$ either contains $t$ or an outgoing gateway $g$ is attached to it.
For the former case, where $REG_j$ (resp. $REG_i$) contains $t$, we orient $g'$ to be an outgoing gateway w.r.t. $REG_i$ (resp. $REG_j$).
In the case that $t$ is not located in $REG$ and $g$ is attached to $REG_j$ (resp. $REG_i$), again we orient $g'$ to be an outgoing gateway w.r.t. $REG_i$ (resp. $REG_j$).
\end{proof}

\begin{figure}
\centerline{\epsfxsize=370pt \epsfbox{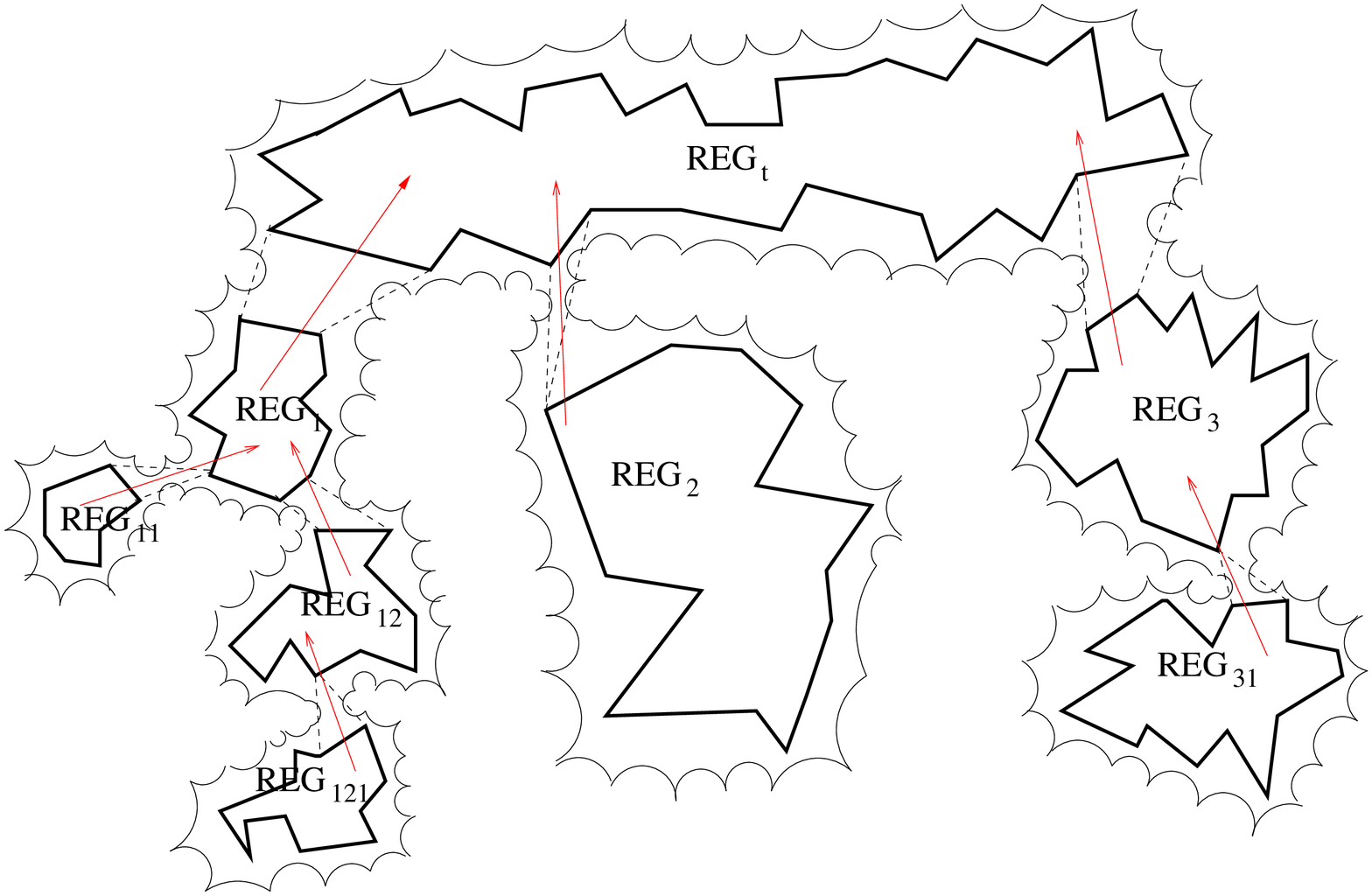}}
\caption{\label{fig:bct}Boundary Cycle Tree and the corresponding region splits}
\end{figure}

The wavefront propagation and the initiation of boundary cycles can be represented as a tree.
The initial shortest path wavefront $w(s)$ is of $\epsilon$ radius, and the untraversed region contains $t$.
Due to the wavefront progression, suppose $REG$ splits into two, now one of these regions contains $t$. 
Let us denote the region containing $t$ by $REG_t$ and the other region by $REG_1$.
Both $REG_t$ and $REG_1$ are connected with a gateway $g_1$.
The orientation of $g_1$ is from $REG_1$ to $REG_t$.
Then we represent $REG_1$ and $REG_t$ as nodes and gateway $g_1$ as an edge in a graph as described below.
\hfil\break

At any point of the algorithm, all the untraversed regions together with the gateways that connect respective regions together are represented with an oriented {\bf Boundary Cycle Tree}, $BCT(V, A)$, where the set $V$ comprises of the set of untraversed regions in the polygonal domain and the set $A$ comprises of gateways, such that every arc $a=(v', v'') \in A$ represents a gateway from the boundary cycle represented at $v'$ to the boundary cycle represented at $v''$.
The nodes and edges are added to the (logical) boundary cycle tree as the algorithm proceeds i.e., whenever there is boundary split, one new node corresponding to the new region and one edge corresponding to the gateway are added to BCT. 
For example, in Fig. \ref{fig:bct}, consider a path $REG_t, REG_{1}, REG_{12}, REG{121}$ in BCT.
This corresponds to a boundary splits that occurred over time among the regions along this path, the first split being at the root.
Also, a shortest path to $t$ that occurs along this path must traverse across a suffix of the regions corresponding to nodes $REG_{121}, REG_{12}, REG_{1}, REG_t$ in the boundary cycle tree.
In other words, consider a section of wavefront $W$ that traverses the region at node $v$ of BCT; for $W$ to cause a shortest path $SP$ to $t$, $SP$ must traverse across the regions and gateways at all the ascendant nodes of $v$.
\hfil\break

A shortest path from $s$ to $t$ may possibly goes through the region corresponding to a node in the tree and traverses across all the gateways occurring in the path from that node to the root.
At every node $v$ of the boundary cycle tree, we merge all the wavefronts that could traverse the gateway corresponding to $v$, say $g$.
In other words, all the wavefronts that traverse the regions/gateways associated with the nodes of the subtree rooted at $v$ are merged at $g$.
\hfil\break

To determine the location of either $t$ or the outgoing gateway w.r.t. a boundary cycle $BC$, we use the algorithm given in \cite{Vaidya86}.
This facilitates in orienting the gateways.
\hfil\break

\begin{lemma}
\label{lem:gatewayscorr}
Given that there exists a path from $s$ to $t$, a shortest path can be found using some (none) of the gateways at any stage of the wavefront progression.
\end{lemma}
\begin{proof}
The proof is by induction on the number of regions present when the shortest path from $s$ to $t$ is found.
Consider the base case in which $t$ resides on or within one boundary cycle by the time the shortest path is found.  
Here, we find a shortest path from $s$ to $t$ without using any gateways. 
Assume that the induction hypothesis holds for the given polygonal region having $k$ regions.  
Now we extend the argument for $k+1$ regions.
Let a region $REG$ be split into two regions, say $REG_i$ and $REG_j$.  
Suppose the region $REG_j$ has either $t$ or an outgoing gateway $g'$ w.r.t. $REG$ attached to it.  
Also, suppose $g$ is the outgoing gateway w.r.t. $REG_i$ attached to both $REG_i$ and $REG_j$.  
See Fig. \ref{fig:pathinbc}.
\hfil\break

\begin{figure}
\centerline{\epsfxsize=370pt \epsfbox{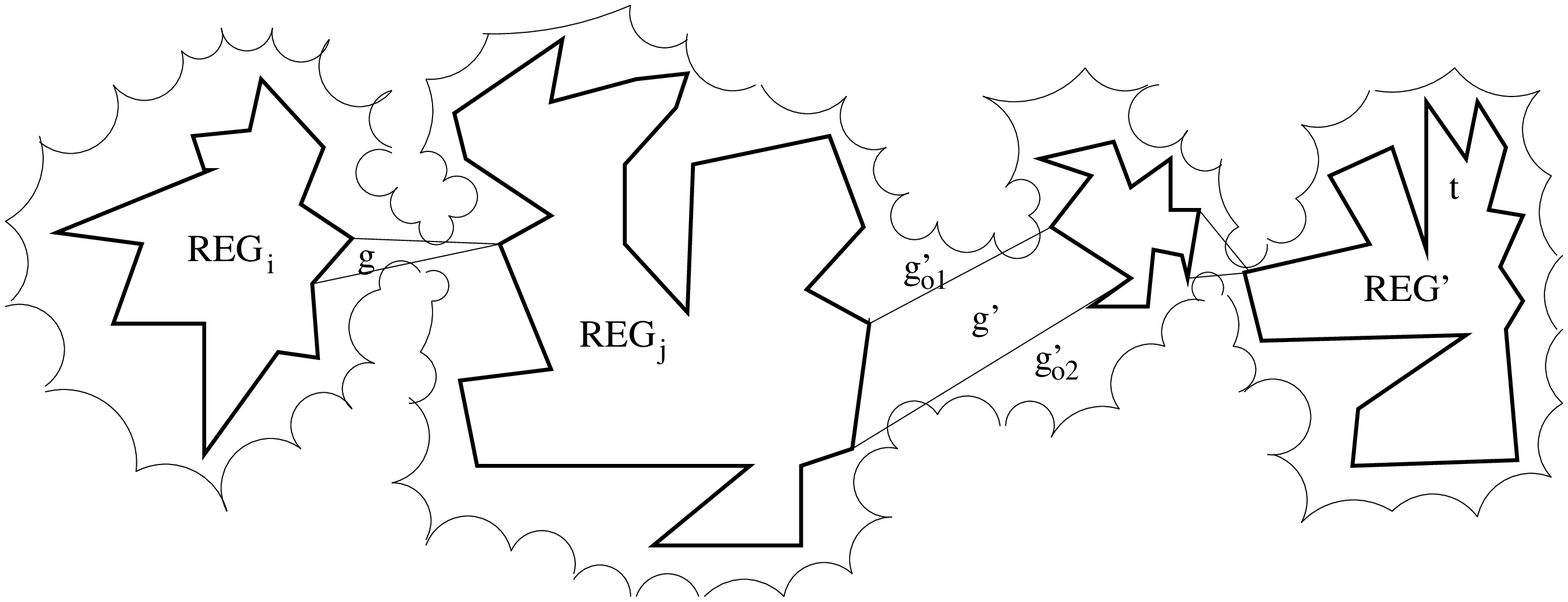}}
\caption{\label{fig:pathinbc}A path in the boundary cycle tree}
\end{figure}

>From the induction hypothesis, we know that with $k$ gateways, there is a shortest path from $s$ to $t$ where that shortest path uses some (none) of the gateways.  
If this shortest path does not traverse across $REG$, there is nothing to prove.
Otherwise, this shortest path may go in one of these ways: traversing across $REG_j$ only; traversing across $REG_i$ first and then entering $REG_j$; traversing across $REG_j$ first and later entering $REG_i$ before reentering $REG_j$.  
In the first case, we can find a shortest path without using $g$.  
For the second case, the orientation of $g$ with the non-crossing property of the wavefront assure that such a shortest path is retained.  
The Euclidean metric makes the last case inessential, hence the orientation of $g$ eliminates this case altogether. 
\end{proof}

Let $A$ and $B$ be the sections of wavefronts associated with the edges $b_i$ and $b_j$ (the bounding edge that is just struck) respectively.
Based on the above Lemma, $A$ and $B$ need to be associated with $BC_j$ only.
Let $W$ be the set of bunches associated with boundary cycle $BC_j$ that either contains $t$ or attached with an outgoing gateway.
Let $A'$ denote the listing of bunches in $W$ ordered so as to end with the last  bunch $A$ associated with the bounding edge $b_{i}$ while traversing the wavefront.
Similarly, we define $B'$ to be the wavefront bunches ordered in the reverse direction so as to end with $b_j$.
We merge $A'$ and $B'$ by invoking the MERGE procedure.
This determines the association of bunches in $A$ and $B$ to the bounding edges of the boundary cycle $BC_j$.
\hfil\break

\paragraph{Wavefront propagation along the gateways} \hfil\break

We next show how to process the wavefront propagation over the boundary cycles.
The orientation of the gateways imposes an ordering of boundary cycles for processing the regions.
Consider the current wavefront associated with region $REG_i$ and the outgoing gateway $g$ attached to it. 
We expand the shortest path wavefront over $REG_i$, until every corridor/junction bounding edge in that region is struck.  
When this happens, we say that $REG_i$ is traversed.
When $BC_i$ is a degenerate cycle comprising one vertex, which is not $t$, this traversal is not required and the algorithm can proceed to processing of $REG_j$.
During the traversal of $REG_i$, all the bunches that strike gateway-edges of $g$ are combined into one section of wavefront, termed as $B$.
Suppose $g'$ is the outgoing gateway attached to $REG_j$.
Also, let $g'_{o1}, g'_{o2}$ be the sides of edges of $g$ from outside of $g$. 
Consider any region $REG'$ corresponding to an ancestor of node corresponding to $REG_i$ in the boundary cycle tree.
See Fig. \ref{fig:pathinbc}.
>From non-crossing property of sections of wavefront (Lemma \ref{lem:spnoncrossing}), note that there is no way for a wavefront segment in $B$ to reach a point located in $REG'$, without traversing some of: $g'_{o1}$, $g'_{o2}$, edges in $BC_j$.    
Since wavefront segments in $B$ could potentially expand and traverse into $REG_j$ and the regions that are ancestors of $REG_j$, the section of wavefront $B$ is associated with the bounding edges of $BC_j, g'_{o1}$ and $g'_{o2}$ so that to cover all possible non-crossing shortest paths. 
\hfil\break

We backtrack in time and restart expanding the shortest path wavefront from the time at which $BC$ was split into $BC_i$ and $BC_j$, say time $t_j$.
During this re-traversal, we are interested in associations of $B$ with $REG_j, g'_{o1}$ and $g'_{o2}$.
Let $A$ be the sections of wavefront associated to the bounding edges of $BC_j$, $g'_{o1}$ and $g'_{o2}$.
We invoke MERGE procedure to merge $B$ with $A$.
Let $A'$ denote the listing of bunches in $A$ starting with the first bunch among the bunches in $A$ associated with the bounding edge $b_{i+1}$. 
Then we invoke MERGE procedure with $A'$ and $B$ as arguments.
The combined wavefront will then expand into the region $REG_j$, starting at time $t_j$.
\hfil\break

However, note that $B$ is the result of wavefront propagation that specifies the bunches at a future time $t'$, whereas $A$ is the wavefront section when the boundary split occurs (i.e., at time $t_j$).
We thus create an offset at the root of $B$ equal to time $-t'$.
This negative value ensures that the wavefront section $B$ and $A$ are processed at the same time.
A segment that has a negative offset is not active and will not be considered for the shortest distance computations.
The computations will take into account the offsets of segments in $B$. 
\hfil\break


\begin{lemma}
\label{lem:mergecorr}
The associations of all sections of wavefront involved in an invocation of a merge procedure are updated correctly.
\end{lemma}
\begin{proof}
Suppose the boundary cycle $BC$ is split into $BC_i$ and $BC_j$. 
Let $REG_i, REG_j$ be the regions bounded by $BC_i, BC_j$, respectively.
Let $g', g$ be outgoing gateways w.r.t. $BC, BC_i$ respectively.
Further, let $g'_{o1}, g'_{o2}$ be the sides of $g'$ from outside $g'$. 
The region $REG_j$ is guaranteed to either contain $t$ or attached with an outgoing gateway.
In both these cases, from  Lemma~\ref{lem:gatewayscorr} it suffice to associate the sections of wavefront traversing through gateway $g$ solely to the bounding edges of $BC_j, g'_{o1}, g'_{o2}$. 
Hence, while progressing the wavefront over $REG_i$, we combine all the sections of wavefront associated with the gateway-edges of $g$ into a section of wavefront.   
And, we associate this section of wavefront with the bounding edges of boundary cycle $BC_j$ and the gateway-edges of $g'$.
To reduce the number of event points, at first we traverse $REG_i$ and then we determine these associations by invoking MERGE procedure.  
After these associations are determined, we offset the wavefront segments that traverse region $REG_i$ by the time difference from the time at which $BC$ was split.
\hfil\break

Suppose a bunch $w$ was determined to be in a section of wavefront $B$ which struck gateway-edges of $g$.  
As the wavefront progresses, the bunch $w$ may get split or merged with other sections of wavefronts.
Since we are, in effect, restarting the wavefront from the time at which $BC$ was split into $BC_i$ and $BC_j$, the correctness proof needs to take into account the following:

\begin{enumerate} \itemsep -2pt
	\item bunch $w$ may either not be alive.
	\item bunch $w$ associated with an untraversed edge in $BC_i$.
	\item bunch $w$ may be associated to an edge of $BC_j$ while that edge was already traversed.
\end{enumerate}

In the first two cases, although $w$ is split and may be part of various boundary-/waveform-sections, the invalid marks in the corresponding boundary-/waveform-section trees guarantee that these can not cause Type-I/Type-II events.  
To handle the third case, as soon as we backtrack to the time at which boundary cycle $BC$ was split, we are initiating all possible bunches corresponding to boundary-sections and waveform-sections which will traverse $g$ in future (as determined by the wavefront propagation within $REG_i$).  
This facilitates in moving a bunch (resp. waveform-section) in $B$ forward with another section of wavefront which is also associated with the same boundary-section (resp. edge) as $B$.
Hence the third case is not possible.
\hfil\break

The correctness of updating associations, boundary-sections and waveform-sections of a given two sections of wavefront relies on the MERGE procedure, whose correctness is proved in the Lemma \ref{cor:merge}.
\end{proof}

\begin{lemma}
\label{lem:nummerges} The total number of merges during the entire algorithm are $O(m)$.
\end{lemma}
\begin{proof}
The merging is done in a partially traversed junctions or corridors, and there are $O(m)$ corridors or junctions. 
\end{proof}

\begin{lemma}
\label{lem:numgateways} The total number of gateways/gateway-edges are $O(m)$.
\end{lemma}
\begin{proof}
Since there are a total of $O(m)$ corridors/junctions together, the number of possible boundary cycles are $O(m)$.  
>From Lemma \ref{lem:outgoinggw}, each boundary cycle except the one having $t$ is attached with one outgoing gateway.  
Hence the complexity.
\end{proof}

\begin{theorem}
\label{thm:typeItypeIIevtsduetoGateways} The total number of Type-I/Type-II events due to gateways are $O(m)$. 
\end{theorem}
\begin{proof}
Let $g$ be an outgoing gateway w.r.t. boundary cycle $BC$, and let $ge$ be a gateway-edge of $g$.
Let $REG$ be the region bounded by $BC$.
A gateway-edge is processed only when all the processing is complete for $REG$.
Hence, a side of gateway-edge can cause at most $O(1)$ Type-I/Type-II events. 
Since there are $O(m)$ gateway-edges (Lemma \ref{lem:numgateways}), the total number of Type-I/Type-II events due to gateways are $O(m)$.
Therefore, the total number of Type-I/Type-II events are $O(m)$.
\end{proof}

\begin{lemma}
\label{lem:tinsplits}
The amortized time complexity in orienting all the gateways is $O(m(\lg{m})(\lg{n}))$. 
\end{lemma}
\begin{proof}
We can solve point location problem among two polygons, where each is having $O(m_1)$ and $O(m_2)$ edges respectively, in $O(min(m_1,m_2))$ time (this approach has been used before in \cite{Vaidya86}).  
However, we are interested in the case of two polygons bounded with $O(m_1)$ and $O(m_2)$ edges/convex chains with each convex chain having $O(n)$ edges.  
The planar point location algorithm given in \cite{Vaidya86} can be extended to yield an algorithm with the time complexity $O(min(m_1,m_2)\lg{n})$ for this case. 
For each boundary split, we need to locate either $t$ or a gateway, and, the boundary splits occur recursively.
Then the recurrence representing this recursion is $T(m)=T(m_1) + T(m_2) + \min \{m_1, m_2\} \lg n$, where $T(m)$ is the time required to solve the location problem in a polygonal region with $O(m)$ corridor chains and entry/exit boundaries.
This results in an overall complexity of $O(m(\lg{m})(\lg{n}))$.
\end{proof}

\begin{lemma}
\label{lem:typeIIInumeventsinmerger}
The total number of Type-III events generated in merge procedure are $O(m)$.
\end{lemma}
\begin{proof}
We will defer the proof till later where in fact we will show that the total 
number of Type-III events in the entire algorithm are bounded.
\end{proof}

\section{Event Points}
\label{sect:evtpts}

\subsection{Type-I and Type-II Events}
\label{subsect:typeItypeII}

The Type-I event determination involves computing the shortest distance for a segment in a waveform-section to strike the associated corridor convex chain or enter/exit boundary. 
The Type-II event determination involves finding the shortest distance for the bunch to expand before striking either a corridor convex chain or enter/exit boundary in the boundary-section associated with it.  
These shortest distance computations are explained in Section \ref{sect:sdcomp}.  
Both the Type-I and Type-II events change $\partial B$.  
\hfil\break

\begin{figure}
\centerline{
\subfigure[Split due to a junction]{\epsfysize=120pt \epsfbox{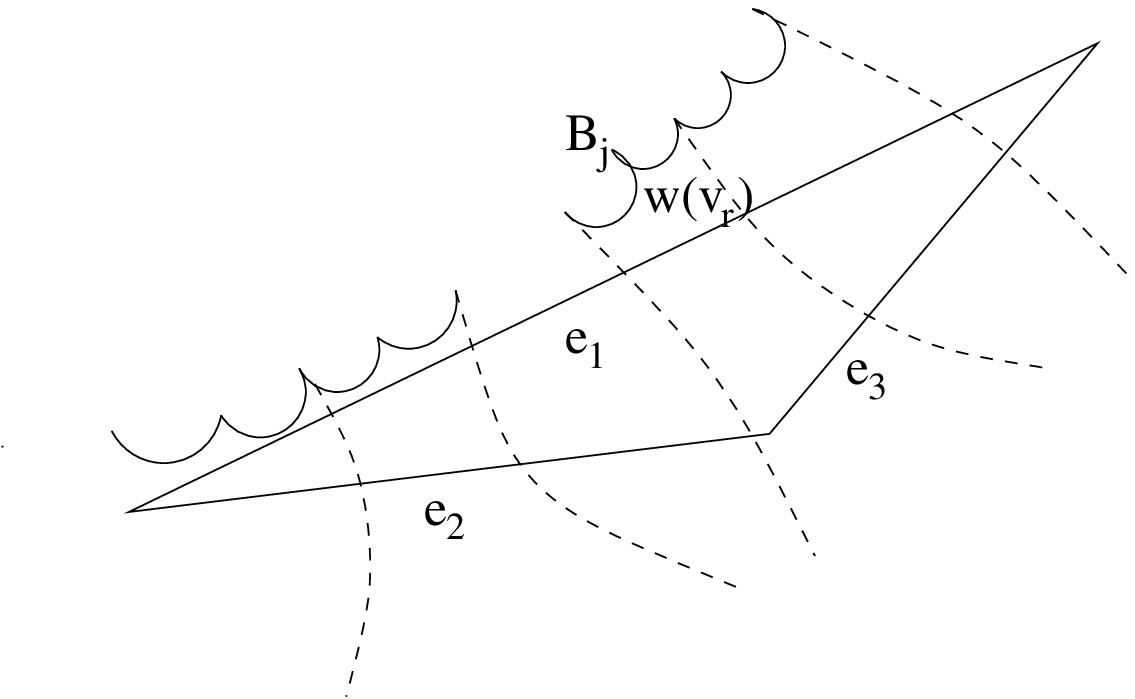}}
\subfigure[Split due to a corridor]{\epsfysize=200pt \epsfbox{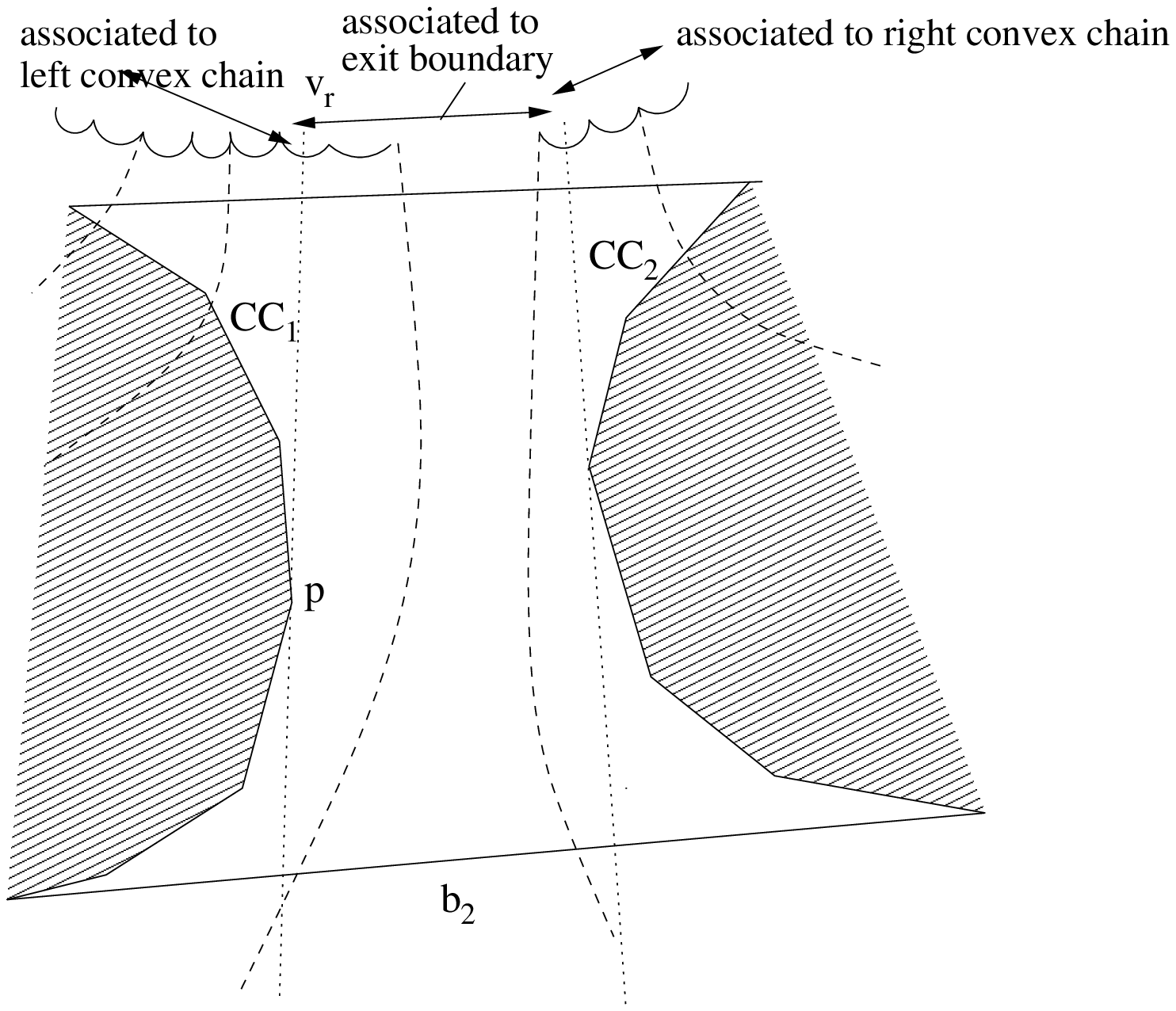}}
}
\caption{\label{fig:bunchsplit} Split of a section of wavefront}
\end{figure}

Let $J$ be an untraversed junction i.e., no edge/vertex of it has been traversed.
When a section of wavefront $SW = \{B_i, \ldots, B_j, \ldots, B_k\}$ strikes an edge $e_1=(v_1, v_2)$ of $J$, $SW$ may be split into at most two sections: a sequence of wavefront segments $SW_1$ that could possibly strike an edge $e_2$ of $J$ and a sequence of wavefront segments $SW_2$ that could possibly strike an edge $e_3$ of $J$ such that $SW_1 \cup SW_2 = SW$. 
See Fig. \ref{fig:bunchsplit}(a).
The edge $e_1$ in $\partial B$ is replaced by $e_2$ and $e_3$.
We intend to determine two successive inter-bunch I-curves, I-curve($B_{j_{prev}}, B_{j}$), I-curve($B_j, B_{j_{succ}}$) in the wavefront such that I-curve($B_{j_{prev}}, B_{j}$) intersects $e_2$ and I-curve($B_j, B_{j_{succ}}$) intersects $e_3$.
Here, $B_{j_{prev}}$ (resp. $B_{j_{succ}}$) is the bunch that precede (resp. succeed) $B_j$ in the wavefront.
This allows associating $B_i, \ldots, B_{j-1}$ to $e_2$ and $B_{j+1}, \ldots, B_k$ to $e_3$.
We do binary search over the inter-bunch I-curves of $SW$ to find such $j$.
In this binary search, primitive operation computes the intersection of a inter-bunch I-curve with either $e_2$ or $e_3$.
Since I-curves are higher-order curves, explicitly computing their intersections with $e_2$ and/or $e_3$ is not efficient.
Hence we use IIntersect procedure described in Section \ref{sect:IIntersectProc}.
Once we determine such $j$, we do binary search over the I-curves of bunch $B_j=(w(v_l), \ldots, w(v_r), \ldots, w(v_u))$ to compute two I-curves such that I-curve($w(v_{r-1}), w(v_r)$) would intersect $e_2$ and I-curve($w(v_r), w(v_{r+1})$) would intersect $e_3$.
Since intra-bunch I-curves are straight-lines, we compute the intersection between an I-curve and an edge.
Then using bunch split procedure listed in Section \ref{subsect:bhtdatastr}, we split $B_j$ into $B_j', B_j''$ such that $B_j'$ comprise the wavefront segments $\{w(v_l), \ldots, w(v_r-1)\}$, and $B_j''$ comprise the wavefront segments $\{w(v_r), \ldots, w(v_u)\}$.
Then we insert $w(v_r)$ to $B_j'$ so that $w(v_r)$ is present in both $B_j'$ and $B_j''$ as it can strike either $e_2$ or $e_3$.
After that, $SW_1$ comprise $B_i, \ldots, B_{j-1}, B_j'$, and $SW_2$ comprise $B_j'', B_{j+1}, \ldots, B_k$.
\hfil\break

Let $C$ be an untraversed corridor i.e., no edge/vertex of it has been traversed.
Also, let $C$ has $CC_1, CC_2$ convex chains, and $b_1, b_2$ enter/exit bounding edges.
When a section of wavefront $SW = \{B_i, \ldots, B_j, \ldots, B_k, \ldots, B_l\}$ strikes an enter boundary $b_1$ of $C$, $SW$ may be split into at most three sections: 
a sequence of wavefront segments $SW_1$ that could possibly strike $CC_1$, a sequence of wavefront segments $SW_2$ that could possibly strike an edge $b_2$, and a sequence of wavefront segments $SW_3$ that could possibly strike an edge $CC_2$ of $C$ such that $SW_1 \cup SW_2 \cup SW_3 = SW$. 
See Fig. \ref{fig:bunchsplit}(b).
The $b_1$ in $\partial B$ is replaced by $CC_1, b_2, CC_2$.
We determine two successive inter-bunch I-curves, I-curve($B_{j_{prev}}, B_j$), I-curve($B_j, B_{j_{succ}}$) in $SW$ such that I-curve($B_{j_{prev}}, B_j$) intersects $CC_1$ and I-curve($B_j, B_{j_{succ}}$) intersects $b_2$.  
Here, $B_{j_{prev}}$ (resp. $B_{j_{succ}}$) is the bunch that precede (resp. succeed) $B_j$ in the wavefront.
Once we determine such $j$, we do binary search over the I-curves of bunch $B_j$ to find $w(v_r) \in B_j$ which could possibly intersect both $CC_1$ and $b_2$. 
As in the case of junction $J$, we split $B_j$ such that $w(v_r) \in B_j$ is in both $B_j'$ and $B_j''$.
The binary search follows the same procedure as listed above (the case of junction $J$).
Similarly, we determine two successive inter-bunch I-curves, I-curve($B_{k_{prev}}, B_k$), I-curve($B_k, B_{k_{succ}}$) in $SW$ such that I-curve($B_{k_{prev}}, B_k$) intersects $b_2$ and I-curve($B_k, B_{k_{succ}}$) intersects $CC_2$ and split $B_k$ into $B_k'$ and $B_k''$.
Here, $B_{k_{prev}}$ (resp. $B_{k_{succ}}$) is the bunch that precede (resp. succeed) $B_k$ in the wavefront.
Then $SW_1$ comprise $\{B_i, \ldots, B_{j-1}, B_j'\}$, $SW_2$ comprise $\{B_j', B_{j+1}, \ldots, B_{k-1}, B_k'\}$, whereas $SW_3$ comprise $\{B_k', B_{k+1}, \ldots, B_l\}$.
Due to the non-crossing nature of I-curves (Lemma \ref{lem:spnoncrossing}), point of tangency corresponding to Type-III event on $CC_1$ due to $SW_1$ is caused by the wavefront segment $w(v_r)$.
The distance to point of tangency $p$ on $CC_1$ along $v_rp$ from the periphery of $w(v_r)$ is pushed to the event heap. 
Another Type-III event due to the interaction of $CC_2$ and $SW_3$ is also pushed to the min-heap.
As a whole two Type-III event points, and a shortest distance between $SW_2$ and $b_2$ are the event points determined during this procedure.
Note that no other interactions between $SW_1$ and $CC_1$ or $SW_3$ and $CC_2$ worth further consideration, as the destination $t$ is in its own corridor. 
\hfil\break

Let a bounding edge of $J$ (resp. $C$) was traversed, and the wavefront struck the other bounding edge of $J$ (resp. $C$). 
This causes boundary split.
The merging procedure is explained in section \ref{sect:merging}.
\hfil\break

Consider a junction $J = (e_1, e_2, e_3)$ that is traversed.
Suppose the edge $e_1$ is traversed whereas the edges $e_2$ and $e_3$ are not traversed when $\cW (d)$.
Also, suppose the bounding edge $e_2$ is struck when $\cW (d')$ 
This causes boundary split.  
While updating the relevant associations of sections of wavefront, the merging procedure described in Section \ref{sect:merging} may in turn cause Type-III events. 
The same is true when a boundary splits due to merger in corridors.
\hfil\break 

A section of wavefront $SW$ striking a junction (resp. corridor) boundary vertex is handled similar to the above cases, except that both the edges adjacent to that vertex are considered as struck by the wavefront.

\begin{lemma}
\label{lem:numsplits} The total number of bunch/waveform-section splits during the entire algorithm are of $O(m)$.
\end{lemma}
\begin{proof}
The waveform-sections are associated with the edges which are not struck as of now.  
The split occurs when a waveform-section strikes the boundary of a junction/corridor and there are $O(m)$ such boundary edges in all the junctions/corridors together.
\end{proof}

\begin{lemma}
\label{lem:splitcompl} The complexity of bunch/waveform-section splits during the entire algorithm are \\ $O(m (\lg{m})(\lg{n}))$.
\end{lemma}
\begin{proof}
Each split takes $O(\lg{n})$ time as it involves binary search over the I-curves.
The appropriate boundary- and waveform-section updates take $O((\lg{m})(\lg{n}))$ time (Lemmas \ref{lem:wsttimeandspace}, \ref{lem:bsttimeandspace}).
Combining this with Lemma \ref{lem:numsplits}, leads to the proof.
\end{proof}

\begin{lemma}
\label{lem:typeItypeIInumevents}
The total number of Type-I and Type-II events together are $O(m)$.
\end{lemma}
\begin{proof}
The Type-I or Type-II events occur due to the following:
\begin{enumerate}[(1)] \itemsep -2pt
	\item A section of wavefront strikes either a corridor convex chain or a corridor enter/exit boundary.
	\item A section of wavefront striking gateway-edges in merge procedure.
\end{enumerate}
Since there $O(m)$ corridors, there are $O(m)$ Type-I/Type-II events due to (1). 
>From Lemma \ref{thm:typeItypeIIevtsduetoGateways}, the number of Type-I/Type-II events due to (2) are bounded with $O(m)$.
\end{proof}

\begin{lemma}
\label{lem:typeItypeIIhandlingcost}
The total time in determining and handling all Type-I/Type-II events is $O(m(\lg{m})(\lg{n}))$.
\end{lemma}
\begin{proof}
>From Lemma \ref{lem:sdcwt}, determining Type-I event is of complexity $O((\lg{m})(\lg{n}))$.
>From Lemma \ref{lem:sdcbt}, determining Type-II event is of complexity $O((\lg{m})(\lg{n})$.
\hfil\break

Both of these event types involve updating $WST$ and/or $BST$, either by splitting or by invoking the MERGE procedure.
>From Lemma \ref{lem:wsttimeandspace}, updating $WST$ takes $O((\lg{m})(\lg{n}))$ amortized time. 
>From Lemma \ref{lem:bsttimeandspace}, updating $BST$ takes $O((\lg{m})(\lg{n}))$ amortized time.
>From Lemma \ref{lem:tinsplits}, amortized time in orienting all the gateways is $O(m(\lg{m})(\lg{n}))$.
>From Theorem \ref{thm:mergecompl}, amortized time involved in the merge procedure is $O((\lg{m})(\lg{n}))$.
\hfil\break

Type-III event determination and handling costs are considered in \ref{lem:typeIIIhandlingcost}.

\end{proof}

\subsection{Type-III Event}
\label{subsect:typeIII}

This event occurs when a wavefront segment $w(v_r)$ in a section of wavefront $SW$ strikes a corridor convex chain $CC$ at a point of tangency $p$ along the tangent $v_rp$ from $v_r$ to $CC$.
The event determination procedure is listed as part of Type-I and Type-II event handling.
See Subsection \ref{subsect:typeItypeII}.
The event possibly causes either the initiation of a new bunch from $p$, or modifying an existing bunch initiated from a vertex of the convex chain $CC$.
The cases in which a bunch is initiated/modified are explained in Section \ref{subsect:bhtdatastr}.

\begin{lemma}
\label{lem:typeIIInumevents}
The total number of Type-III events are $O(m)$.
\end{lemma}
\begin{proof}
Since there are $O(m)$ corridor convex chains, number of $BHT$s initiated (in Case (1) of Subsection \ref{subsect:bhtdatastr}) are $O(m)$.
Let $C = \{ v_1 , v_2 , \ldots v_{n'} \}$ be a convex chain from which a 
bunch is initiated.
New bunches from the same convex chain
may be introduced during the merge procedure (corresponding
to Case (4) of  Subsection\ref{subsect:bhtdatastr}).
We note that at the introduction of a new bunch hull tree an old bunch of 
segments, call it $B_1$ is removed  from the wavefront. Thus
at most one bunch $B(v_z, v_{n'}$ exits the bounding edge of the
corridor defining the chain $C$. Let $B_p$ be the bunch that
led to the generation of the bunch $B_1$.
The removal of bunches can thus be charged to $B_p$. We thus need to
bound the number of bunches that a particluar bunch, say $B_p$ can
generate. If a bunch $B_p$ generates more than  two bunches, then
it does so in two different corridors. In this case, the bunch
$B_p$ is split into two at a  vertex $v_J$ in junction $J$.
The shortest path to this junction vertex is thus determined due
to a segment in $B_p$ (due to the non-crossing property of bunches).
Thus the split of bunches can be charged to junction vertices. These
are $O(m)$ in number. And each bunch that is generated and removed,
as  $B_1$ is above,  can be charged to a bunch or split portion of a bunch. 
This gives the desired bound of $O(m)$.
\end{proof}

\subsection{Type-IV Event}
\label{subsect:typeIV}

The intersection of inter-bunch I-curves within a waveform-section are captured with this event.
Using these points of intersection, we can detect when two non-adjacent bunches within a waveform-section meet.
Since the intra-bunch I-curves are diverging, only the intersection of inter-bunch I-curves and I-curves from two different bunches are considered.

\begin{lemma}
\label{lem:nonadjicurves}
Let $d'$ be the shortest distance between a waveform-section $WS$ and its associated vertex/edge $e$, which corresponds to an event $evt$.
Even though some of the I-curves in that waveform-section intersect among themselves before $evt$ occurs, $WS$ does strike $e$ for the first time when the $evt$ occurs (provided that the association of $e$ does not change due to mergers).
\end{lemma}

\begin{proof}

\begin{figure}
\centerline {
\subfigure[A view of I-curve intersection in a $WST$]{\epsfysize=140pt \epsfbox{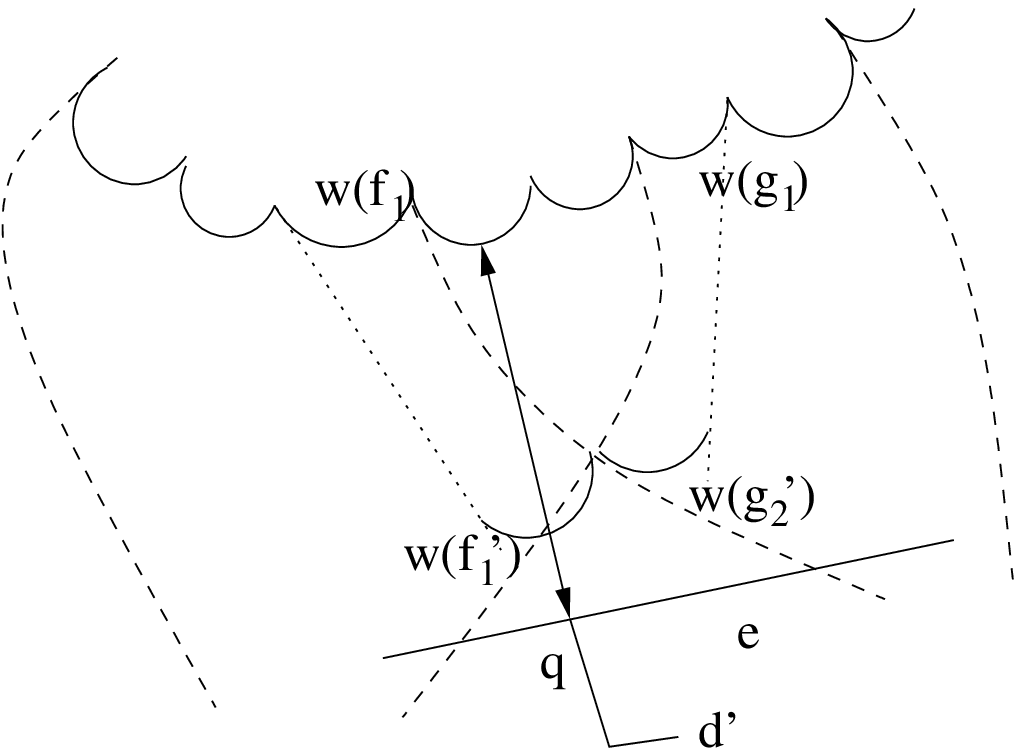}}
\hspace{0.2in}
\subfigure[Event point distance before and after I-curve intersection]{\epsfxsize=210pt \epsfbox{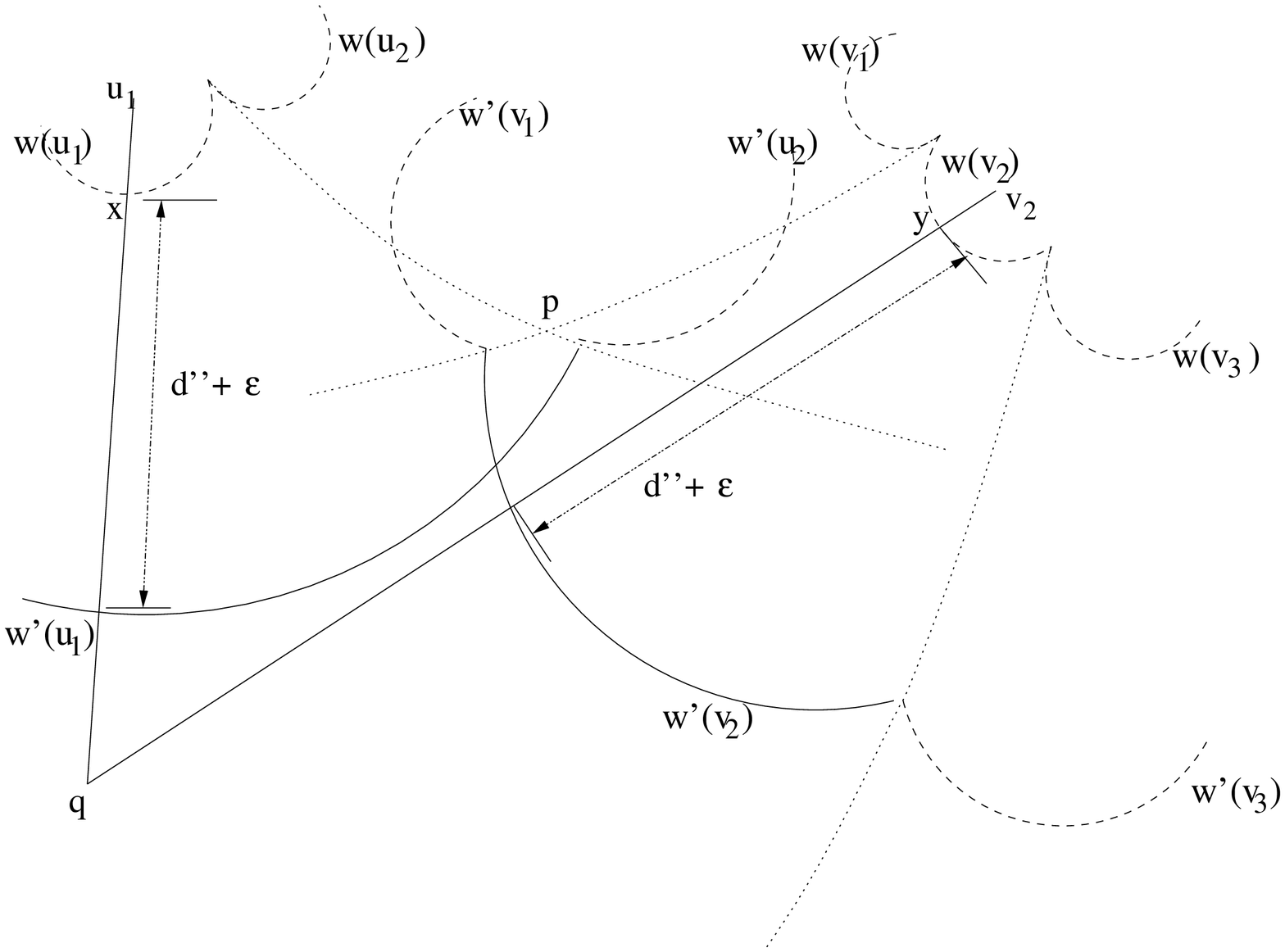}}
}
\caption{\label{fig:nonadjicurves} I-curve intersection and the computed shortest distance}
\end{figure}

\noindent
See Fig. \ref{fig:nonadjicurves}.
Suppose the inter-bunch I-curves in $WS$, I-curve$(w(u_1), w(u_2))$ and I-curve$(w(v_1), w(v_2))$, intersect at point $p$, whereas the Euclidean distance from $p$ to $WS$ is $d''$, for $d'' < d'$.
Consider the wavefront just after the $evt$ occurred i.e., after it traverses an Euclidean distance $\epsilon$ (for a positive constant $\epsilon$).
At this wavefront progression, let $WS'$ be the waveform-section comprising bunches in $WS$ sans the the bunches removed due to the inter-bunch I-curve intersection. 
Let $w'(u_1)$ (resp. $w'(u_2), w'(v_1), w'(v_2)$) be the wavefront segment in $WS'$ corresponding to the wavefront segment $w(u_1)$ (resp. $w(u_2), w(v_1), w(v_2)$) in $WS$.
Consider any point $q$ s.t. $q$ is external to both $WS$ and $WS'$, and, $q$ is closer to $w(u_1)$ than $w(v_2)$ (symmertic cases can be argued similarly).
Let $x (y)$ be the point of intersection of line segment $qu_1$ (resp. $qv_2$) with $w(u_1)$ (resp. $w(v_2)$).
Given that $\Vert qx \Vert < \Vert qy \Vert$, it is trivial to note that $(\Vert qx \Vert-(d''+\epsilon)) < (\Vert qy \Vert -(d''+\epsilon))$.
In other words, it is guaranteed that the point $q$ is closer to $w'(u_1)$ than $w'(v_2)$ whenever $q$ is closer to $w(u_1)$ than $w(g_2)$.
\end{proof}

However, for the sake of utilizing I-curves for further computations, and to avoid overlap of wavefront segments within the wavefront, we detect their intersection.
There are three types of I-curve intersections possible in a waveform-section:
\begin{enumerate}[(1)]  \itemsep -2pt
\item \label{adjint} Intersection of adjacent inter-bunch I-curves
\item \label{nonadjint} Intersection of non-adjacent inter-bunch I-curves
\item \label{partialbelim} Non-adjacent I-curve intersection causing partial elimination of bunches
\end{enumerate}

The dirty bridges discussed in Section \ref{sect:dirtybridges} takes care of Case (\ref{partialbelim}).
Since explicitly computing the I-curve intersections occurring within a $WS$ is compute intensive,  both the I-curve intersections mentioned in Case (\ref{adjint}) and Case (\ref{nonadjint}) are detected by finding the shortest distance between the sibling hulls stored at the internal nodes of $WST$.  
Let $UH_l, UH_r$ be the hulls at sibling nodes $v_l$ and $v_r$ of a $WST$ respectively.
Let $b^{rm}_l$ (resp. $b^{lm}_r$) be the rightmost (resp. leftmost) bunch in the bunches stored at the leaves of $v_l$ (resp. $v_r$).
Also, let $UH^{rm}_l$ (resp. $UH^{lm}_r$) be the hull of $b^{rm}_l$ (resp. $b^{lm}_r$).
Since two adjacent bunches always intersect along an I-curve, we avoid detecting event corresponding to the intersection of adjacent bunches.
Hence we compute the shortest distances' $d', d''$ between the hulls $UH_l-UH^{rm}_l, UH_r$ and $UH_l, UH_r-UH^{lm}_r$ respectively.
And, we push the Type-IV event point with distance $\min(\frac{d'}{2}, \frac{d''}{2})$ to the event heap.

\begin{lemma}
\label{lem:icurvehullinter}
Two I-curves in a $WST$ intersect if and only if Type-IV event occurs. 
\end{lemma}
\begin{proof}
Let $b_i, b_j$ be two bunches occurring in the left-to-right ordering of bunches in a waveform-section.
Let $b_{i_{succ}}$ is the successor of $b_i$ in the wavefront; also, let $b_{j_{pred}}$ is the predecessor of $b_j$ in the wavefront.
Suppose the I-curve($b_i, b_{i_{succ}}$) and I-curve($b_{j_{pred}}, b_j$) intersect at point $p$.
>From the I-curve definition, $p \in (b_i \cap b_j)$ when the bunch $b_i$ or $b_j$ strikes $p$.
Consider the least common ancestor node $v$ of nodes that store $b_i$ and $b_j$.
The left and right children of $v$, say $v_l$ and $v_r$, implicitly store upper hulls $UH_l, UH_r$ such that one contains $b_i$ and the other contains $b_j$.
Let $UH^{rm}_l$ (resp. $UH^{lm}_r)$) be the hull of the right-most (resp. left-most) bunch in the bunches stored at the leaves of $v_l$ (resp. $v_r$).
It is immediate to see that the bunches $b_i, b_j$ intersect, whenever either the hulls  $UH_l-UH^{rm}_l, UH_r$ or the hulls $UH_l, UH_r-UH^{lm}_r$ intersect.
This is captured as a Type-IV event.
\end{proof}

As shown in the following Lemma, a Type-IV event causes the disappearance of one or more bunches, in turn, the waveform-section is dynamically updated.
Using the IIntersect procedure (detailed in Section \ref{sect:IIntersectProc}) with $RV$ of a section of wavefront and an I-curve as parameters, we determine new associations and update the shortest distance in these new associations.
As part of this update, we compute the new inter-bunch I-curves among the bunches which became adjacent and their intersection points.
The event point min-heap is updated accordingly.
The following Lemma is useful in Analysis.

\begin{lemma}
\label{lem:hullinter}
Whenever a Type-IV event occurs, there exists at least one bunch which is not required to be progressed further.
\end{lemma}
\begin{proof}
\begin{figure}
\center{
\subfigure[Adjacent I-curve intersection]{\epsfysize=180pt \epsfbox{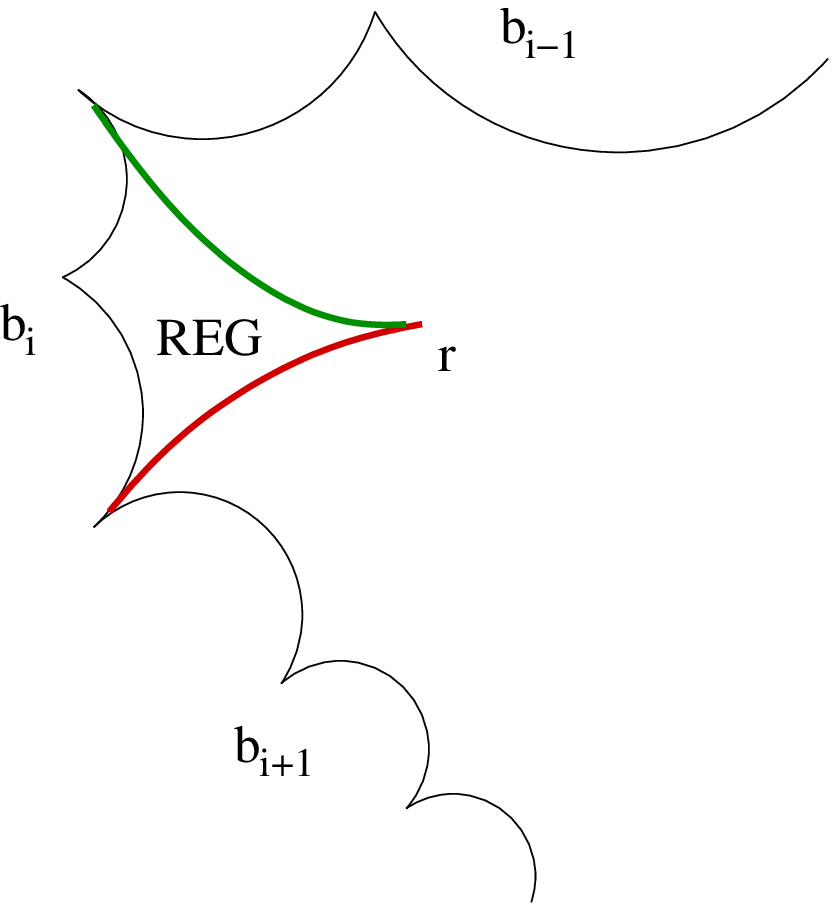}}
\subfigure[Hull intersection]{\epsfysize=180pt \epsfbox{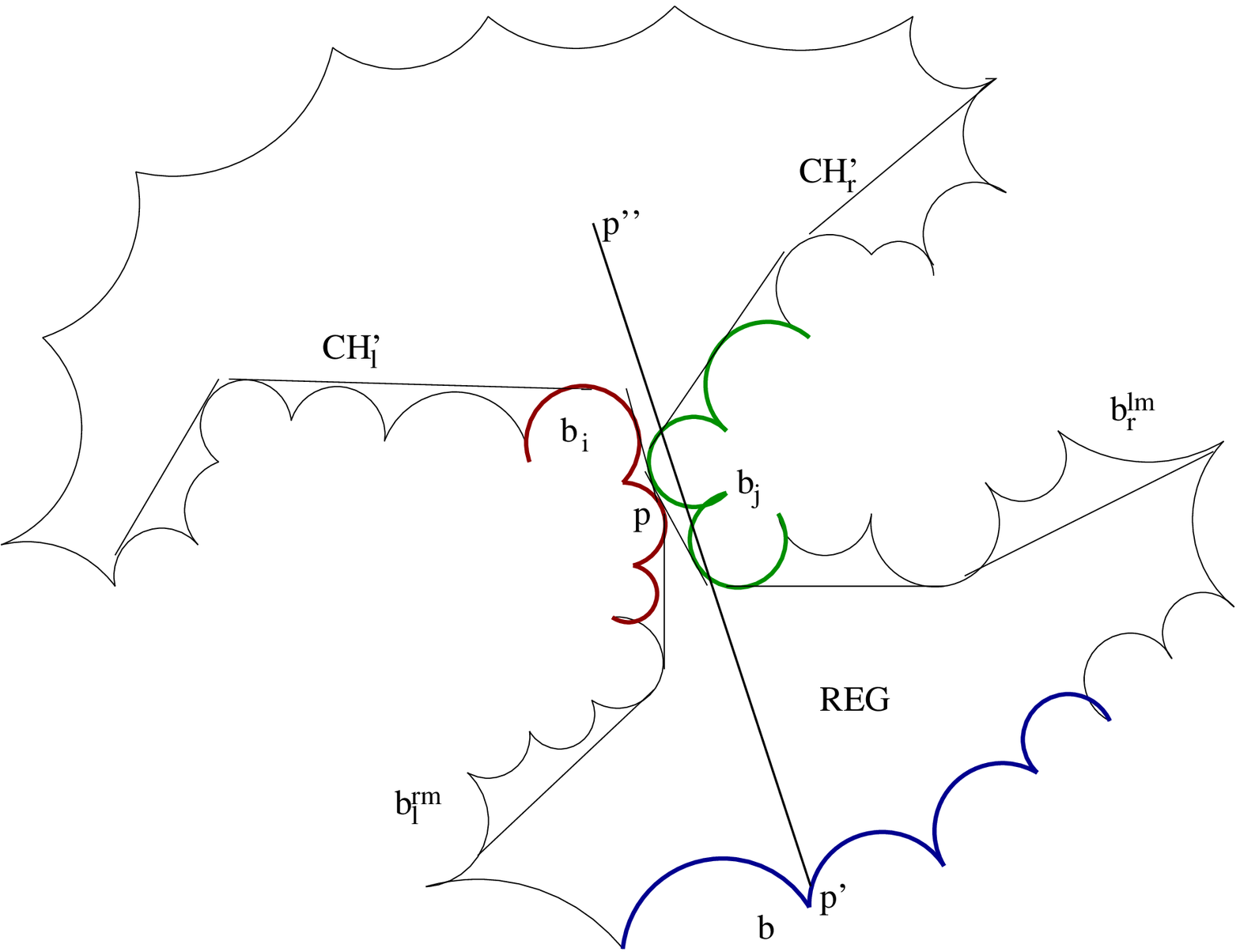}}
}
\caption{\label{fig:typeivevt} Type-IV Events}
\end{figure}
Consider two adjacent inter-bunch I-curves, I-curve($b_{i-1}, b_i$) and I-curve($b_i,b_{i+1}$).
Suppose the Type-IV event occurred due to the intersection of these two I-curves (Fig \ref{fig:typeivevt}(a)).
Let $r$ be the point of intersection.
Also, let $REG$ be the untraversed region bounded by these two I-curves and the bunch $b_i$.
Suppose the destination $t$ is located in $REG$.
Then due to the non-crossing property of shortest paths (Lemma \ref{lem:spnoncrossing}), the bunch $b_i$ is not required to be progressed in $\overline{REG}$.
Further the bunches other than $b_{i-1}, b_i$ and $b_{i+1}$ need not be 
propagated further.
Suppose the destination $t$ is located in the region $\overline{REG}$.
Again due to the non-crossing property of shortest paths (Lemma \ref{lem:spnoncrossing}), it is not required to progress the bunch $b_i$ (in $\overline{REG}$).  \hfil\break

Suppose the Type-IV event occurred due to an event point that
corresponds to the distance between two hulls and does not correspond to
the intersection of adjacent I-curves.
When this event occurs, let $CH_l', CH_r'$ be the hulls at nodes $l$ and $r$ of $WST$, and, $p$ be the point at which the hulls $CH_l', CH_r'$ intersect.
Let $B_l'$ (resp. $B_r'$) be the set of bunches at all the leaf nodes of the subtree formed with the descendents of node $l$ (resp. $r$).  
We define the bunches $b_i$ and $b_j$ herewith:
\begin{itemize}

\item Suppose the point $p \notin b$, for any bunch $b \in B_l'$ i.e., the point $p$ is located on a bridge of $CH_l'$.
Let $B_l'' \subseteq B_l'$ be the set of bunches, where a bunch $b \in B_l''$ if and only if there exists a line segment $LS$ from $p$ to a point located on $b$ s.t. the interior of $LS$ does not intersect with any bunch.
Among all the bunches in $B_l''$, we let the bunch stored at the rightmost leaf node of $WST$ be termed as $b_i$.
Otherwise, suppose the point $p$ is located on a bunch $b \in B_l'$.
Then the bunch $b$ is termed as $b_i$.  
\item Suppose the point $p \notin b$, for any bunch $b \in B_r'$ i.e., the point $p$ is located on a bridge of $CH_r'$.
Then the set $B_r'' \subseteq B_r'$ is defined similar to the above case.
Among all the bunches in $B_r''$, we let the bunch stored at the leftmost leaf node of $WST$ be termed as $b_j$.
Otherwise, suppose the point $p$ is located on a bunch $b \in B_r'$.
Then the bunch $b$ is termed as $b_j$.  

\end{itemize}
Let $B$ be the set of bunches strictly between $b_i$ and $b_j$ in the waveform-section.
$B$ is non-empty since the event corresponds to an intersection of
non-adjacent I-curves.
Let the set $S$ comprise untraversed regions s.t. a region $R \in S$ if and only if the closed set representing the region $R$ intersects with at least one bunch in $B$.
Also, let the region $REG$ be the union of the regions in $S$.
Note that the region $REG$ is not necessarily connected.

Suppose the destination $t$ is located in $\overline{REG}$.
Let $p'$ be any point chosen on any bunch $b$ in $B$.
For any point $p''$ in $\overline{REG}$, the line segment $p'p''$ intersects the boundary of $CH_l' \cup CH_r'$.
In turn, $p'p''$ intersects a bunch contained in $CH_l' \cup CH_r'$.
Then from the non-crossing property of shortest paths (Lemma \ref{lem:spnoncrossing}), bunch $b$ is not required to be progressed further. 
In other words, no bunch belonging to $B$ need to be progressed in $\overline{REG}$. 

Suppose the destination $t$ is located in $REG$.
Then by defining bunches $b_i, b_j$ symmetric to the above, the proof would be similar to the case mentioned in the above paragraph. 
\end{proof}

As required in the above Lemma, when a Type-IV event occurs we determine the region in which the destination point $t$ resides. 
This is done using the procedure from Vaidya \cite{Vaidya86}.

\section{Dirty Bridges in $WST$}
\label{sect:dirtybridges}

\begin{figure}
\centerline{
\centerline{\epsfysize=210pt \epsfbox{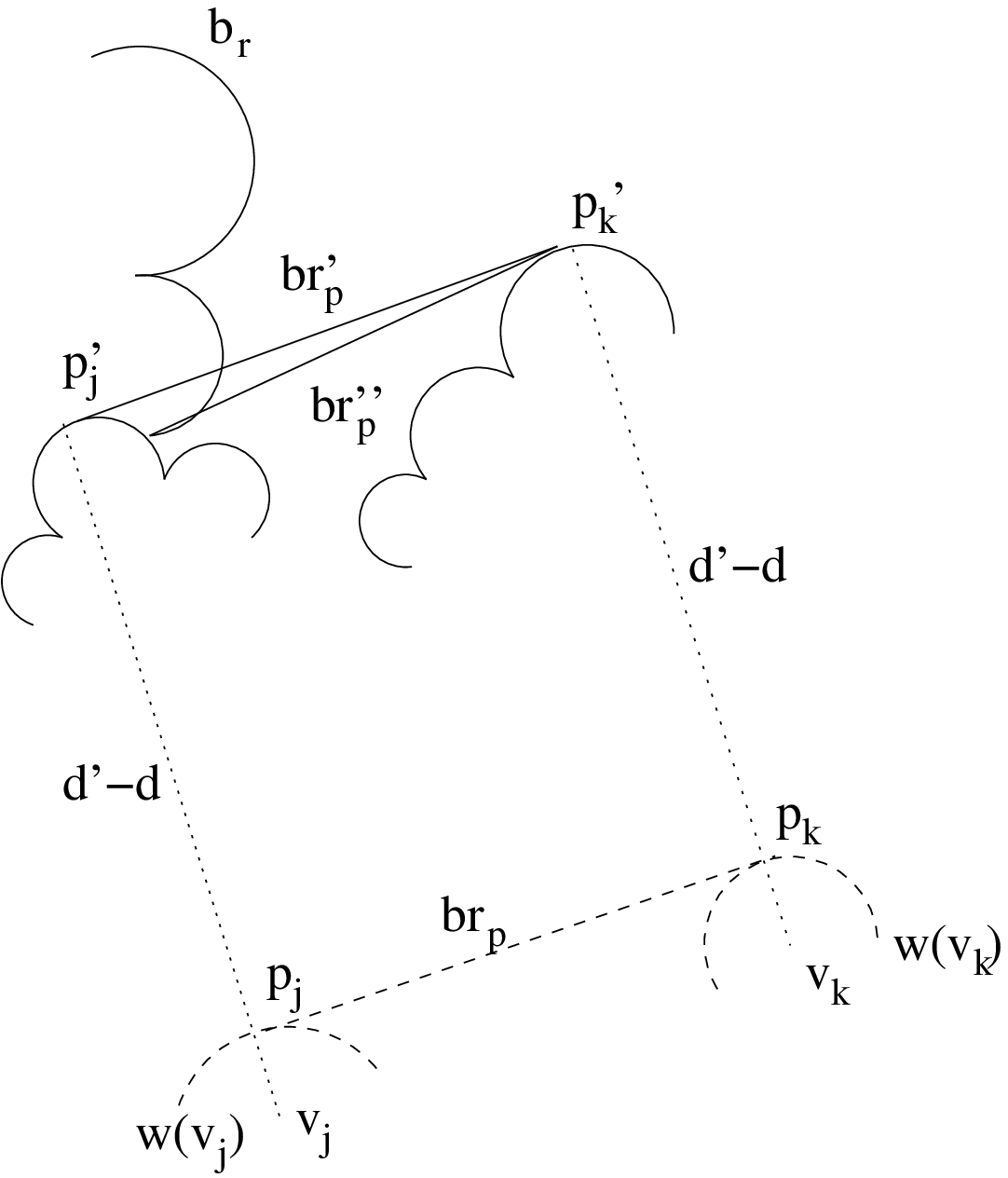}}
}
\caption{\label{fig:dirtybridges1} Dirty Bridges in $WST$}
\end{figure}

Consider the maintenance of bridges at internal nodes as the wavefront progresses.
For $p_j \in w(v_j), p_k \in w(v_k)$, let $br_p=(p_j, p_k)$ be a bridge at node $p$ of a $WST$ when $\cW (d)$, and $p_j, p_k$ are respectively the points of tangencies of $w(v_j)$ and $w(v_k)$.
Let $p_j'$ be a point at Euclidean distance $(|v_jp_j|+d'-d)$ from $v_j$ along the vector $\overrightarrow{p_jv_j}$.
Also, let $p_k'$ be a point at Euclidean distance $(v_kp_k+d'-d)$ from $v_k$ along the vector $\overrightarrow{p_kv_k}$.
As the wavefront expands, the bridge $br_p$ moves as fast as any point on a bunch $b_r$, for $i \le r \le k$, expands (Property \ref{prop:wstbridgemovement}).
Therefore, the traversal of $p_j$ (or $p_k$) could happen due to some bunch $b_r$, where $r < i$ or $r > k$ i.e., $b_r$ would be located in a subtree other than the one at $p$.  
Let an endpoint of $br_p'$, say $p_j'$, does not belong to $\cW(d')$ for $d' > d$.   
Since the Property \ref{prop:wstbridgemovement} is no more applicable, $br_p'$ is not a bridge at node $p$ at $\cW(d')$.
As shown in Fig. \ref{fig:dirtybridges1}, due to the definition of a bridge in $WST$, $br_p''$ is the bridge at node $p$ when $\cW (d')$.
However, it is inefficient and impractical to update the bridges infinitely often at nodes such as $p$.
As long as at least one wavefront segment in each of the bunches $b_j$ and $b_k$ (on which $p_j'$ and $p_k'$ incident) belongs to $\cW (d')$ i.e., as long as $b_j$ and $b_k$ are at the leaves of $WST$ corresponding to $\cW (d')$, we continue to have $br_p$ as a dirty bridge at node $p$ of $WST$.

\begin{definition}
A bridge is considered as a {\bf dirty bridge} whenever it is defined at a $WST$ node, but one or both of its endpoints do not belong to the shortest path wavefront. 
If the bridge is dirty at a node $p$ of $WST$, then the corresponding upper hull at $p$ is termed as a {\it dirty upper hull}. 
\end{definition}

When bridge $br_p$ is determined, say at $\cW (d)$,  at a node $p$ of $WST$, let $p_j' \in w(v_j)$ and $p_k \in w(v_k)$ be its endpoints such that $p_j'$ is  not a point of tangency of $w(v_j)$. 
Let $C$ be a circle with the same center and radius as $w(v_j)$.
To facilitate in reconstructing bridge at $\cW (d')$, by utilizing property \ref{prop:wstbridgemovement}, we find a point of tangency $p_j$ on $C$ (closest along $C$ to $p_j'$) so that $p_jp_k$ is a tangent to both $w(v_j)$ and $w(v_k)$. 
Again, so formed bridge $br_p = (p_j, p_k)$ is a dirty bridge.

\begin{lemma}
\label{lem:dirtybricorr} The dirty bridges in a $WST$ do not affect the correctness of shortest distance computations between $WST$ and its associated corridor convex chain or enter/exit boundary $C$.
\end{lemma}

\begin{figure}
\centerline{\epsfysize=220pt \epsfbox{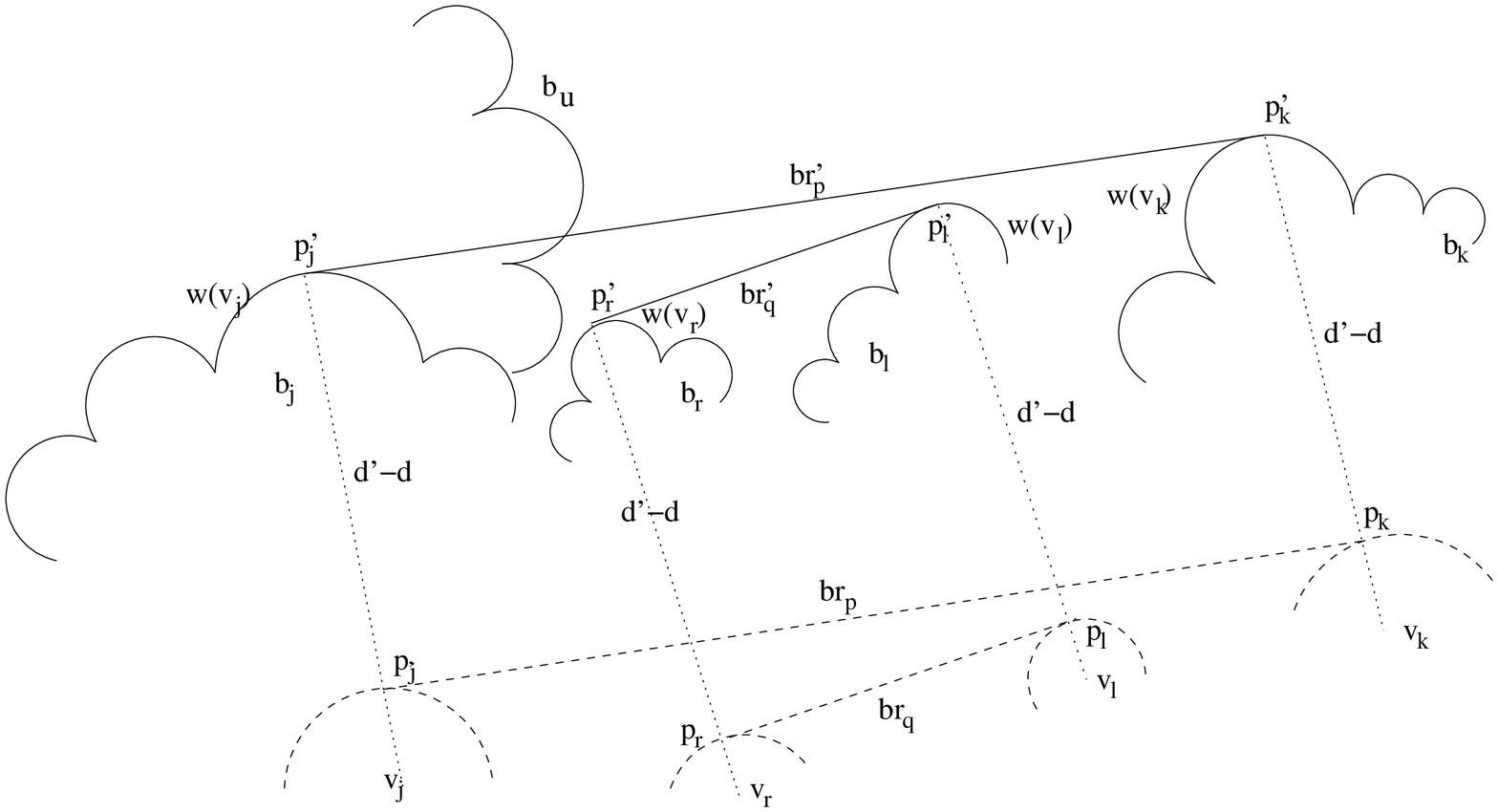}}
\caption{\label{fig:dirtybridges2} Dirty bridges versus the correctness of shortest distance computations}
\end{figure}

\begin{proof}
Consider a node $p$ of $WST$ with the bunches at its descendant leaf nodes as $b_i, b_{i+1}, \ldots, b_j, \\ \ldots, b_k, \ldots, b_{w-1}, b_w$.  
See Fig. \ref{fig:dirtybridges2}.
Let $B=\bigcup_{l \in \{i,\ldots,w\}} b_l$.
Let a dirty bridge $br_p'$ at node $p$ of $WST$ join two points $p_j' \in b_j$ and $p_k' \in b_k$.
Suppose the point $p_j$ was traversed by a bunch $b_u$, where $u \notin \{i, i+1, \ldots, w\}$, when $\cW (d')$ .
Let $ch_p'$ be the dirty upper hull at node $p$.
Consider bridge $br_q'$ at node $p$ computed at this stage of wavefront progression.  
Since we know that at least two bunches $b_j$ and $b_k$ exist among bunches in $B$, the bridge $br_q'$ is guaranteed to exist.  
Let the bridge $br_q'$ joins two points $p_r'$ and $p_l'$, where the point $p_r'$ is located on the wavefront segment $w(v_r) \in b_r$ whose center is $v_r$ and the point $p_l'$ is located on the wavefront segment $w(v_l) \in b_l$ whose center is $v_l$.  
Note that $w(v_r)$ (resp. $b_r, p_l', w(v_l), b_l$) is not necessarily distinct from $w(v_j)$ (resp. $b_j, p_k', w(v_k), b_k$).
\hfil\break

Consider a wavefront progression $\cW (d)$ at which the bridge at node $p$ was not dirty.
Let $br_p$ be the bridge at node $p$ for $\cW (d)$ with $p_j, p_k$ as its endpoints.
Let the corresponding upper hull be $ch_p$.
Choose a point $p_r \in w(v_r)$ at that stage of wavefront expansion so that the line segment $v_rp_r$ extended yields $p_r'$.
Then the Euclidean length of line segment $p_rp_r'$ be $d'-d$.
>From the property \ref{prop:wstbridgemovement}, for an expansion of $d'-d$ along $v_rp_r'$, the Euclidean length of segments $p_jp_j'$, $p_lp_l'$ would be $d'-d$ too.  
This with the fact that the bunches remaining in $B$ after a Type-IV event are in the same order along $B$ as they were before assures the following: $p_r' \in ch_p'$ and $p_l' \in ch_p'$. 
The convexity of $ch_p'$ guarantees that the line segment $br_q'$ (joining points $p_r', p_l'$) belongs to $ch_p'$.
Hence the shortest distance from $C$ to $br_p'$ is less than (or equal to) the shortest distance from $C$ to $br_q'$.  
Similar arguments can be given when both the endpoints of a dirty bridge were traversed.
\hfil\break

The correctness of the shortest distance computation procedure relies on the property shown above for a dirty bridge, when the shortest distance from $C$ is to that dirty bridge.  
Suppose the shortest distance between $br_p'$ and $C$ is less than the shortest distance to either of the bridges at $p$'s immediate descendants.  
Since the shortest distance from $C$ to bridge $br_p'$ is less than (or equal to) the shortest distance from $C$ to bridge $br_q'$, the bridge $br_p'$ will be struck whenever the wavefront can strike the bridge $br_q'$.  
Therefore, the correctness of shortest distance computation procedure is not affected.
\end{proof}

\begin{lemma}
\label{lem:invalidbridges}
The bridges constructed over invalid segments of a $BHT$ do not affect the correctness of shortest distance computations.
\end{lemma}
\begin{proof}
The argument is similar to the last paragraph of Lemma \ref{lem:dirtybricorr}.
\end{proof}

\begin{lemma}
\label{lem:typeIVnumevents}
The total number of Type-IV events are $O(m)$.
\end{lemma}
\begin{proof}
The Type-IV events are caused by those bunches which passed over the exit boundary of the corridor in which they were initiated.
Since there are $O(m)$ such bunches (Lemma \ref{lem:numsplits}), there can be $O(m)$ Type-IV events.
\end{proof}

\begin{lemma}
\label{lem:typeIVhandlingcost}
The total cost in determining and handling all the Type-IV events is $O(m(\lg{m})(\lg{n}))$.
\end{lemma}
\begin{proof}
Initially, in all $WST$s together, there are $O(m)$ bunches.
The insertion or deletion of a bunch $B$ to a $WST$ affect only the Type-IV events along the path from the leaf node at which $B$ is inserted to the root.
Since the depth of a $WST$ is $O(m)$, computing the shortest distance between two upper hulls in $WST$ is $O(\lg{n})$, the amortized time to compute Type-IV events and update the min-heap takes $O((\lg{m})(\lg{n}))$ time.
Only the bunches that traversed the exit boundary of the corridor in which they are initiated would require these updates, which are $O(m)$ in number. 
\hfil\break

As mentioned in Lemma \ref{lem:hullinter}, each Type-IV causes at least one bunch to be removed from the given $WST$.
>From Lemma \ref{lem:wsttimeandspace}, updating $WST$ corresponding to this deletion takes $O((\lg{m})(\lg{n}))$ amortized time.
Also, from Lemma \ref{lem:sdcwt}, updating shortest distances in the min-heap takes $O((\lg{m})(\lg{n})$ amortized time.
Since there are at most $O(m)$ Type-IV events, total computation involved in updating $WST$ takes $O(m(\lg{m})(\lg{n}))$.
\hfil\break

The amortized time complexity involved in locating the destination $t$ using the procedure from \cite{Vaidya86} takes $O((\lg{m})(\lg{n}))$ time (Analysis is similar to Lemma \ref{lem:tinsplits}.).
The updates to $WST$ and the shortest distance computations due to a Type-IV event take $O((\lg{m})(\lg{n}))$ time.
The deletion/updating of events in the event queue take $O(\lg{n})$ amortized complexity.
\end{proof}

\section{More Analysis}
\label{sect:moreanalysis}

\begin{theorem}
\label{thm:numevents} The total number of event points are $O(m)$.
\end{theorem}
\begin{proof}
>From Lemmas \ref{lem:typeItypeIInumevents}, \ref{lem:typeIIInumevents}, and \ref{lem:typeIVnumevents}.
\end{proof}

\begin{theorem}
\label{thm:numcwtcbt} The number of waveform- and boundary-sections created/updated are $O(m)$.
\end{theorem}
\begin{proof}
Since there are at most $O(m)$ bunches at any point of execution of the algorithm (Lemma \ref{lem:numbunches}) and a bunch is part of at most one boundary-section, initially there are at most $O(m)$ boundary-sections.  
Since there are $O(m)$ corridor convex chains/exit boundaries and each can have at most one waveform-section associated to it, initially there can be at most $O(m)$ waveform-sections.  
We need to update these initial boundary- and waveform-sections due to events, which are upper bounded by $O(m)$ (Theorem \ref{thm:numevents}).  
Hence the complexity.
\end{proof}

\begin{theorem}
\label{thm:mergecompl} The total processing involved in the merge procedure, excluding the computation involved in finding Type-III events, during the entire algorithm is of $O(m(\lg{m})(\lg{n}))$ time complexity.
\end{theorem}
\begin{proof}
We need to invoke MERGE procedure constant number of times at most per each merge operation.  
The MERGE procedure complexity relies on ASSOCATOB procedure.  
Primarily, the complexity of the ASSOCATOB procedure relies on two components: how many times we call the binary search procedure within it, and on the total number of $WST, BST$ updates.
\hfil\break

In the procedure, the binary search over the I-curves of $B$ determines the next I-curve which does not intersect the edge $c_{prev}$.  
The $RV$ of a bunch may be $\phi$ before a merge, or it can become $\phi$ after calling the ASSOCATOB procedure.  
However, the latter is possible only when the merge occurs and the number of merges are bounded by $O(m)$ from Lemma \ref{lem:nummerges}.  
>From Theorem \ref{thm:numcwtcbt}, there can be at most $O(m)$ waveform-sections during the entire algorithm, and hence there can be at most $O(m)$ bunches with $\phi$ $RV$s.
The $RV$ of a bunch can change from $\phi$ to non-$\phi$ because of splits, however, the splits are bounded by $O(m)$ (from Lemma \ref{lem:numsplits}); 
hence there are $O(m)$ bunches with non-$\phi$ $RV$s over the course of the algorithm.
Since each binary search over the bunches finds a bunch whose $RV$ is non-$\phi$, we charge the $O(\lg{m})$ binary search complexity to bunches with non-$\phi$ $RV$s, leading to $O(m \lg{m})$ complexity.  
The IIntersect procedure takes $O(\lg{m})$ time per invocation.
Also, determining the proximity of a corridor convex chain or enter/exit boundary $c_s$ (line 4 of ASSOCATOB procedure) w.r.t. a bunch or $WST$ in $A$ versus a bunch or $WST$ in $B$ takes $O((\lg{m})(\lg{n}))$.
Hence, the overall time is $O(m(\lg{m})(\lg{n}))$.
\hfil\break

Since there are $O(m)$ waveform and boundary-sections possible in the entire algorithm (Theorem \ref{thm:numcwtcbt}), and, each $WST/BST$ update takes $O((\lg{m})(\lg{n}))$ (using \cite{Overmars81}), 
we spend $O(m(\lg{m})(\lg{n}))$ to dynamically maintain waveform- and boundary-sections during all the invocations of ASSOCATOB procedure together.
\hfil\break

Also, assigning all the corridor convex chains or enter/exit boundaries prior to the intersection of the first inter-bunch I-curve($b_1,b_2$) of $B$ with the $RV(A)$, and, adjusting the associations of already existing corridor convex chains/exit boundaries in $RV(b_1)$ takes $O((\lg{m})(\lg{n}))$ time. 
Hence we spend $O(m(\lg{m})(\lg{n}))$ time in ASSOCATOB procedure during the entire algorithm.
Therefore, MERGE procedure, and, all merge operations together are of $O(m(\lg{m})(\lg{n}))$ time complexity.
\end{proof}

\begin{lemma}
\label{lem:typeIIIhandlingcost} Building and maintaining all $BHT$s during the entire algorithm takes $O(m \lg{n}+n)$.
\end{lemma}
\begin{proof}
The $BHT$ maintenance is divided into four cases (See Section \ref{subsect:bhtdatastr}.).  
In Case (1), we create nodes corresponding to the bunch vertices all at once in $O(n)$ time, even though some of them are invalid by the time we construct $BHT$.  
In Cases (2) and (3), we do nothing.
In Case (4), splitting and balancing $T_{new}$ takes $O(\lg{n})$ time.
Since there are $O(m)$ waveform- and boundary-sections (Theorem \ref{thm:numcwtcbt}) together, Case (4) takes $O(m \lg{n})$ time in the worst-case.  
\end{proof}

\begin{theorem}
\label{thm:eventdu} Any event point determination with the associated updates, excluding the initiation and maintenance of $BHT$s, can be accomplished in $O((\lg{m})(\lg{n}))$ amortized time. 
\end{theorem}
\begin{proof}
Immediate from Lemmas \ref{lem:typeItypeIIhandlingcost}, \ref{lem:typeIIIhandlingcost}, and \ref{lem:typeIVhandlingcost}.
\end{proof}

\begin{theorem}
\label{thm:finalth} The shortest distance from $s$ to $t$ can be computed in $O(T+m(\lg{m})(\lg{n}))$ time using $O(n)$ space.
\end{theorem}
\vspace{-.15in}
\begin{proof}
The triangulation of polygonal region takes $O(n + m\lg^{1+\epsilon}m)$, represented as $O(T)$.  
Given a triangulation, finding useful corridors and junctions takes $O(m\lg{n})$.  
>From theorems \ref{thm:numevents}, \ref{thm:eventdu} all the event points determination and associated updates, excluding the updates to $BHT$s, can be done in $O(m(\lg{m})(\lg{n}))$ time.  
>From Lemma \ref{lem:typeIIIhandlingcost}, initiating and maintaining all $BHT$s take $O(m\lg{n}+n)$.
Hence the algorithm is of $O(T+m(\lg{m})(\lg{n}))$ time complexity.
\hfil\break

All the four data structures $BHT$, $BST$, $WST$ and Event Heap require at most $O(n)$ space at any instance during the entire algorithm.  
Hence, the algorithm is of $O(n)$ space complexity.
\end{proof}

\section{Conclusions}
\label{sect:conclu}

We have described an algorithm for finding the Euclidean shortest path in polygonal domain with $O(T+m(\lg{m})(\lg{n}))$ time complexity using $O(n)$ space.
It would be of interest to investigate for a solution with $O(n+m\lg{m})$ time and $O(n)$ space.
Also, exploring the applicability of the above technique to weighted geodesic shortest path computation, and determining approximate Euclidean shortest paths in polygonal domains is of interest.

\bibliography{spcorrwav-rinkulu10}

\begin{thebibliography}{10}

\bibitem{Chaz94}
Reuven Bar-Yehuda and Bernard Chazelle.
\newblock {Triangulating Disjoint Jordan Chains}.
\newblock {\em International Journal of Computational Geometry \&
  Applications}, 4(4):475--481, 1994.

\bibitem{Prep85}
{Franco P. Preparata and Michael Ian Shamos}.
\newblock {\em Computational Geometry: An Introduction}.
\newblock Springer-Verlag, 1985.

\bibitem{Ghosh91}
Subir~Kumar Ghosh and David~M. Mount.
\newblock {An Output-Sensitive Algorithm for Computing Visibility Graphs}.
\newblock {\em SIAM Journal on Computing}, 20(5):888--910, 1991.

\bibitem{Hersh93}
John Hershberger and Subhash Suri.
\newblock {Efficient computation of Euclidean shortest paths in the plane}.
\newblock In {\em Proceedings of the 1993 IEEE 34th Annual Foundations of
  Computer Science}, pages 508--517, 1993.

\bibitem{Hersh97}
{John Hershberger and Subhash Suri}.
\newblock {An Optimal Algorithm for Euclidean Shortest Paths in the Plane}.
\newblock {\em SIAM Journal on Computing}, 28:2215--2256, 1997.

\bibitem{Kapoor98}
Sanjiv Kapoor.
\newblock {An efficient wavefront method}.
\newblock Technical report, Indian Institute of Technology Delhi, India, 1998.

\bibitem{Kapoor99}
Sanjiv Kapoor.
\newblock {Efficient Computation of Geodesic Shortest Paths}.
\newblock In {\em {Proceedings of the 32nd annual symposium on Theory of
  Computing}}, pages 770--779, 1999.

\bibitem{Kapoor88}
Sanjiv Kapoor and S.~N. Maheshwari.
\newblock {Efficient algorithms for Euclidean shortest path and visibility
  problems with polygonal obstacles}.
\newblock In {\em Proceedings of the fourth annual symposium on Computational
  Geometry}, pages 172--182, 1988.

\bibitem{Kapoor00}
Sanjiv Kapoor and S.~N. Maheshwari.
\newblock {Efficiently Constructing the Visibility Graph of a Simple Polygon
  with Obstacles}.
\newblock {\em SIAM Jounrnal on Computing}, 30(3):847--871, 2000.

\bibitem{Kapoor97b}
Sanjiv Kapoor, S.~N. Maheshwari, and Joseph S.~B. Mitchell.
\newblock {An Efficient Algorithm for Euclidean Shortest Paths among Polygonal
  Obstacles in the Plane}.
\newblock {\em Discrete \& Computational Geometry}, 18(4):377--383, 1997.

\bibitem{Mitchell88}
Joseph S.~B. Mitchell.
\newblock {On maximum flows in polyhedral domains}.
\newblock In {\em Proceedings of the fourth annual symposium on Computational
  Geometry}, pages 341--351, 1988.

\bibitem{Mitchell93}
Joseph S.~B. Mitchell.
\newblock {Shortest paths among obstacles in the plane}.
\newblock In {\em Proceedings of the ninth annual symposium on Computational
  geometry}, pages 308--317, 1993.

\bibitem{Mitchell00}
Joseph~S.B. Mitchell.
\newblock {Geometric Shortest Paths and Network Optimization}.
\newblock In {\em Handbook of Computational Geometry}, pages 633--701, 1998.

\bibitem{Overmars81}
Mark~H. Overmars and Jan van Leeuwen.
\newblock Maintenance of configurations in the plane.
\newblock {\em Journal of Computer and System Sciences}, 23(2):166 -- 204,
  1981.

\bibitem{Reif94}
James~A. Storer and John~H. Reif.
\newblock {Shortest paths in the plane with polygonal obstacles}.
\newblock {\em Journal of ACM}, 41(5):982--1012, 1994.

\bibitem{Vaidya86}
Pravin~M. Vaidya.
\newblock {An optimal algorithm for the all-nearest-neighbors problem}.
\newblock In {\em Proceedings of the 27th Annual Symposium on Foundations of
  Computer Science}, pages 117--122, 1986.

\bibitem{Welzl85}
Emo Welzl.
\newblock {Constructing the visibility graph for $n$-line segments in O($n^2$)
  time}.
\newblock {\em Information Processing Letters}, 20(4):167--171, 1985.

\end{thebibliography}

\pagebreak

\end{document}